\documentclass{llncs}

\usepackage[letterpaper,top=2cm,bottom=2cm,left=3cm,right=3cm,marginparwidth=1.75cm]{geometry}

\usepackage{amsmath}
\usepackage{amsfonts}
\usepackage{graphicx}
\usepackage[colorlinks=true, allcolors=blue]{hyperref}
\usepackage[english]{babel}
\usepackage[capitalise,nameinlink]{cleveref}
\usepackage{xspace}
\usepackage{mathtools}
\usepackage{braket}
\usepackage{multirow}
\usepackage{xcolor}
\usepackage{dsfont}
\usepackage{authblk}
\usepackage{verbatim}

\newcommand{\lrw}{\textsf{LRW}\xspace}
\newcommand{\xex}{\textsf{XEX}\xspace}
\newcommand{\xext}{\textsf{XEX2}\xspace}
\newcommand{\idhashxext}{\ensuremath{\mathsf{XEX2}^{\mathrm{ideal-hash}}}\xspace}
\newcommand{\xts}{\textsf{XTS}\xspace}
\newcommand{\emr}{\textsf{Even-Mansour}\xspace}
\newcommand{\fx}{\textsf{FX}\xspace}
\newcommand{\cbc}{\textsf{CBC}\xspace}
\newcommand{\ecbc}{\textsf{ECBC}\xspace}
\newcommand{\cmac}{\textsf{CMAC}\xspace}
\newcommand{\gcm}{\textsf{GCM}\xspace}

\newcommand{\prp}{\textsf{PRP}\xspace}
\newcommand{\sprp}{\textsf{SPRP}\xspace}
\newcommand{\qprp}{\textsf{PRP-PQ}\xspace}
\newcommand{\qsprp}{\textsf{SPRP-PQ}\xspace}
\newcommand{\prf}{\textsf{PRF}\xspace}

\newcommand{\mac}{\textsf{MAC}\xspace}

\newcommand{\forgery}{\textsf{Forgery}\xspace}

\newcommand{\Oinv}{O^{\mathrm{inv}}}

\newcommand{\A}{\mathcal{A}}
\newcommand{\D}{\mathcal{D}}
\newcommand{\B}{\mathcal{B}}
\newcommand{\Or}{\mathcal{O}}

\def\bool{\{0,1\}}

\renewcommand{\epsilon}[0]{\varepsilon}

\newcommand{\Hyb}{\operatorname{\mathbf{H}}}

\newcommand{\tweak}{\mathcal{T}}

\newcommand{\swap}[2]{\mathsf{swap}_{#1,\,#2}}

\newcommand{\expref}[2]{\texorpdfstring{\hyperref[#2]{#1~\ref{#2}}}{#1~\ref{#2}}}

\newcommand{\formerqC}{q_C}
\newcommand{\from}{\leftarrow}
\newcommand\algo{\mathcal}
\newcommand{\permset}[1]{{\cal P}(n)}
\newcommand{\proj}[1]{\ensuremath{|#1\rangle \langle #1|}}
\def\xor{\oplus}
\DeclareMathOperator*{\Exp}{\mathbb{E}}
\definecolor{mgreen}{rgb}{.1,.7,0}
\def\S{{\sf Swap}}

\newcommand{\bit}{\{0,1\}}

\usepackage{ifthen}
\newif\ifdraft

\drafttrue 

\ifdraft
\newcommand{\ks}[1]{{\color{cyan}[{\bf ks:} #1]}}
\newcommand{\ga}[1]{{\color{red}[{\bf GA:} #1]}}
\newcommand{\chen}[1]{{\color{orange}[{\bf Chen:} #1]}}
\newcommand{\cm}[1]{{\noindent \textcolor{mgreen}{\emph{(CM:  #1)}}}{}}
\else
\newcommand{\ga}[1]{}
\newcommand{\cm}[1]{}
\newcommand{\chen}[1]{}
\newcommand{\ks}[1]{}
\fi

\title{Post-Quantum Security of Block Cipher Constructions}
\author{\vspace{-.8cm}}
\institute{\vspace{-.8cm}}

\author{Gorjan Alagic\inst{1} \and Chen Bai\inst{2} \and Christian Majenz\inst{3} \and Kaiyan Shi\inst{4}}

\institute{\small QuICS, University of Maryland/NIST \\ \email{galagic.umd.edu}
\and Dept.\ of Computer Science, Virginia Tech \\ \email{cbai1@vt.edu} \and Dept.\ of Applied Mathematics and Computer Science, Technical University of Denmark \\ \email{chmaj@dtu.dk} \and Dept.\ of Computer Science, University of Maryland \\ \email{kshi12@umd.edu}}
\date{}

\begin{document}
{\def\addcontentsline#1#2#3{}\maketitle}
\begin{abstract}
 Block ciphers are versatile cryptographic ingredients that are used in a wide range of applications ranging from secure Internet communications to disk encryption. While post-quantum security of public-key cryptography has received significant attention, the case of symmetric-key cryptography (and block ciphers in particular) remains a largely unexplored topic. In this work, we set the foundations for a theory of post-quantum security for block ciphers and associated constructions. Leveraging our new techniques, we provide the first post-quantum security proofs for the key-length extension scheme FX, the tweakable block ciphers LRW and XEX, and most block cipher encryption and authentication modes. Our techniques can be used for security proofs in both the plain model and the quantum ideal cipher model. Our work takes significant initial steps in establishing a rigorous understanding of the post-quantum security of practical symmetric-key cryptography.   
\end{abstract}

\setcounter{tocdepth}{3}
\tableofcontents
\newpage




\section{Introduction}

\subsection{Background}

\paragraph{Block ciphers.}
A block cipher is a keyed family of efficiently-implementable permutations of $\bit^n$. A block cipher $\pi$ is secure if, for a uniformly random key $k$, the permutation $\pi_k$ is indistinguishable from random to adversaries that can make forward and inverse queries to $\pi_k$. The most well-known example of a block cipher is the Advanced Encryption Standard (AES) \cite{dworkin2001advanced}, which is a fundamental ingredient in secure Internet communications and much more.

Block ciphers are ubiquitous in real-world cryptography. They are versatile ingredients that can be adapted to a variety of cryptographic applications. Their security can be increased via simple key-length extension schemes like FX \cite{C:KilRog96,JC:KilRog01}. They can also be expanded into an exponentially-large family of ciphers via tweakable constructions like \lrw and \xext~\cite{liskov2002tweakable,rogaway2013evaluation}. Finally, by means of \emph{block cipher modes}, block ciphers can be used to construct a multitude of symmetric-key encryption and/or authentication schemes, with a wide range of options for both security and performance (see, e.g.,~\cite{dworkin2001recommendation,dworkin2007sp,dworkin2010recommendation,gueron2017aes}). In applications, block cipher modes are ubiquitous in secure communications, while tweakable ciphers are commonly used in disk encryption schemes.

\paragraph{Post-quantum security.}
The impact of large-scale quantum computation on the security of block ciphers and associated cryptographic constructions is not fully understood.
While Shor's algorithm~\cite{shor1994algorithms} spurred intense scrutiny of the impact of quantum computers on public-key cryptography, post-quantum security of symmetric-key cryptography has received scant attention. While there are no known dramatic quantum attacks against currently-deployed symmetric-key schemes, there are important reasons to understand this setting in general (and the block cipher case in particular), as we now briefly explain. 

First, it is \emph{in principle} possible that some symmetric-key schemes in use today are not post-quantum secure\footnote{One can easily construct relatively natural examples of schemes that are secure against classical computers but not against quantum computers. However, we doubt that this holds for any currently-deployed schemes.}. Second, many symmetric-key constructions rely crucially on security proofs, and only a handful of such proofs for post-quantum security are known (mainly for Even-Mansour-type constructions~\cite{EC:ABKM22,EC:ABKMS24,bai2024quantum} and the Ascon scheme~\cite{hosoyamada2025post}). Finally, fine-grained bounds in proofs often inform parameter selection when deploying symmetric-key schemes, and such bounds are essentially unavailable in the post-quantum case. Instead, practitioners typically assume that one can simply double the key length to avoid Grover search. However, the few results we do have indicate that this is sometimes overkill~\cite{EC:ABKM22} and could in principle even be insufficient~\cite{EC:BonSchSib22}.

\subsection{This work}

In this work, we set the foundations for a theory of post-quantum security for block ciphers and associated constructions. We then apply our techniques to give the first post-quantum security proofs for a variety of schemes, including key-length extension schemes, tweakable block ciphers, and block cipher modes. Our work takes significant initial steps in establishing a rigorous understanding of the post-quantum security of practical symmetric-key cryptography.

Our techniques can be used for security proofs both in the \emph{plain model} (i.e., without any special assumptions or idealizations) as well as the \emph{ideal cipher model} (ICM), a commonly used model in relevant classical proofs. In the ICM, one assumes that a \emph{perfect} block cipher $\pi$ was sampled (i.e., $\pi_k$ is a uniformly random permutation for all $k$), and all parties have oracle access to it. In the traditional classical setting, this oracle access allows anyone to make queries of the form
\begin{equation}
    (x, k) \mapsto \pi_k(x)
    \qquad \text{and} \qquad
    (x, k) \mapsto \pi_k^{-1}(x)\,.
\end{equation}
In a post-quantum setting, however, adversaries in possession of a quantum computer will be able to execute the circuits of \textsf{AES} in superposition. It is thus natural to also grant quantum adversaries such access in the ideal cipher setting, via black-box unitaries
\begin{align}
    \ket{k}\ket{x}\ket{y} \mapsto \ket{k}\ket{x}\ket{y \oplus \pi_k(x)} 
    \qquad \text{and} \qquad
    \ket{k}\ket{x}\ket{y} \mapsto \ket{k}\ket{x}\ket{y \oplus \pi^{-1}_k(x)}\,.
\end{align}
The resulting model is called the quantum ideal cipher model (or QICM)\footnote{The ICM (resp., QICM) is the natural block cipher analogue of the well-known Random Oracle Model (ROM, resp., QROM), which is often used to model hash functions in security proofs.} \cite{cojocaru2025quantum,sato2019so}. Proving results in the QICM has turned out to be quite difficult. Classical techniques are based on transcripts and do not translate to this setting. While Zhandry's compressed oracle~\cite{zhandry2019record} does offer a somewhat analogous quantum lower bound method, it only applies to \emph{random function oracles}, and extending it to permutations and ideal ciphers has proved challenging. Moreover, the post-quantum setting involves mixed query types: the ideal cipher is queried both \emph{classically} (through the construction) and \emph{quantumly} (directly by the adversary). This is the right post-quantum model: the construction is always a classical cipher implemented by the honest party on a classical computer using knowledge of the actual secret keys, while the block cipher is a public algorithm. Moreover, correct modeling of the construction oracle makes a crucial difference in security, as many commonly-used ciphers are \emph{insecure} if the construction can be queried quantumly~\cite{C:KLLN16}, but \emph{secure} if the construction can only be queried classically~\cite{EC:ABKM22}. 

In this paper, we only consider the model in which the ideal cipher is quantumly-accessible and any constructions involving the secret key are only classically-accessible. We refer to this model as the post-quantum model. In previous literature, this model was sometimes referred to as the \textsf{Q1} model~\cite{bonnetain2019quantum,hosoyamada2018cryptanalysis,jaeger2021quantum,C:KLLN16}, in contrast with the (unrealistic) \textsf{Q2} model where all oracles can be queried quantumly.

Our results are summarized as follows. We give a more technical summary further below, including the precise security bounds we establish. 
\begin{enumerate}
    \item \textbf{QICM resampling.} We give a general \emph{resampling lemma} that can be used to ascertain the ability of a quantum adversary to detect modifications to the ideal cipher (e.g., in a simulation as part of a proof). Previous resampling lemmas only held for random functions and permutations~\cite{EC:ABKM22,EC:ABKMS24,grilo2021tight,hosoyamada2025post}. This new tool can be combined with an adaptation of a proof approach of~\cite{EC:ABKM22} to yield a powerful technique for QICM security proofs. We apply this technique to establish further results below.
    \item \textbf{FX construction.} We prove post-quantum security of the FX key-length extension scheme in the QICM. This problem has evaded analysis for a number of years (a 2021 result only held for non-adaptive adversaries~\cite{jaeger2021quantum}). The bound we give for the number of queries needed to achieve a constant distinguishing probability is tight. A straightforward application of our result establishs post-quantum security for the lightweight cipher PRINCE~\cite{borghoff2012prince}. 
    \item \textbf{Tweakable ciphers.} We prove security of two widely-used tweakable block ciphers (LRW and XEX2) in both the plain model and in the QICM (with different bounds). We note that XEX2 is the basis of the XTS-AES disk encryption scheme used by most operating systems. 
    \item \textbf{Block cipher modes.} Finally, we observe that the security proofs of most block cipher modes (including all modes used in secure Internet traffic) translate easily to the post-quantum setting. Moreover, one can translate classical security bounds to post-quantum ones simply by replacing the appropriate strong pseudorandomness advantage term with its post-quantum analogue. 
\end{enumerate}

\subsection{Technical summary of results}\label{sec:tech-sum}

A technical summary of our main results is as follows.

\paragraph{\textbf{\emph{A proof technique for the QICM.}}} The first proof of post-quantum security of a block cipher construction was for the plain Even-Mansour construction~\cite{EC:ABKM22}. Since then, the technique of~\cite{EC:ABKM22} has been extended to show security of tweakable Even-Mansour~\cite{EC:ABKMS24}, the Ascon lightweight cipher~\cite{hosoyamada2025post}, key-alternating ciphers~\cite{bai2024quantum,basak2025post}. The technique has even been adapted to give certain proofs in the quantum Haar-random oracle model~\cite{hhan2024pseudorandom}. 

At a high level, the technique of~\cite{EC:ABKM22} can be described as follows. The goal is to show indistinguishability between (i.) a pair of ``real'' oracles $(E_{(k)},\ket{E})$ and (ii.) an ``ideal'' uncorrelated pair $(R,\ket{E})$. In both worlds, the first oracle is classical and the second oracle is quantum. Starting with a real-world $q$-query sequence
\begin{equation}
    \ket{E}, E_{(k)}, \ket{E}, E_{(k)}, \dots, \ket{E}, E_{(k)}
\end{equation}
we switch the classical queries to the ideal case, one-by-one. A naive $j$-th hybrid would then be
\begin{equation}
    \underbrace{\ket{E}, R, \dots, \ket{E}, R}_j, \underbrace{\ket{E}, E_{(k)}, \dots, \ket{E}, E_{(k)}}_{q-j}\,.
\end{equation} 
Unfortunately, this does not work: showing indistinguishability of hybrids $j-1$ and $j$ would require first showing that the adversary cannot extract the key $k$ during the last $q-j$ queries; indeed, with knowledge of $k$ one easily notices the switch in past queries. This leads to a circular argument.

Instead, when switching classical queries to the ideal case, we also modify $E$ (for all remaining queries) in a small number of locations, and then reset these modifications at a later stage:
\begin{equation}
    \underbrace{\ket{E}, R, \dots, \ket{E}, R}_j, \underbrace{\ket{E^{[j]}}, E^{[j]}_{(k)}, \dots, \ket{E^{[j]}}, E^{[j]}_{(k)}}_{q-j}\,.
\end{equation}
where $E^{[j]}$ denotes the quantum oracle with roughly $j$ modified locations. Making such modifications necessitates the use of a particular technical ingredient: a \emph{resampling lemma}\footnote{One also needs a \emph{reprogramming lemma} (to reset these modifications at a later stage) but this is a relatively straightforward quantum search lower bound~\cite{EC:ABKM22}.}. Such a lemma states that an adversary cannot detect such modifications unless it has made a very large number of quantum queries to $E$ prior to the modification being made. 

\paragraph{Ideal cipher resampling.}
In this work, we develop the above technique in the case where $E$ is an ideal cipher. This enables us to give the first full post-quantum security proofs for a variety of block cipher constructions. The key technical advance is a resampling lemma for the ideal cipher model. Consider the advantage of a computationally unbounded distinguisher $\algo D$ in the following three-phase experiment.

\begin{enumerate}
    \item An ideal cipher $E: \bit^m \times \bit^n \rightarrow \bit^n$ is sampled, and $\algo D$ gets (two-way) quantum-query access to $E$. Then $\algo D$ (adaptively) selects and outputs a description of a distribution $M$.
    \item A sample $(k, s_0, s_1)$ is drawn from $M$, and  $E_k(s_0)$ and $E_k(s_1)$ are swapped ($b=1$) or not ($b=0$).
    \item Now $\algo D$ continues with access to (the possibly modified) $E$, and must eventually output a guess $b'$ for the value of $b$.
\end{enumerate}
We show that the guessing advantage of $\algo D$ is at most $4\sqrt{q\cdot 2^{n-h}}$ where $h$ is the min-entropy of $M$ and $q$ is the number of queries made by $\algo D$ \emph{in phase 1}. Note that the bound does not depend on the number of queries made in phase 3.

The full technical statement of this resampling lemma is given in \Cref{sec:resampling_ic}, with proofs in \Cref{app:lemma-proofs}. Using this lemma and the hybrid technique described above, we are able to show the security of several constructions, discussed below.

\paragraph{\textbf{\emph{FX construction.}}} The \fx construction~\cite{C:KilRog96,JC:KilRog01} transforms a block cipher $E: \{0,1\}^m \times \{0,1\}^n \rightarrow \{0,1\}^n$ with key length $m$ into another block cipher with the same block size but a longer key, as follows.
\begin{equation} \label{eqn:fx}
    \fx^{E}_{k_0, k_1, k_2}(x) = E_{k_0}(x \oplus k_1) \oplus k_2
\end{equation}
Here $k_0 \leftarrow \{0,1\}^m$ and the marginal distributions of $k_1, k_2 \in \{0,1\}^n$ are uniform. For post-quantum security in the QICM, the relevant quantity is the distinguishing advantage
\begin{equation}
\textsf{Adv}_\A^{\fx} := \left| \Pr_{\substack{E, k_0,k_1,k_2}} \left[ \A^{\fx[E],\ket{E}} = 1\right]  - \Pr_{\substack{R,E}} \left[ \A^{R,\ket{E}} = 1\right] \right|
\end{equation}
of an (unbounded) adversary $\A$ in distinguishing $\fx$ from a random permutation $R$ using $q_C$ classical queries, while making $q_Q$ quantum queries to $E$. Both forward and inverse queries are allowed. In \expref{Section}{sec:fx}, we show that
\begin{align}
    &\textsf{Adv}_\A^{\fx} \leq 4(q_C\sqrt{q_Q} + q_Q\sqrt{q_C}) \cdot 2^{-(m+n)/2}
    &&\text{for } q_C \ll 2^n\\
    &\textsf{Adv}_\A^{\fx} \leq q_Q^2 \cdot 2^{-m}
    &&\text{for } q_C \approx 2^n
\end{align}
For calculating the number of queries required for constant success probability, our bounds are tight. For small $q_C$ as well as $q_C \approx 2^n$ (i.e., the full codebook case), Grover search for remaining keys matches the bound. For $q_C \approx 2^{n/3}$, BHT collision-finding \cite{kuwakado2012security} also matches our bound. 

In contrast, in the classical ideal cipher model, the \fx construction is tightly secure (i.e., indistinguishable from an independent uniformly random permutation) with distinguishing advantage~\cite{JC:KilRog01}
\begin{equation}\label{eq:fx-classical}
       \textsf{Adv}^{\fx}_{\A} \leq \frac{q_{\fx}\cdot q_E}{2^{m+n}}.
\end{equation}
against an adversary making $q_{\fx}$ queries to $\fx_K$ and $q_E$ queries to the ideal cipher $E$.

\paragraph{Applications to lightweight ciphers.}
A relatively straightforward application of our lower bound for FX establishes post-quantum security of the \textsf{PRINCE}~\cite{borghoff2012prince} cipher in the QICM. The security of \textsf{PRINCE} is derived from the security of a primitive $\widetilde{\fx}$ which is closely related to standard \fx. Specifically, let $H$ be an $(m-1)$-dimensional subspace of $\mathbb{F}_2$, choose nonzero $\alpha \in \mathbb{F}_2^m$ not in $H$, and define
\begin{align*}
    \widetilde{\fx}_{k_0,k_1,k_2}(x) \;=\;
    \begin{cases}
        \fx_{k_0,k_1,k_2}(x), & \text{if } k_0 \in H, \\
        \fx^{-1}_{k_0 \oplus \alpha,\,k_1,\,k_2}(x), & \text{if } k_0 \in \alpha \oplus H.
    \end{cases}
\end{align*}
The post-quantum security of $\widetilde{\fx}$ then follows directly from the same security for \fx. 

Another application of our \fx\ analysis is
\textsf{PRIDE}. Since \textsf{PRIDE} instantiates \fx\ with the constraint
$k_1 = k_2$, the relevant construction in the QICM is
$\fx_{k_0,\,k_1,\,k_1}(x)$, which is already covered by our general proof.
Because our analysis relies only on the marginal uniformity of the keys,
this special case fits directly within our framework. As in \expref{Section}{subsec:application}, the post-quantum security of \textsf{PRIDE} follows immediately from our \fx\ result.

\paragraph{\textbf{\emph{Tweakable ciphers.}}} As a second major application of our technique for the QICM, we establish post-quantum security of two tweakable block ciphers: \lrw and \xext in \expref{Section}{sec:TBC}. Tweakable block ciphers like \lrw have a variety of applications, e.g., to modes for authentication and authenticated encryption~\cite{AC:Rogaway04}. The \xext cipher forms the core of the widely-used and standardized XTS-AES disk encryption scheme~\cite{P1619D4,NIST.SP.800-38E}.

The \lrw construction is defined by
\begin{equation}
    \lrw^{E,h}_{k,k'}(\tau,x) = E_k(x \oplus h_{k'}(\tau)) \oplus h_{k'}(\tau)\,.
\end{equation}
Here, $\tau$ is a tweak parameter, $h$ is an XOR-universal hash function, and $E$ is a block cipher. \xext is defined similarly, except that (roughly speaking) $E_k'$ (for independent $k'$) is used in place of $h$. Our post-quantum bound for distinguishing \lrw from an ideal tweakable cipher is similar to that for \fx, but with an additional additive term of $6 q_C^2 \cdot 2^{-n}$. This corresponds to a (purely classical) collision attack that is possible against \lrw but not against \fx. For \xext, our bound is $q_Q^2\cdot 2^{-m} + 3q_C^2\cdot 2^{-n}$, and is derived from \Cref{thm:intro-general} below.

In the classical setting, \lrw achieves birthday-bound security. In the ideal cipher model, its distinguishing advantage is bounded by approximately $2^{-n/2} + 2^{-m/2}$. Prior work~\cite{liskov2002tweakable,jha2020tight} establishes that any classical adversary requires $q = O(2^{n/2})$ queries to distinguish \lrw from an ideal tweakable cipher, and this bound is tight due to the existence of matching collision-based attacks.
Similarly, the standardized $\xext$ construction used in $\textsf{XTS-AES}$~\cite{P1619F} has been extensively analyzed~\cite{clunie2008public,liskov2008comments,rogaway2013evaluation}, exhibiting comparable birthday-bound security in the classical model. 

In the post-quantum setting, our additional term $q_C^2 \cdot 2^{-n}$ in the bound corresponds directly to the classical collision attack that also limits \lrw’s classical security. The remaining terms in our bound arise from the inherent quantum speedups in key search and collision finding.

\paragraph{\textbf{\emph{Block cipher modes.}}} In \expref{Section}{sec:BCM}, we give the following general lifting theorem for establishing post-quantum security of block cipher modes.

\begin{theorem}[informal]\label{thm:intro-general}
Let $\textsf{Exp}$ be a security experiment and $\textsf{Con}$ a construction instantiated with a block cipher $E$. Then the post-quantum security of $\textsf{Con}$ is bounded as follows:
\begin{align}
    \textsf{Adv}^{\textsf{Exp-PQ}}_{\textsf{Con}[E]}(q,t) \leq \textsf{Adv}_E^{\qsprp}(q',t) + \delta(q),
\end{align}
where $q'$ is the number of $E$-queries made by $\textsf{Con}$ and the challenger,  $(q,t)$ denotes the (query, time)-complexity of the quantum adversary, and $\delta(q)$ is the classical information-theoretic security of $\textsf{Con}$.
\end{theorem} 
While this lifting theorem is straightforward to prove, it offers a useful and rather general observation for establishing post-quantum security for a wide variety of block cipher modes, including \cbc, \ecbc, \cmac, \gcm, and \textsf{GCM-SST}. It can also be applied in the QICM, with the SPRP advantage term becoming $(q')^2 / 2^m$, corresponding to a lower bound for quantum key search. This also yields lower bounds for \lrw and \xext, although, in these cases, better bounds are available using the hybrid technique discussed above. Some lower bounds computed via \Cref{thm:intro-general} are given in \Cref{app:modes-table}. 

We remark that, while \Cref{thm:intro-general} can be applied to \fx, \lrw, and \xext in the ideal cipher model, the specialized bounds we establish using the hybrid technique are far better.

\subsection{Related Work}\label{sec:related_work}

We focus on related work concerning post-quantum (i.e, Q1-model) security of keyed symmetric ciphers. We first discuss lower bounds, and then attacks. Since our results are query lower bounds, we will account separately for offline computation time and queries to the ``offline primitive'' (e.g., a public random permutation). Some cryptanalytic works do not make this distinction.

Alagic et al.~\cite{EC:ABKM22} established a tight lower bound for \emr using the hybrid technique discussed above. These methods were then extended to prove security of tweakable Even-Mansour~\cite{EC:ABKMS24}. Recently, these techniques were also used to prove results about the security of key-alternating Feistel~\cite{basak2025post}, multi-round Even-Mansour~\cite{bai2024quantum}, and Ascon~\cite{hosoyamada2025post}.

All ciphers under consideration can be attacked by first making (a small number of) classical queries, and then performing Grover search on the offline primitive for the full key. Kuwakado and Morii observed that BHT collision-finding~\cite{brassard1997quantum} can be used to perform key recovery for Even-Mansour using $O(2^{n/3})$ classical and quantum queries (and the same amount of time and space)~\cite{kuwakado2012security}. Bonnetain et al. improved on this with the offline Simon algorithm, achieving the same query complexity but only using polynomial quantum space~\cite{AC:BHNSS19}.

The offline-Simon algorithm also applies to key recovery on the \fx, demanding $O(2^{(m+n)/3})$ classical queries while maintaining $\textsf{poly}(n)$ classical memory. Additionally, the earlier meet-in-the-middle attack~\cite{hosoyamada2018cryptanalysis} on the \fx requires $O(2^{3(m+n)/7})$ classical queries and $O(2^{(m+n)/7})$ classical memory. Bonnetain et al.~\cite{EC:BonSchSib22} propose the first general quantum key-recovery attack in the post-quantum setting on a symmetric block cipher, offering a super quadratic speedup compared to the best classical attacks. Their work extends the offline-Simon algorithm to attack the $\textsf{2XOR}$ Cascade construction~\cite{EC:GazTes12}, achieving a $2.5\times$ quantum speedup in the exponent over the best-known classical attack.

\section{Preliminaries}

\paragraph{Notations and definitions.}
Sampling an element $s$ uniformly at random from a set $S$ is denoted by $s \leftarrow S$. We let $\permset{n}$ denote the set of all permutations on~$\bool^n$. A block cipher $E: \bool^m \times \bool^n \rightarrow \bool^n$ is a keyed permutation, i.e.,  $E_k(\cdot) = E(k, \cdot)$ is a permutation of $\bool^n$ for all $k \in \bool^m$. We also define the \emph{swap operator}  $\swap{a}{b}: \mathcal{X} \to \mathcal{X}$ as 
\begin{align}
    \swap{a}{b}(x) =
    \begin{cases}
    a, & \text{if } x = b,\\
    b, & \text{if } x = a,\\
    x, & \text{otherwise.}
    \end{cases}
\end{align}
That is, $\swap{a}{b}$ exchanges $a$ and $b$ and leaves all other elements unchanged.

For an oracle $\Or$, we write $\pm \Or$ to denote two-directional access, i.e., the adversary gets access to both $\Or$ and $\Or^{-1}$. 
Given a function $F:\bool^t \rightarrow \bool^\ell$, we let $\ket{F}$ denote the appropriate quantum oracle $\ket{x}\ket{y} \rightarrow \ket{x}\ket{y \oplus F(x)}$ for making quantum queries to $F$. 
So, for example, the notation $\A^{F, \ket{\pm P}}$ denotes an algorithm $\A$ that can make classical queries to a function $F$ (forward only) and quantum queries to a permutation $P$ (in both the forward and inverse direction).

\paragraph{A reprogramming lemma.}
For a function $F:\bool^{\ell} \rightarrow \bool^n$ and a set $B \subset \bool^{\ell} \times \bool^n$ such that each $x \in \bool^{\ell}$ is the first element of at most one tuple in~$B$,   define
\begin{align}
F^{(B)}(x) := 
\begin{cases}
y &\text{if } (x, y) \in B\\
F(x) &\text{otherwise.}
\end{cases}
\end{align}
The following is taken verbatim from~\cite[Lemma~3]{EC:ABKM22}:

\begin{lemma}[Reprogramming Lemma]
	\label{lma:reprogramming}
	Let $\D$ be a quantum distinguisher in the following experiment:
	\begin{description}
		\item[Phase 1:]  $\D$ outputs descriptions of a function $F_0=F: \bool^{\ell} \rightarrow \bool^n$ and a randomized algorithm $\mathcal B$ whose output is a set $B \subset \bool^{\ell} \times \bool^n$ where each $x \in \bool^{\ell}$ is the first element of at most one tuple in~$B$. Let $B_1 = \{x \mid \exists y: (x, y) \in B\}$
		and $	\epsilon = \max_{x \in \bool^{\ell}}\left\{\Pr_{B \leftarrow \B}[x \in B_1]\right\}.$
		
		\item[Phase 2:]  $\mathcal{B}$ is run to obtain~$B$. Let $F_1=F^{(B)}$.
		A uniform bit~$b$ is chosen, and
		$\D$ is given quantum access to~$F_b$. 
		\item[Phase 3:]  $\D$ loses access to $F_b$, and  receives the randomness~$r$  used to invoke~$\algo B$ in phase~2. Then $\D$ outputs a guess~$b'$. 
	\end{description}
	For any $\D$ making $q$ queries in expectation when its oracle is $F_0$, 
	it holds that 
	\begin{align}
	\left|\Pr[\mbox{$\D$ outputs 1} \mid b=1] - \Pr[\mbox{$\D$ outputs 1} \mid b=0]\right| \leq 2q \cdot \sqrt{\epsilon}
	\,.
	\end{align}
\end{lemma}
 \section{Ideal Cipher Resampling}
\label{sec:resampling_ic}

The hybrid technique described in \Cref{sec:tech-sum} is rather general, and could in principle be applied to a wide variety of settings. One particular setting of interest is the Quantum Ideal Cipher Model, where the oracle $\ket{E}$ gives quantum oracle access to the ideal cipher, and the oracle $E_{(k)}$ gives classical oracle access to some construction based on $E$ (e.g., a tweakable cipher.) The main technical ingredient that is then required is an appropriate resampling lemma. We now briefly introduce the Quantum Ideal Cipher Model (QICM) and state our new resampling lemma for that model.

\paragraph{Quantum Ideal Cipher Model.} The QICM was first discussed in~\cite{AC:HosYas18}. Define $\mathcal{E}(m,n)$ to be the space of all keyed permutations where the key $k \in \{0,1\}^m$ and the input $x \in \{0,1\}^n$. In the quantum ideal cipher model, we treat the underlying block cipher as an ideal cipher $E \leftarrow \mathcal{E}(m,n)$, meaning that for each key $k$, the function $E_k$ is a random permutation. Additionally, adversaries have quantum-computational capabilities in this model, and are thus allowed to make both forward and backward quantum queries to the ideal cipher.

In the security notions mentioned above, we consider algorithms having only classical access to secretly keyed primitives. When we consider constructions of keyed primitives (e.g., a tweakable block cipher) from public primitives (e.g., a random permutation), however, we provide the distinguisher with \emph{quantum} oracle access to the public primitive. Thus, for example, a quantum distinguisher in the QICM can apply the unitary operators
\begin{align}
\ket{k}\ket{x}\ket{y} &\mapsto \ket{k}\ket{x}\ket{E_k(x) \oplus y}\\
\ket{k}\ket{x}\ket{y} &\mapsto \ket{k}\ket{x}\ket{E_k^{-1}(x) \oplus y}
\end{align}
to quantum registers of the adversary's choice. We will denote an algorithm $\A$ with such access to a block cipher $E$ by $\A^{\ket{\pm E}}$.

\paragraph{A resampling lemma for the QICM.} 
We first give a resampling lemma for the ideal cipher model, generalizing a lemma of Hosoyamada~\cite[Lemma~3]{hosoyamada2025post}. Our generalization differs from~\cite{hosoyamada2025post} in that it works in the QICM rather than the random permutation model, and in that it allows an arbitrary (distinguisher-chosen) distribution of the key and resampling points, rather than restricting to uniform sampling.

\begin{lemma}\label{lem:resampling-ic}
Let $\algo D=(\D_0,\D_1)$ be a quantum distinguisher interacting with the 
following experiment:
\begin{description} 
    \item[Phase 1:] Choose a uniform ideal cipher $E \in \algo E(m, n)$ and give $\D$ 
    quantum access to $E$ and $E^{-1}$. Then $\algo D_0$ outputs a distribution $D$ 
    over $\bool^{m+2n}$.
    \item[Phase 2:] Sample $(k_0,s_0,s_1) \in \bool^m \times \bool^n \times \bool^n$ 
    according to $D$. Define $E^{(0)} = E$ and $E^{(1)}$ as
    \begin{equation}
        E_{k^*}^{(1)}(x) =
        \begin{cases}
            E_{k^*}(x), & \text{if } k^* \neq k_0, \\[1ex]
            E_{k^*} \circ \swap{s_0}{s_1}(x), & \text{if } k^* = k_0,
        \end{cases}
    \end{equation}
    where $\swap{s_0}{s_1}$ is the transposition swapping $s_0$ and $s_1$. 
    A uniform bit $b \in \bool$ is chosen, and $\D$ is given $k_0, s_0, s_1$ 
    along with quantum access to $E^{(b)}$. Finally, $\D$ outputs a guess $b'$. 
\end{description}

For a distribution $D$ on $\{0,1\}^{m+2n}$, define
    \begin{equation}
        \epsilon \;=\; \max_{(k_0^*, s_0^*, s_1^*) \in \{0,1\}^{m+2n}} D(k_0^*, s_0^*, s_1^*) \,.
    \end{equation}
Then for any distinguisher $\D$ making at most $q$ queries to $E$ in Phase~1, it holds that
    \begin{equation}
        \left|\Pr[\D \text{ outputs } 1 \mid b=1] - \Pr[\D \text{ outputs } 1 \mid b=0]\right|
        \;\leq\; 4 \sqrt{2^n \cdot q \cdot \epsilon}\,.
    \end{equation}
\end{lemma}

The proof follows the approach of~\cite{EC:ABKM22,hosoyamada2025post}, which proceeds by bounding the trace distance between the quantum states produced after Phase~1 in the original case ($b=0$) and the reprogrammed case ($b=1$). Crucially, the distinguishing advantage depends only on the queries made in Phase~1: the swap operation itself does not alter the distribution of $E$ globally, and so unless the adversary has queried the resampled points during Phase~1, the two cases are identically distributed. The proof of \Cref{lem:resampling-ic} is given in \expref{Appendix}{app:arl-ic}.

\section{Post-Quantum Security of \fx}
\label{sec:fx}

\subsection{The construction}

The \fx key-length extension construction was analyzed in~\cite{C:KilRog96,JC:KilRog01}. It is defined as follows.

\begin{definition} [\fx construction] \label{def:fx}
    Let $m$ and $n$ be positive integers. Let $E: \{0,1\}^m \times \{0,1\}^n \rightarrow \{0,1\}^n$ be a block cipher. The \fx construction is defined as 
    \begin{equation}
        \fx^{E}_{K}(x) = E_{k_0}(x \oplus k_1) \oplus k_2\,,
    \end{equation} 
     where $K = (k_0, k_1, k_2)$ with $k_0 \leftarrow \{0,1\}^m$ sampled independently and uniformly, and the marginal distributions of $k_1, k_2 \in \{0,1\}^n$ are uniform. For the remainder of this paper, we use $\mathcal{K}$ to denote this distribution over $K$.
\end{definition}

\subsection{Post-Quantum Security}

We now state and prove a post-quantum security bound for \fx (\expref{Definition}{def:fx}). As in the classical case discussed above, we are concerned with the maximum distinguishing advantage between \fx and an independent uniformly random permutation in the ideal-cipher model. 

\begin{theorem} \label{thm:Q1-secure-FX}
    Let $\A$ be an adversary making $q_C$ classical queries to its first oracle and $q_Q$ quantum queries to its second oracle. Then
    \begin{align*}
        \left| \Pr_{
        	{(k_0,k_1,k_2) \leftarrow \{0,1\}^{m+2n};\  
        		  E \leftarrow \mathcal{E}(m,n)}} \left[ \A^{\pm\fx_{K}[E],\ket{\pm E}} = 1\right] \right. \\
        \left .- \Pr_{
        		  {R \leftarrow  \mathcal{P}(n); 
        		  	 E \leftarrow  \mathcal{E}(m,n)}} \left[ \A^{\pm R,\ket{\pm E}} = 1\right] \right|  \\ 
       \leq \frac{4}{\sqrt{2^m(2^n-q_C+1)}}\cdot  q_C \sqrt{q_Q} + \frac{2}{\sqrt{2^{m+n}}}q_{Q}\sqrt{q_C}.
    \end{align*}
\end{theorem}
\begin{proof}

For a general adversary, whether a particular classical query is made in the forward or inverse direction can be chosen in an arbitrary, adaptive manner. Without loss of generality, we consider adversaries who fix this order in advance; this incurs a factor 2 cost in the total number of queries. Our proof will proceed using a hybrid-by-classical-queries approach. We will now give the proof in the case where every classical query is in the forward direction. The case where the relevant classical query is in the inverse direction is handled analogously, and we omit the detailed proof.

We begin by setting down a way of modifying a given cipher based on a choice of key and a list of classical queries to an \fx oracle. Let $E \in \mathcal{E}(m,n)$, $K=(k_0,k_1,k_2) \leftarrow \mathcal{K}$, and fix a list $\{(x_1,y_1),(x_2,y_2),\ldots, (x_{q_C},y_{q_C})\}$. Each pair $(x_i,y_i)$ represents an input--output pair of a classical query. Repeated queries are not allowed, i.e., $x_i\neq x_{i'}$ for all $i\neq i'$. For any $j \leq q_C$, let $T_j$ be the list containing the first $j$ queries. The ``modified cipher'' after $j$-many queries is denoted $E^{T_j,K}$, and is given by
    \begin{align} \label{eqn:adaptswap}
    E_{k^*}^{T_j,K}(x)= \begin{cases}
        E_{k^*}(x) &\text{ if } k^* \neq k_0\\
        E^{T_j,K}_{k_0}(x) &\text{ if } k^* = k_0.
    \end{cases}
    \end{align}
where $E^{T_j,K}_{k_0}$ is defined as follows. First, define $E^{T_0,K}=E$, and for all $j \in [1, q_C]$,
\begin{equation}
     E^{T_j,K}_{k_0}(x) = \swap{E^{T_{j-1},K}_{k_0}( x_{j}\oplus k_1)}{y_{j}\oplus k_2} \circ E^{T_{j-1},K}_{k_0}(x)\,.
\end{equation}

Note that the above definition differs from prior work~\cite{EC:ABKM22,EC:ABKMS24},\ \cite{basak2025post}. Specifically, $E^{T_i,K}$ is now constructed recursively based on $E^{T_{i-1},K}$, whereas in previous work, it was not. Roughly speaking, $E^{T_j, K}$ denotes a slight modification of $E$ that remains consistent with the transcript $T_j$, as stated in \expref{Proposition}{prop:E-T_jk-y}. The proof is straightforward and we omit the routine details.
    
\begin{proposition}\label{prop:E-T_jk-y}
    For any $E \in \mathcal{E}(m,n)$, $K = (k_0, k_1, k_2) \leftarrow \mathcal{K}$, $j \in \{1, \ldots, q_C\}$, transcript $T_j = \{(x_1, y_1), \ldots, (x_j, y_j)\}$ without repetition, and any $i \in \{1, \ldots, j\}$, it holds that 
    \begin{align*}
        E^{T_j,K}_{k_0}(x_i \oplus k_1) \oplus k_2 = y_i.
    \end{align*} 
\end{proposition} 
    For compactness, we occasionally write $E^j$ in place of $E^{T_j,K}$ when $T_j$ and $K$ are understood from the context. We also set $E^0=E$.

    We now define a sequence 
    \begin{equation}
        \Hyb_0, \Hyb_0^{1}, \Hyb_0^{2}, \Hyb_0^{3}, \Hyb_1, \Hyb_1^1,  \dots, \Hyb_{q_C}^{3}
    \end{equation}
    of hybrid experiments. Each experiment begins with sampling uniform $R \in \mathcal{P}(n)$ and $E \in \mathcal{E}(m,n)$, and a uniform key $K=(k_0,k_1,k_2) \leftarrow \mathcal{K}$. The remaining steps of each hybrid are as follows.
    
  \noindent \textbf{Experiment $\Hyb_j$.}
    \begin{enumerate}
        \item Run $\A$, answering its classical queries using $R$ and quantum queries using $E$, stopping immediately \textit{before} its $(j+1)^{\text{st}}$ classical query. Let $T_j$ be the list of classical queries so far. 
        \item For the remainder of the execution of $\A$, answer its classical queries by $\fx_{K}[E^{T_j,K}]$ and its quantum queries by $E^{T_j, K}$.
    \end{enumerate}

    \noindent \textbf{Experiment $\Hyb_j^{1}$.}
    \begin{enumerate}
        \item Run $\A$, answering its classical queries using $R$ and its quantum queries using $E$, until $\A$ makes its $(j+1)^{\text{st}}$ classical query, which we assume to be in the forward direction.
        Let $T_{j}$ be the list of classical queries so far. 
        
        \item Define set $S= \bool^n \backslash \{x_1\oplus k_1,\cdots,x_j\oplus k_1\}$. Choose uniform $s \in S$, and define $E^{(1)}$ as
        \begin{align*}
            E_{k^*}^{(1)}(x)=
    	\begin{cases}
    		E_{k^*}(x) &\text{if } k^*\neq k_0\\
    		\left(E_{k_0} \circ \swap{k_1\oplus x_{j+1}}{s} \right)(x) &\text{if } k^*=k_0 \text{.}
    	\end{cases}
        \end{align*}
        Continue running $\A$, answering its remaining classical queries (including the $(j+1)^{\text{st}}$) using $\fx_K[(E^{(1)})^{T_j,K}]$, and its quantum queries using $(E^{(1)})^{T_j,K}$.
    \end{enumerate}

   \noindent \textbf{Experiment $\Hyb_j^{2}$.}
    \begin{enumerate}
        \item Run $\A$, answering its classical queries using $R$ and its quantum queries using $E$, until $\A$ makes its $(j+1)^{\text{st}}$ classical query, which we assume to be in the forward direction. Let $T_{j+1}$ be the list of classical queries so far, and $T_j$ the subset consisting of the first $j$ queries.
        \item Answer query $j+1$ as in $\Hyb_j^{1}$, i.e., with $(j+1)^{\text{st}}$ query using $y_{j+1} := \fx_K[(E^{(1)})^{T_j,K}](x_{j+1})$.
        Then construct $E^{j+1} \equiv E^{T_{j+1},K}$ as defined adaptively as in \expref{Equation}{eqn:adaptswap}.
        Continue running $\A$, answering its remaining classical queries using $\fx_K[E^{T_{j+1},K}]$, and its quantum queries using $E^{T_{j+1},K}$.
    \end{enumerate}

   \noindent \textbf{Experiment $\Hyb^3_j$.}
    \begin{enumerate}
        \item Run $\A$, answering its classical queries using $R$ and its quantum queries using $E$, stopping immediately \textit{after} its $(j+1)^{\text{st}}$ classical query. Let $T_{j+1}$ be the set of classical queries so far.
        \item For the remainder of the execution of $\A$, answer its classical queries using $\fx_K\left[E^{T_{j+1},K}\right]$ and its quantum queries using $E^{T_{j+1},K}$, i.e. $\ket{E^{j+1}}$. 
    \end{enumerate}

    We can compactly represent the hybrids $\left\{\Hyb_j, \Hyb_j^{1}, \Hyb_j^{2}, \Hyb_j^{3}, \Hyb_{j+1}\right\}$ as the experiments in which $\A$'s queries are answered using the following oracle sequences. Let $(E^{(1)})^j$ denote $(E^{(1)}_{k_0})^{T_j,K}$.
    \begin{alignat*}{5}
        \Hyb_j: &\ket{E}, R, \ket{E}, \cdots, R, \ket{E}, 
        ~&& \fx_K[~~~E^j~~], \ket{E^j}~~~~~~\;, 
        \fx_K \left[E^j\right], \ket{E^j}, \cdots\\
        \Hyb_j^{1}: &\ket{E}, R, \ket{E}, \cdots, R, \ket{E}, &&\fx_K[(E^{(1)})^j], \ket{(E^{(1)})^j} , 
        \fx_K[(E^{(1)})^j], \ket{(E^{(1)})^j}, \cdot\\
        \Hyb_j^{2}:  &\ket{E}, R, \ket{E}, \cdots, R, \ket{E},
        ~&&\fx_K[(E^{(1)})^j], \ket{E^{j+1}}~~~, 
        \fx_K[E^{j+1}], \ket{E^{j+1}}, \cdots\\
        \Hyb_j^{3}: &\ket{E}, R, \ket{E}, \cdots, R, \ket{E}, ~&&~~~~~~~R~~~~~~~,\ket{E^{j+1}}~~~,
        \fx_K[E^{j+1}], \ket{E^{j+1}}, \cdots\\
        \Hyb_{j+1}:  &\underbrace{\ket{E}, R, \ket{E}, \cdots, R, \ket{E}}_{j \text { classical queries }}, 
        ~&&\underbrace{~~~~~~~R~~~~~~~,\ket{E}~~~~~~}_{(j+1)^{\text{st}} \text{ classical query}}, 
        \underbrace{\fx_K[E^{j+1}], \ket{E^{j+1}}, \cdots}_{q_C-j-1 \text { classical queries }}.
    \end{alignat*}
    
    Then we establish the following bounds on the distinguishability of $\Hyb_j$ and $\Hyb_{j+1}$, step by step, for $0 \leq j < q_C$:
    \begin{align*}
       \text{\expref{Lemma}{lma:fx-Hj-Hj1}:} & \quad | \Pr[\A(\Hyb_j)=1] - \Pr[\A(\Hyb_j^{1})=1]| \leq 4\sqrt{\frac{q_Q}{2^m(2^n-j)}} \\
        \text{\expref{Lemma}{lma:fx-Hj1-Hj2}:}  & \quad \Hyb_j^{1} = \Hyb_j^{2} \qquad \\
        \text{\expref{Lemma}{lma:fx-Hj2-Hj3}:}& \quad \Hyb_j^{2} = \Hyb_j^{3}\\
        \text{\expref{Lemma}{lma:fx-Hj3-Hj+1}:} & \quad | \Pr[\A(\Hyb_j^{3})=1] - \Pr[\A(\Hyb_{j+1})=1]| \leq 2 \cdot q_{Q,j+1} \sqrt{\frac{2(j+1)}{2^{m+n}}},
    \end{align*}
     where $q_{Q,j+1}$ is the expected number of queries $\A$ makes to $P$ in the $(j+1)^{\text{st}}$ stage, i.e., the stage between the $(j+1)^{\text{st}}$ and $(j+2)^{\text{nd}}$ classical queries.
    
    Using the above, we have 
    \begin{align} \label{eqn:bound}
        &| \Pr[\A(\Hyb_0)=1] - \Pr[\A(\Hyb_{q_C})=1]| \nonumber \\
        &\leq \sum_{j=0}^{q_C-1} \left(4\sqrt{\frac{q_Q}{2^m(2^n-j)}}+ 2 \cdot q_{Q,j+1} \sqrt{\frac{2(j+1)}{2^{m+n}}} \right) \nonumber\\
        &\leq \sum_{j=0}^{q_C-1} \left(4\sqrt{\frac{q_Q}{2^m(2^n-j)}}+ 2 \cdot q_{Q,j+1} \sqrt{\frac{2q_C}{2^{m+n}}}\right) \nonumber \\
        &\leq \frac{4}{\sqrt{2^m(2^n-q_C+1)}}\cdot  q_C \sqrt{q_Q} + \frac{2}{\sqrt{2^{m+n}}}q_{Q}\sqrt{q_C}.
    \end{align}
    \qed
\end{proof}
\noindent\textbf{Remark.} 
When $q_C <\frac 3 4 2^n$, our bound can be simplified to 
\[
\frac{8}{\sqrt{2^{m+n}}}\bigl(q_C\sqrt{q_Q} + q_Q\sqrt{q_C}\bigr).
\] 

As $q_C$ approaches $2^n$ (in particular when $q_C = 2^n$), the bound above is no longer tight. In this case, as explained in \expref{Section}{sec:q1-sec-ic}, the distinguishing problem between $(R,E)$ and $(\fx_{K},E)$ is equivalent to the \textsf{UNIQUE-SEARCH} problem on $k_0$. By \expref{Corollary}{thm:pq-security-ideal-cipher}, this yields an advantage of $\tfrac{q_Q^2}{2^m}$.

We now prove \expref{Lemma}{lma:fx-Hj-Hj1}, \expref{Lemma}{lma:fx-Hj1-Hj2}, \expref{Lemma}{lma:fx-Hj2-Hj3} and \expref{Lemma}{lma:fx-Hj3-Hj+1}.

\begin{lemma}\label{lma:fx-Hj-Hj1}
For $j=0, \ldots, \formerqC$, 
\begin{align*}
    \left|\Pr[\A(\Hyb_j) = 1] - \Pr[\A(\Hyb^{1}_j)=1]\right| \leq 4\sqrt{\frac{q_Q}{2^m(2^n-j)}}\,.
\end{align*}
\end{lemma}
\begin{proof}
In this lemma, we bound the distinguishability of $\Hyb_j$ and $\Hyb_j^{1}$.
\begin{alignat*}{5}
\Hyb_j: \;\; &\ket{E}, R, \ket{E}, \cdots, R, \ket{E}, 
~&& \fx_K[~~~E^j~~], \ket{E^j}~~~~~~, 
\fx_K \left[E^j\right], \ket{E^j}, \cdots\\
\Hyb_j^{1}: \;\; &\ket{E}, R, \ket{E}, \cdots, R, \ket{E},  ~&&\fx_K[(E^{(1)})^j], \ket{(E^{(1)})^j}, 
\fx_K[(E^{(1)})^j], (E^{(1)})^j, \cdots
\end{alignat*}
Let $\A$ be a distinguisher between $\Hyb_j$ and $\Hyb_j^{1}$. 
We construct from $\A$ a distinguisher $\D$ for the resampling experiment of \expref{Lemma}{lem:resampling-ic}.
$\D$ does:	
	\begin{description}
		\item[Phase 1:]
		$\D$ is given quantum access to an ideal cipher~$E$.
		It samples a uniform $R\from \algo P_n$ and then runs $\A$, answering its quantum queries with $E$ and its classical queries with $R$ (in the appropriate directions), until $\A$ submits its \mbox{$(j+1)^{\text{st}}$} classical query~$x_{j+1}$.
		At that point, $\D$ has a list $T_j=\{(x_1, y_1), \cdots, (x_j, y_j)\}$ of the queries/answers $\A$ has made to its classical oracle thus far. Next, $\D$ constructs a distribution $D$ on $\bool^{m+2n}$ and its sampling algorithm $\Pi$. To sample a tuple $(a_0,z_0,z_1)\leftarrow D$, $\Pi$ does the following:
  \begin{description}
      \item[1]:  Sample uniform $a_0\in \bool^m$ and $z_0 \in \bool^n$.
      \item[2]: Construct $S=\bool^n \backslash \{z_0\oplus x_1\oplus x_{j+1},\cdots, z_0 \oplus x_j \oplus x_{j+1}\}$.
      \item[3]: Sample uniform $z_1\in S$, and output $(a_0,z_0,z_1)$.
  \end{description}
    \item[Phase 2:]  
    $\D$ is given $(k_0,s_0,s_1)\leftarrow D$ 
    and quantum oracle access to a cipher~$E^{(b)}$.
    Then $\D$  
    sets $k_1=x_{j+1}\oplus s_0$, $k_2 \leftarrow \mathcal{K}_{|k_0,k_1}$\footnote{
This denotes the conditional distribution of $k_2$ given $(k_0,k_1)$.
Note that while this conditional distribution may depend on $(k_0,k_1)$, 
the induced marginal distribution of $k_2$ is uniform by definition of $\mathcal{K}$.
} and $K=(k_0,k_1,k_2)$. 
    It then continues running $\A$, answering its remaining classical queries (including the $(j+1)^{\text{st}}$) using $\fx_K[(E^{(b)})^{T_j,K}]$, and its remaining quantum queries using $(E^{(b)})^{T_j,K}$.	
    $\D$ outputs whatever $\A$ does. 
    \end{description}
    
 Defining $ S^{\bot}=\{x_1\oplus k_1, \cdots x_j \oplus k_1\}$, we have $s_1 \in \bool^n \setminus S^{\bot}$ from algorithm $\Pi$. In phase~1, distinguisher $\D$ perfectly simulates experiments $\Hyb_j$ and $\Hyb_j^{1}$ for $\A$ until the point where $\A$ makes its $(j+1)^{\text{st}}$ classical query. In phase~2, we first note that $k_1$ is uniform since $s_0$ is uniform and independent of $x_{j+1}$. If $b=0$, $\D$ gets access to $E^{(0)} = E$ in phase~2.
Since $\D$ answers all quantum queries using $(E^{(0)})^{T_j,K}$ and all classical queries using $\fx_K[(E^{(0)})^{T_j,K}]$, we see that $\D$ perfectly simulates $\Hyb_j$ for~$\A$ in that case.
If, on the other hand, $b=1$ in phase~2, then $\D$ gets access to $(E^{(1)})^{T_j,K}$, where
\begin{alignat*}{1}
    E_{k^*}^{(1)}(x)&=
    \begin{cases}
        E_{k^*}(x) &\text{if } k^*\neq k_0\\
        E_{k_0} \circ \swap{s_0}{s_1}(x)  &\text{if } k^*=k_0
        \,.
    \end{cases}
\end{alignat*}
Since $k_1 \coloneqq s_0 \oplus x_{j+1}$, it holds that 
\begin{alignat*}{1} 
    E_{k^*}^{(1)}(x)&=
    \begin{cases}
        E_{k^*}(x) &\text{if } k^*\neq k_0\\
        E_{k_0} \circ \swap{k_1 \oplus x_{j+1}}{s_1}(x)  &\text{if } k^*=k_0
        \,.
    \end{cases}
\end{alignat*}
Moreover, the fact that $s_0$ (and hence $s_0 \oplus x_{j+1}$), $k_0$ and $k_2$ are uniform implies that $\D$ perfectly simulates $\Hyb_j^{1}$ for~$\A$. Applying \expref{Lemma}{lem:resampling-ic} thus gives
\begin{align*}
    \left|\Pr[\A(\Hyb_j) = 1] - \Pr[\A(\Hyb_j^{1})=1]\right|
    \leq 4\sqrt{q_Q\cdot\varepsilon\cdot 2^n},
\end{align*}
where,
\begin{align*}
    \epsilon&=\max_{(k_0^*, s_0^*, s_1^*)\in\{0,1\}^{m+2n}}D(k_0,s_0,s_1)
    =\frac{1}{2^{m+n}(2^n-j)}.
\end{align*}

Therefore, we have
\begin{align*}
    \left|\Pr[\A(\Hyb_j) = 1] - \Pr[\A(\Hyb_j^{1})=1]\right|
    \leq 4\sqrt{\frac{q_Q}{2^m(2^n-j)}}.
\end{align*}
\qed
\end{proof}

\begin{lemma}\label{lma:fx-Hj1-Hj2}
For $j=0, \ldots, \formerqC$, $\Hyb_j^{1}=\Hyb_j^{2}$.
\end{lemma}
\begin{proof}
\begin{alignat*}{5}
    \Hyb_j^{1}: \;\; &\ket{E}, R, \ket{E}, \cdots, R, \ket{E},  ~&&\fx_K[(E^{(1)})^j], \ket{(E^{(1)})^j}, 
    \fx_K[(E^{(1)})^j], (E^{(1)})^j, \cdots \\
    \Hyb_j^{2}: \;\; &\ket{E}, R, \ket{E}, \cdots, R, \ket{E},
    ~&&\fx_K[(E^{(1)})^j], \ket{E^{j+1}}~~~, 
    \fx_K[E^{j+1}], \ket{E^{j+1}}, \cdots.
\end{alignat*}
We first prove the following proposition. 
\begin{proposition}\label{prop:E(1)=E-fx}
    For any $K=(k_0,k_1,k_2) \leftarrow \mathcal{K}$, $j \in \{1,\ldots,q_C\}$, $i \in \{1, \ldots, j\}$, and all $r \in \{0,\dots, j\}$
    \begin{align*}
        (E_{k_0}^{(1)})^{T_r,K}(x_i \oplus k_1) \;=\; E_{k_0}^{T_r,K}(x_i \oplus k_1),
    \end{align*}
    when $s \notin \{x_1\oplus k_1, \cdots x_j \oplus k_1\}$.
\end{proposition}
\begin{proof}
    We will do a proof by induction on $r$. We start by the base case $r=0$. We note that since the classical queries are not repeated, $x_{j+1}\oplus k_1 \notin \{x_1\oplus k_1, \cdots x_j \oplus k_1\}$. Additionally since $s \notin \{x_1\oplus k_1, \cdots x_j \oplus k_1\}$, we have that for all $i \in \{1,\cdots, j\}$
     \begin{align*}
        (E_{k_0}^{(1)})^{T_0,K}(x_i \oplus k_1) 
    &:= E_{k_0}^{(1)}(x_i \oplus k_1)
    = E_{k_0} \circ \swap{x_{j+1} \oplus k_1}{s}(x_i \oplus k_1) \\
        &=E_{k_0}(x_i \oplus k_1).
    \end{align*}
    Assume for some $r-1 \geq 0$ that
    \begin{align*}
        (E_{k_0}^{(1)})^{T_{r-1},K}(x_i \oplus k_1) = E_{k_0}^{T_{r-1},K}(x_i \oplus k_1).
    \end{align*}
    Then by the definition,
    \begin{align*}
        (E^{(1)}_{k_0})^{T_{r},K}(x_i\oplus k_1)&= \swap{(E^{(1)}_{k_0})^{T_{r-1},K}(x_{r}\oplus k_1)}{y_{r}\oplus k_2} \circ (E^{(1)}_{k_0})^{T_{r-1},K}(x_{i}\oplus k_1)\\
          &= \swap{E_{k_0}^{T_{r-1},K}(x_{r} \oplus k_1)}{y_{r}\oplus k_2} \circ E_{k_0}^{T_{r-1},K}(x_{i} \oplus k_1)\\
        &= E_{k_0}^{T_{r},K}(x_i\oplus k_1).
    \end{align*}
    By induction, the proposition holds for all $r \in \{0,\dots, j\}$.
\qed
\end{proof}
In particular, for all $i \in \{1, \ldots, j\}$,
\begin{align*}
    (E_{k_0}^{(1)})^{T_{i-1},K}(x_i \oplus k_1) = E_{k_0}^{T_{i-1},K}(x_i \oplus k_1).
\end{align*}

To prove $\Hyb_j^{1}$ and $\Hyb_j^{2}$ are identical, it suffices to show that the quantum oracles $(E^{(1)})^j$ and $E^{j+1}$ are identical. Note that $(E_{k^*}^{(1)})^{j}=E_{k^*}^{j+1}$ for all $k^* \neq k_0$, we only need to consider the case where $k^*=k_0$.

In both $\Hyb_j^{1}$ and $\Hyb_j^{2}$, the response to the $(j+1)^{\text{st}}$ classical query is:
\begin{align*}
    y_{j+1} &\stackrel{\rm def}{=} \fx_K[(E^{(1)})^{T_j,K}](x_{j+1})= (E^{(1)}_{k_0})^{T_j,K}(x_{j+1} \oplus k_1) \oplus k_2\\
    &= \left(\prod_{i=j}^1\swap{(E_{k_0}^{(1)})^{T_{i-1},K}( x_{i}\oplus k_1)}{y_{i}\oplus k_2} \right) \circ E_{k_0}^{(1)}(x_{j+1} \oplus k_1) \oplus k_2\\
    &= \left(\prod_{i=j}^1\swap{(E_{k_0})^{T_{i-1},K}( x_{i}\oplus k_1)}{y_{i}\oplus k_2} \right) \circ E_{k_0}(s) \oplus k_2
    = E_{k_0}^{T_{j},K}(s) \oplus k_2,
\end{align*}
where the fourth equality comes from \expref{Proposition}{prop:E(1)=E-fx}. We use ``$\prod$'' to denote sequential composition of operations, i.e., $\prod_{i=1}^{n}f_i=f_1\circ\cdots\circ f_n$. By rearranging $s = \left( E_{k_0}^{T_{j},K}\right)^{-1}(y_{j+1} \oplus k_2),$ it follows that for any $x \in \{0,1\}^{n}$,
\begin{align*}
    &(E_{k_0}^{(1)})^j(x)\\
    &= \swap{(E^{(1)}_{k_0})^{{j-1}}( x_{j}\oplus k_1)}{y_{j}\oplus k_2} \circ \cdot \circ \swap{E_{k_0}^{(1)}(x_1\oplus k_1)}{y_1\oplus k_2}\circ E_{k_0} \circ \swap{x_{j+1}\oplus k_1}{s} (x) \\
    &= \swap{E^{j-1}_{k_0}( x_{j}\oplus k_1)}{y_{j}\oplus k_2} \circ \cdots \circ \swap{E_{k_0}(x_1\oplus k_1)}{y_1\oplus k_2}\circ E_{k_0} \circ \swap{x_{j+1}\oplus k_1}{s} (x) \\
    &= E_{k_0}^{j} \circ \swap{x_{j+1}\oplus k_1}{(E_{k_0}^j)^{-1}(y_{j+1}\oplus k_2)}(x) \\
    &= \swap{E_{k_0}^j(x_{j+1}\oplus k_1)}{y_{j+1}\oplus k_2} \circ E_{k_0}^j(x)
   =E_{k_0}^{j+1}(x),
\end{align*}
where the second equality comes from \expref{Proposition}{prop:E(1)=E-fx}. This concludes the proof that $(E^{(1)})^j \equiv E^{j+1}$. It follows that $\Hyb_j^{1} = \Hyb_j^{2}$.
\qed
\end{proof}

\begin{lemma}\label{lma:fx-Hj2-Hj3}
For $j=0, \ldots, \formerqC$, 
$\Hyb_j^{2} = \Hyb_j^{3}.$
\end{lemma}
\begin{proof}

\begin{alignat*}{5}
    \Hyb_j^{2}: \;\; &\ket{E}, R, \ket{E}, \cdots, R, \ket{E},
    ~&&\fx_K[(E^{(1)})^j], \ket{E^{j+1}}~~~, 
    \fx_K[E^{j+1}], \ket{E^{j+1}}, \cdots\\
    \Hyb_j^{3}: \;\; &\ket{E}, R, \ket{E}, \cdots, R, \ket{E}, ~&&~~~~~~~R~~~~~~~,\ket{E^{j+1}}~~~,
    \fx_K[E^{j+1}], \ket{E^{j+1}}, \cdots.
\end{alignat*}
We observe that $\Hyb_j^{2}$ and $\Hyb_j^{3}$ differ only in the response to the $(j+1)^{\text{st}}$ classical query. In $\Hyb_j^{2}$, this query is answered using
\begin{align*}
    y_{j+1} \;=\; \fx_K\!\big[(E^{(1)})^{T_j,K}\big](x_{j+1}) 
\;=\; E_{k_0}^{T_j,K}(s) \oplus k_2,
\end{align*}
as shown in \expref{Lemma}{lma:fx-Hj1-Hj2}. In contrast, in $\Hyb_j^{3}$ the same query is answered with $y_{j+1} \;=\; R(x_{j+1}).$

Next, we prove that $y_{j+1}$ is distributed the same in both $\Hyb_j^{2}$ and $\Hyb_j^{3}$. Recall \expref{Proposition}{prop:E-T_jk-y}, we have  
\begin{align*}
    y_i \;=\; E_{k_0}^{T_j,K}(x_i \oplus k_1) \oplus k_2, \quad \forall i \in [1,j],
\end{align*}
and since $s \in \{0,1\}^n \setminus \{x_1 \oplus k_1, \ldots, x_j \oplus k_1\}$, it follows that in $\Hyb_j^{2}$,
\begin{align*}
    y_{j+1} = E_{k_0}^{T_j,K}(s) \oplus k_2 \in \{0,1\}^n \setminus \{y_1, \ldots, y_j\}.
\end{align*}
Moreover, because $E_{k_0}^{T_j,K}$ is a permutation, the mapping
\begin{align*}
    x_i \oplus k_1 \;\mapsto\; y_i = E_{k_0}^{T_j,K}(x_i \oplus k_1) \oplus k_2
\end{align*}
is injective. Since $s$ is chosen uniformly from
$\{0,1\}^n \setminus \{x_1 \oplus k_1, \ldots, x_j \oplus k_1\}$
and $k_2$ is uniform, the output $y_{j+1}$ is uniformly distributed over $\{0,1\}^n \setminus \{y_1, \ldots, y_j\}$ in $\Hyb_j^2$. In $\Hyb_j^{3}$, since classical queries are not repeated, $y_{j+1} = R(x_{j+1})$ is also uniformly distributed 
over $\{0,1\}^n \setminus \{y_1, \ldots, y_j\}$. Thus, the distribution of $y_{j+1}$ conditioned on all previous queries is identical in the two hybrids.

Moreover, in both $\Hyb_j^{2}$ and $\Hyb_j^{3}$, the construction of $E^{j+1}$ follows exactly the same procedure, i.e., from the first $j+1$ classical input–output pairs and E as specified in \expref{Equation}{eqn:adaptswap}. Consequently, the two hybrids yield identical distributions, and hence $\Hyb_j^{2} = \Hyb_j^{3}$.
\qed
\end{proof}

\begin{lemma}\label{lma:fx-Hj3-Hj+1}
    For $j=0, \ldots, \formerqC-1$, 
    \begin{align*}
	\Pr[\A(\Hyb_j^{3}) = 1] - \Pr[\A(\Hyb_{j+1})=1]| \leq 2 \cdot      q_{Q,j+1} \sqrt{\frac{2(j+1)}{2^{m+n}}}\,,
    \end{align*}
    where $q_{Q,j+1}$ is the expected number of queries $\A$ makes to $\ket{E^{j+1}}$ in the $(j+1)^{\text{st}}$ stage in the ideal world (i.e., in $\Hyb_{q_C}$). 
\end{lemma}

\begin{proof}
	
	Let $\A$ be a distinguisher between $\Hyb_j^{3}$ and $\Hyb_{j+1}$. We construct a distinguisher $\D$ for the experiment from \expref{Lemma}{lma:reprogramming}:
	\begin{description}
		\item[Phase 1:] 
		$\D$ samples uniform $E \in \mathcal{E}(m,n)$ and $R \in \permset{n}$. It then runs $\A$, answering its quantum queries using $R$ and its classical queries using~$E$, until after it responds to $\A$'s $(j+1)^{\text{st}}$ classical query. Let $T_{j+1} = \{(x_i, y_i)\}_{i=1}^{j+1}$ be the list of classical queries by $\A$ thus far. $\D$ defines $F(a,k_0,x) \coloneqq E^{a}_{k_0}(x)$ for $a \in \{1, -1\}$.
		
        It also defines the following randomized algorithm~$\B$: sample $K\leftarrow \mathcal{K}$ and write $K=(k^*_0,k^*_1,k^*_2)$. Then it computes the set $B$ of input/output pairs to be reprogrammed so that $F^{(B)}(a, k^*_0, x)={\left(E_{k^*_0}^{T_{j+1},K}\right)}^{a}(x)$ for all $a, k^*_0, x$.  Finally, $\D$ outputs $(F, \B)$.
		\item[Phase 2:] $\B$ is run to generate~$B$, and
		$\D$ is given quantum access to an oracle~$F_b$. $\D$ resumes running~$\A$, answering its quantum queries using $F_b$. Phase~2 ends before $\A$ makes its next (i.e., $(j+2)^{\text{nd}}$) classical query.
		
		\item[Phase 3:] $\D$ is given the randomness used by $\B$ to generate $k$. It resumes running $\A$, answering its classical queries using $\fx_K[E^{T_{j+1},K}]$ and its quantum queries using~$E^{T_{j+1},K}$. Finally, it outputs whatever $\A$ outputs.
	\end{description}
	It is immediate that if $b=0$ (i.e., $\D$'s oracle in phase~2 is~$F_0=F$), then $\A$'s output is identically distributed to its output in~$\Hyb_{j+1}$, whereas if $b=1$ (i.e., $\D$'s oracle in phase~2 is~$F_1=F^{(B)}$), then $\A$'s output is identically distributed to its output in~$\Hyb_j^{3}$. It follows that $|\Pr[\A(\Hyb_j^{3}) = 1] - \Pr[\A(\Hyb_{j+1})=1]|$ is equal to the distinguishing advantage of $\D$ in the reprogramming experiment of \expref{Lemma}{lma:reprogramming}. To bound this quantity, we bound the parameter~$\epsilon$ and the expected number of queries made by $\D$ in phase~2 (when $F = F_0$).
	
The value of $\epsilon$ can be bounded using the definition of $E_{k^*_0}^{T_{j+1},K}$ and the fact that $F^{(B)}(a,k^*_0, x) = {\left(E_{k^*_0}^{T_{j+1},K}\right)}^{a}(x)$. Fixing $E$ and $T_{j+1}$, the probability that any particular input $(a, k^*_0, x)$ is reprogrammed is at most the probability (over $k$) that it lies in the set 
\begin{align*}
    \left\{\begin{array}{c}
    \left(1, k_0, x_i \oplus k_1\right),\left(1, k_0, E_{k_0}^{-1}\left(y_i \oplus k_2\right)\right), \\
    \left(-1, k_0, E_{k_0}\left(x_i \oplus k_1\right)\right),\left(-1, k_0, y_i \oplus k_2\right)
    \end{array}\right\}_{i=1}^{j+1}.
 \end{align*}
We compute the probability that $(a,k^*_0, x)=(1, k_0, x_i \oplus k_1)$ for some fixed~$i$. As $k_0$ and $k_1$ are uniform, 
\begin{equation*}
	{\textstyle \Pr_k[(a, k^*_0, x)=(1, k_0, x_i \oplus k_1)]} = \begin{cases}
		2^{-(m+n)}&a=1\\
		0&a=-1
	\end{cases}\,.
\end{equation*}

A similar bound holds for the other possibilities. By distinguishing the cases $a=1$ and $a=-1$ and applying a union bound, we get $\varepsilon \leq 2(j+1)/2^{m+n}$.

The expected number of queries made by $\D$ in phase~2 when $F=F_0$ is equal to the expected number of queries made by $\A$ in its $(j+1)^{\text{st}}$ stage in~$\Hyb_{j+1}$. 
Since $\Hyb_{j+1}$ and $\Hyb_{q_C}$ are identical until after the $(j+1)^{\text{st}}$ stage is complete, this is precisely~$q_{Q,j+1}$.
\qed
\end{proof}

\noindent\textbf{Remark.}
Our proof covers the case $k_1 = k_2$.
Since $\mathcal{K}$ allows arbitrary dependence between $k_1$ and $k_2$ as long as their marginals are uniform, choosing $k_2 = k_1$ is simply a valid (fully correlated) instantiation of $\mathcal{K}$.  
Our analysis never relies on $k_1$ and $k_2$ being distinct, so all arguments apply unchanged.

\subsection{Tightness}

Let $q_C$ represent the number of online classical queries and $q_Q$ denote the number of offline quantum computations. Different attacks with corresponding trade-off are presented in \expref{Table}{tab:fx-attack}.

Typical attacks involve full-key recovery, which classically requires $2^{m+n}$ offline queries, while quantumly, it requires only $2^{(m+n)/2}$ offline queries using Grover's algorithm. The Grover+BHT attack employs Grover's algorithm to recover the cipher key $k_0$ and uses the BHT/Kuwakado--Morii~\cite{kuwakado2012security}  to determine the whitening keys $k_1$ and $k_2$. It requires $q_Q = 2^{m/2 + n/3}$ quantum queries and $q_C = 2^{n/3}$ classical queries. Details of this attack is similar to that on \lrw, which is provided in \expref{Appendix}{app:attacks}. The resources used align with $\frac{2}{\sqrt{2^{m+n}}} q_Q \sqrt{q_C}$ in the security bound established in \expref{Theorem}{thm:Q1-secure-FX}.

In Meet-in-the-Middle attack~\cite{hosoyamada2018cryptanalysis}, $q_C \cdot q_Q^6 = 2^{3(m+n)}$ classical online queries are required under the condition that only polynomially many qubits are available, by using the multi-target preimage search~\cite{chailloux2017efficient}.  This attack was later improved by the offline-Simon attack~\cite{bonnetain2019quantum}, which combines quantum search with quantum period finding, with a trade-off $q_C \cdot q_Q^2 = 2^{m+n}$, which also matches $\frac{2}{\sqrt{2^{m+n}}} q_Q \sqrt{q_C}$ term in the security bound.

\begin{table}
    \centering
    \bgroup
\def\arraystretch{1.8}
\begin{tabular}{ |p{4.2cm}||p{8cm}|  }
 \hline
 Reference & Tradeoff between $q_C$ and $q_Q$\\
 \hline
 Classical~\cite{dinur2015cryptanalytic} & $q = 2^{m+n}$\\
 Grover & $q_Q = 2^{\frac{m+n}{2}}, q_c = \text{constant}$\\
 Grover + BHT~\cite{kuwakado2012security} & $q_Q = 2^{\frac{m}{2}+\frac{n}{3}}, q_C = 2^{\frac{n}{3}} $\\
 Meet-in-the-Middle~\cite{hosoyamada2018cryptanalysis}  & $q_C \cdot q_Q^6 = 2^{3(m+n)}, q_C \leq \{2^n, 2^{3(m+n)/7}\}$ \\
 Offline-Simon~\cite{bonnetain2019quantum} & $q_C \cdot q_Q^2 = 2^{m+n}, q_C \leq 2^n$ \\
 \hline
\end{tabular}
\egroup
\vspace{3mm}
\caption{Tradeoffs for quantum attacks on the FX construction.} \label{tab:fx-attack}
\end{table}

\subsection{Application} \label{subsec:application}

In this section, we present post-quantum security analyses of lightweight block ciphers
that instantiate the \fx construction, namely \textsf{PRINCE} and \textsf{PRIDE},
as applications of our framework.

\subsubsection{The PRINCE cipher.}

\textsf{PRINCE}~\cite{borghoff2012prince} is a lightweight block cipher designed 
for ultra-low latency applications, particularly in pervasive and embedded computing. 
Its distinguishing feature is the $\alpha$-reflection property, which makes 
encryption and decryption nearly identical, thereby enabling very efficient 
hardware implementations. Concretely, \textsf{PRINCE} is based on $\mathsf{PRINCE}_{\text{core}}$, which can be viewed as a $\widetilde{\fx}$ construction, namely the $\alpha$-reflection variant of the classical \fx design.

The $\alpha$-reflection property ensures that this permutation can be inverted 
with minimal overhead. The classical security of this design is analyzed in 
Section~4.1 of~\cite{borghoff2012prince}.

The classical security of \textsf{PRINCE} is analyzed in the ideal cipher model via the primitive $\widetilde{\fx}$, the $\alpha$-reflection variant of the standard \fx construction. Let $\mathbb{F}_2^m$ denote the $m$-dimensional vector space over $\mathbb{F}_2$. 
Let $H$ be an $(m-1)$-dimensional linear subspace of $\mathbb{F}_2^m$, and let 
$\alpha \in \mathbb{F}_2^m$ be a fixed nonzero element with $\alpha \notin H$. 
Thus, $\mathbb{F}_2^m$ is partitioned as $H \cup (\alpha \oplus H)$. Define
\begin{align*}
    \widetilde{\fx}_{k_0,k_1,k_2}(x) \;=\;
    \begin{cases}
        \fx_{k_0,k_1,k_2}(x), & \text{if } k_0 \in H, \\
        \fx^{-1}_{k_0 \oplus \alpha,\,k_1,\,k_2}(x), & \text{if } k_0 \in \alpha \oplus H.
    \end{cases}
\end{align*}

\cite[Corollary~1]{borghoff2012prince} shows that $\widetilde{\fx}$ achieves the 
same level of security as the original $\fx$ construction, in the pure classical setting.

\begin{corollary}[\cite{borghoff2012prince}] \label{coro:PRINCE}
    $Adv^{\sprp\textsf{-IC}}_{\fx}(q_{\fx},q_E) 
    = Adv^{\sprp\textsf{-IC}}_{\widetilde{\fx}}(q_{\widetilde{\fx}},q_E).$
\end{corollary}

The classical proof of \expref{Corollary}{coro:PRINCE} carries over directly to the post-quantum setting. Together with \expref{Theorem}{thm:Q1-secure-FX}, this implies the post-quantum security of $\widetilde{\fx}$, and hence the post-quantum security of \textsf{PRINCE} in the QICM, i.e., where $\sf{PRINCE_\text{core}}$ is replaced by an ideal cipher to which everyone has quantum access. 

\subsubsection{The PRIDE cipher.}

The classical security of \textsf{PRIDE} is analyzed via the standard \fx 
construction, instantiated with the ARX-based core $\mathsf{PRIDE}_{\text{core}}$. 
In the specification of \textsf{PRIDE}, the master key is split into two halves, 
where one half is used for input/output whitening and the other half drives 
the round function of $\mathsf{PRIDE}_{\text{core}}$. Concretely, in the ideal cipher model where $\mathsf{PRIDE}_{\text{core}}$ is replaced by an ideal cipher $E$, \textsf{PRIDE} can be expressed as the special case $\fx_{k_0,k_1,k_0}$ of the general \fx construction, i.e.,
\[
\textsf{PRIDE}(k_0,k_1; x)= k_1 \oplus E_{k_0}(x \oplus k_1).
\]
The classical security analysis of this design is provided in 
Section~5 of~\cite{albrecht2014block}. The post-quantum security of \textsf{PRIDE} follows directly from \expref{Theorem}{thm:Q1-secure-FX}.

\section{Tweakable Block Ciphers}
\label{sec:TBC}
\subsection{Definitions}

\begin{definition} [Tweakable Permutation] \label{def:tweakable-perm}
    Let $\mathcal{T}$ be a tweak space. $\widetilde{\Pi}: \mathcal{T} \times \{0,1\}^n \rightarrow \{0,1\}^n$ is a \textsf{tweakable permutation}, where for each tweak $\tau \in \mathcal{T}$, $\widetilde{\Pi}(\tau,\cdot) \xleftarrow[]{\$} \mathcal{P}_n$ is an independently and randomly chosen permutation on $\{0,1\}^n$. We also define $\mathcal{E}(\mathcal{T},n)$ to be the set of all such tweakable permutations. Additionally, we denote the ``inverse" oracle as $\widetilde{\Pi}^{-1}(\tau,\cdot) := \widetilde{\Pi}_{\tau}^{-1}(\cdot)$ for some $\tau \in \mathcal{T}$.
\end{definition}

\begin{definition} [Distinguishing Advantage] \label{def:dis-adv-tbc}
    Let $G: \{0,1\}^m \times \mathcal{T} \times \{0,1\}^n \longrightarrow \{0,1\}^n$ be a family of efficient, keyed permutations and $G^{-1}: \{0,1\}^m \times \mathcal{T} \times \{0,1\}^n \longrightarrow \{0,1\}^n$ be its corresponding decryption oracle. Consider an adversary $\A$ (which can be either quantum or classical) with time complexity $T_{\A}$. The advantage of $\A$ in attacking $G$ is measured by
    \begin{align*}
        \textsf{Adv}_G(\A,q):= \left| \Pr_{k \leftarrow \{0,1\}^m} \left[ \A^{G_k(\cdot,\cdot),G^{-1}_k(\cdot,\cdot)} = 1\right]  - \Pr_{\widetilde{\Pi} \leftarrow \mathcal{E}(\mathcal{T},n)} \left[ \A^{\widetilde{\Pi}(\cdot,\cdot),\widetilde{\Pi}^{-1}(\cdot,\cdot)} = 1\right] \right| ,
    \end{align*}
    where the probabilities are also taken over the randomness of $\A$, and $\A$ is allowed to make $q$ \textbf{classical} queries to the oracles within a time bound of $t$. 
\end{definition}

\begin{definition} [\textsf{SPRP(-PQ)} Security of Tweakable Block Cipher] \label{def:(Q)PRP-security-tbc}
    Let $G: \{0,1\}^m \times \mathcal{T} \times \{0,1\}^n \longrightarrow \{0,1\}^n$ be a tweakable block cipher. The \textsf{(Post-Quantum) strong pseudorandom permutation} security of $G$ is measured by the maximum advantage over all (quantum) adversaries $\A$:
    \begin{align*}
        \textsf{Adv}_G^{\textsf{SPRP(-PQ)}}(q,t) &= \max_{\A:T_{\A}\leq t} \textsf{Adv}_G(\A,q) \\
        &= \max_{\A:T_{\A}\leq t} \left| \Pr_{k \leftarrow \{0,1\}^m} \left[ \A^{G_k(\cdot,\cdot),G^{-1}_k(\cdot,\cdot)} = 1\right]  - \Pr_{\widetilde{\Pi} \leftarrow \mathcal{E}(\mathcal{T},n)} \left[ \A^{\widetilde{\Pi}(\cdot,\cdot),\widetilde{\Pi}^{-1}(\cdot,\cdot)} = 1\right] \right|.
    \end{align*}
    Here, the probabilities are additionally over the randomness of $\A$, and the adversary $\A$ is allowed up to $q$ classical queries to the oracles within a time bound $t$. $\widetilde{\Pi}$ is a tweakable permutation, where $\widetilde{\Pi}(\tau, \cdot)$ is a random permutation, and $\tau \xleftarrow{\$} \mathcal{T}$.  
\end{definition}

\medskip \noindent \textbf{Remark.} \expref{Definition}{def:(Q)PRP-security-tbc} provides a general definition of distinguishing advantage. However, when considering the \qsprp-security of a construction in the QICM, where the underlying block cipher $E$ is an ideal cipher, which can only be accessed through oracle queries rather than a public description; we provide adversaries with these additional oracles, $\ket{E(\cdot,\cdot)}$ and $\ket{E^{-1}(\cdot,\cdot)}$. Formally, if we consider the \qsprp-security of a construction $G[E_{\_}]$ where the cipher $E$ is accessible exclusively through oracle queries, then the security is measured as below:
\begin{align*}
    \textsf{Adv}_{G[E_{\_}]}^{\qsprp}(q,t)=\max_{\A} \left|\Pr_{k \leftarrow\{0,1\}^m}[\A^{G[E_k],G^{-1}\left[E_k^{-1}\right],\ket{E},\ket{E^{-1}}}=1] - \Pr_{\widetilde{\Pi} \leftarrow \mathcal{E}(\mathcal{T},n)} \left[ \A^{\widetilde{\Pi},\widetilde{\Pi}^{-1},\ket{E},\ket{E^{-1}}} = 1\right] \right|.
\end{align*}

\begin{definition}[XOR-Universality] A family of functions $h = \{h_k: \{0,1\}^n \rightarrow \{0,1\}^n \}_{k\in\mathcal K}$ is called $\varepsilon$-XOR universal if for a randomly drawn key $k \in \mathcal{K}$, $\forall x,y,z \in \{0,1\}^n$ and $x \neq y$, $$\Pr_{k\xleftarrow{\$}\mathcal{K}}[h_{k}(x) \oplus h_k(y) = z] \le\varepsilon. $$ 
If $\varepsilon=\frac{1}{2^n}$, $h$ is simply called XOR universal.
\end{definition}

\begin{definition} [Uniformity] A family of functions $h$ is uniform if $\forall x,y \in \{0,1\}^n$,
\begin{align*}
\Pr_{k\xleftarrow{\$}\mathcal{K}}[h_k(x) = y] = \frac{1}{2^n}.
\end{align*}
\end{definition} 

\textsf{Remark.} We note that, for these properties to hold, it is necessary that $|k| \geq n$. This further implies that for all $x, y$, there exists at least one $k$ such that $h_k(x) = y$. 

\begin{definition} [\lrw construction~\cite{liskov2002tweakable}] \label{def:lrw}
    Let $m$ and $n$ be positive integers. Let $E: \{0,1\}^m \times \{0,1\}^n \rightarrow \{0,1\}^n$ be a block cipher. Let $h$ be a hash function. The \lrw construction is defined as 
    \[
    \lrw^{E,h}_{k,k'}(\tau,x) = E_k(x \oplus h_{k'}(\tau)) \oplus h_{k'}(\tau),
    \]
    where $k \xleftarrow{\$} \{0,1\}^m$ is the block cipher key and $k'\xleftarrow{\$} \{0,1\}^{\kappa}$ is the tweak key, $x \in \{0,1\}^n$ is the input, $\tau \in \{0,1\}^*$ is the tweak, and $h$ is a hash function. For simpler notation, we write $\lrw^{E,h}_{k,k'}(\tau,x)$ as $\lrw_{k,k'}(\tau,x)$, when $E$ and $h$ are clear from the context.
\end{definition}

\begin{theorem} [Classical security of \lrw~\cite{liskov2002tweakable,minematsu2006improved}]\label{thm:lrw-security}
   Let $h$ be an $\varepsilon$-XOR-universal hash function. Then for $\lrw^{E,h}$ construction, $$\textsf{Adv}^{\sprp}_{\lrw}(q,t) \leq \textsf{Adv}^{\sprp}_{E}(q,t) + q^2\varepsilon,$$ where $q$ is the number of queries to $\lrw_{k,k'}(\cdot,\cdot)$ within a time bound $t$.
\end{theorem}

We note that the security of \xext is not analyzed via the hybrid argument in this section. Instead, its bound is derived later using the more general theorem presented in~\expref{Theorem}{thm:general-qprp-theorem}.

\begin{definition} [\xext construction~\cite{rogaway2013evaluation}]\label{def:xex2}
    Let $m$ and $n$ be positive integers. Let $E: \{0,1\}^m \times \{0,1\}^n \rightarrow \{0,1\}^n$ be a block cipher, and let $\alpha \in \mathbb{F}^{*}_{2^n}$. The tweakable block cipher construction \xext is 
    \begin{align*} 
        \xext^{E, \alpha}_{k, k'}(x, {i, j}) = E_k(x \oplus \Delta_{i, j, k'}) \oplus \Delta_{i, j, k'}, 
    \end{align*}
    where $\Delta_{i, j, k'} = \alpha^j \times_* E_{k'}(i)$. Here, $\times_*$ denotes multiplication in the finite field $ \mathbb{F}_{2^n} \setminus \{0^n\}$. The key of \xext consists of the block cipher keys $k,k'\xleftarrow{\$} \bool^m$. Furthermore, $x,i \in \{0,1\}^n$, and $j \in [0, 2^{20}-1]$.
\end{definition}

We include a classical security analysis for completeness here. To facilitate the (classical and post-quantum) security analysis of $\xext$, we define an inefficient idealized variant of $\xext$, where $E_{k'}$ is replaced with a random permutation $\pi$, which is then considered part of the key.
\begin{definition} [\xext with idealized hash]\label{def:id-hash-xex2}
The tweakable block cipher construction \idhashxext is 
        \begin{align*} 
		\xext^{E, \alpha}_{k, \pi}(x, {i, j}) = E_k(x \oplus \Delta_{i, j}) \oplus \Delta_{i, j}, 
	\end{align*}
	where $\Delta_{i, j} = \alpha^j \times_* \pi(i)$. The key of \idhashxext consists of the block cipher key $k\xleftarrow{\$} \{0,1\}^m$ and the permutation $\pi\in S_{2^n}$, while the remaining parameters are as in \cref{def:xex2}.
\end{definition}

Observe that $\idhashxext$ is the \lrw construction with the hash function family $$h^{\idhashxext}_\pi(i,j) = \Delta_{i,j} = \alpha^j \times_* \pi(i).$$
Clearly, \xext and \idhashxext are indistinguishable for an $\sprp$ adversary, up to the $\sprp$ security of $E$. 
\begin{proposition}\label{prop:xext2idhashxext}
	The security of \xext and \idhashxext is related as follows.
	\[
	\textsf{Adv}^{\sprp}_{\xext}(q,t) \leq \textsf{Adv}^{\sprp}_{E}(q,t) +\textsf{Adv}^{\sprp}_{\idhashxext}(q,t).
	\]
This inequality holds for various adversarial models, including classical, post-quantum and quantum-access.
\end{proposition}

Next, we show that the hash function in $\idhashxext$ is $\varepsilon$-XOR universal for negligible $\varepsilon$.
\begin{lemma}\label{lem:idhash-is-eps-XOR}
	The hash function family $h^{\idhashxext}$ is $\frac{1}{2^n-1}$-XOR universal.
\end{lemma}
\begin{proof}
	We bound
	\begin{align*}
		\max_{\substack{z \in \{0,1\}^n, (i_1,j_1) \neq (i_2,j_2)}} \Pr_{\pi}[\pi(i_1)\times_{*} \alpha^{j_1} \oplus \pi(i_2)\times_{*} \alpha^{j_2} = z]
	\end{align*}
	If $i_1 \neq i_2$, since $\pi$ is a random permutation, the values $\pi(i_1)$ and $\pi(i_2)$ are distinct random elements from $\text{GF}(2^n)$, implying that the XOR expression is uniformly distributed over the set $\textsf{GF}(2^n)\setminus\{\alpha^{j_1}\oplus \alpha^{j_2}\}$. In this case, the probability expression above us thus bounded by $ 1/(2^n-1)$. When $i_1 = i_2$ and $j_1 \neq j_2$, the probability expression becomes:
	\begin{align*}
		&\Pr_{\pi} [\pi(i_1) \times_{*} \alpha^{j_1} \oplus \pi(i_1) \times_{*} \alpha^{j_2} = z]\\ 
		&= \Pr_{\pi} [\pi(i_1) \times_{*} (\alpha^{j_1} \oplus \alpha^{j_2}) = z]\\
		&=\Pr_{\pi}[\pi(i_1) = z \times_* (\alpha^{j_1} \oplus \alpha^{j_2})^{-1}].
	\end{align*}
	Since $\pi$ is a random permutation, the value $\pi(i_1)$ is uniformly distributed, and thus $\varepsilon \leq \frac{1}{2^n -1}$. 
    \qed
\end{proof}
We now obtain the classical security of \xext.
\begin{theorem} [Classical security of \xext~\cite{rogaway2013evaluation}]\label{thm:xext-security}
Fix $n \geq 1$ and let $\alpha \in \mathbb{F}_{2^n}^*$ be base elements. Fix a block cipher $E: \{0,1\}^m \times\{0,1\}^n \rightarrow\{0,1\}^n$. Then
\begin{align*}
    \textsf{Adv}^{\sprp}_{\xext}(q,t) \leq 2\textsf{Adv}^{\sprp}_{E}(q,t) + \frac{q^2}{2^n-1},
\end{align*}
where $q$ is the number of queries to \xext, within a time bound $t$.
\end{theorem}
\begin{proof}
Immediate by combining \cref{thm:lrw-security,prop:xext2idhashxext,lem:idhash-is-eps-XOR}.
\qed
\end{proof}

The tweakable block cipher $\xext$ is employed in the NIST-standardized $\textsf{XTS-AES}$~\cite{P1619F}, with $\textsf{AES}$ as the underlying block cipher. However, the application-level security goals of $\textsf{XTS}$ remain poorly understood~\cite{clunie2008public,rogaway2013evaluation}. Another variant, which we denote by $\xex$, uses a single key ($k = k'$). Discussions on whether to employ a single key or two separate keys are explored in~\cite{clunie2008public,liskov2008comments}, with no definitive conclusion reached.

\medskip\noindent \textbf{Applications of \lrw and \xext.} \lrw is implemented as \textsf{LRW-AES} Encryption in earlier drafts of IEEE P1619~\cite{P1619D4}, utilizing \textsf{AES} as the block cipher $E$ and the hash function used is $h_{k'}(\tau) = k' \times_* \tau$, where $\times_*$ represents the group multiplication in $\mathbb{F}^*_{2^n}$. For a single $128$-bit block, the block cipher key $k$ is $128$, $192$, or $256$ bits, while the hash key $k'$ is $128$ bits. The tweak $\tau$ is $128$ bits, and $E$ operates on $128$-bit (tweakable) blocks.

This \textsf{LRW-AES} was later replaced by \textsf{XTS-AES} in the final IEEE P1619 standard~\cite{P1619F}, with a constraint that the data unit length for any \textsf{XTS-AES} implementation must not exceed $2^{20}$ \textsf{AES} blocks~\cite{dworkin2010recommendation}. Initially, \lrw was considered slower than \xts due to the full multiplication over $\textsf{GF}(2^n)$ required for each tweak update, but this is no longer the case~\cite{isobe2019plaintext}. The input or output data may consist of multiple $128$-bit blocks followed by a partial block of less than $128$ bits.

In the \textsf{XTS-AES} Encryption procedure~\cite{P1619F} for a single $128$-bit block, the keys $k$ and $k'$ are either $128$ bits or $256$ bits. The tweak is represented by $i$, a $128$-bit value, and $j$, a sequential number of the $128$-bit block within the data unit.

\subsection{Post-Quantum Security of \lrw using the Hybrid Technique}

\begin{theorem}\label{thm:Q1-secure-LRW-1}
    Let \lrw be as in \cref{def:lrw} and let $\A$ be an adversary making $q_C$ classical queries to its first oracle and $q_Q \geq 1$ quantum queries to its second oracle. Assuming $h$ is XOR-universal and uniform, it holds that in the ideal cipher model,
    \begin{align*}
        \left| \Pr_{\substack{k \leftarrow \{0,1\}^m; k' \leftarrow\{0,1\}^{\kappa} \\  E \leftarrow  \mathcal{E}(m,n)}} \left[ \A^{\pm\lrw_{k,k'}[E_{\_}],\ket{\pm E}} = 1\right]  - \Pr_{\substack{\widetilde{\Pi} \leftarrow  \mathcal{E}(\tweak,n); \\ E \leftarrow  \mathcal{E}(m,n)}} \left[ \A^{\pm\widetilde{\Pi},\ket{\pm E}} = 1\right] \right|  \\ 
       \leq \frac{6 q_C^2}{2^n} + \frac{4}{2^{(m+n)/2}}\left( q_C \sqrt{q_Q}+q_Q \sqrt{q_C}\right).
    \end{align*}
\end{theorem}

\begin{proof}
    The proof follows a procedure similar to that of the post-quantum security proof for \fx (\expref{Theorem}{thm:Q1-secure-FX}). However, in \lrw, we introduce bad events related to collisions, defined as $\textsf{Bad}_j = \textsf{Bad}_{j,1} \land \textsf{Bad}_{j,2}$, where $j$ indicates the $j^{\textsf{th}}$ stage of the hybrid:
    \begin{itemize}
        \item $\textsf{Bad}_{j,1}$: A collision occurs in the set $\{x_1 \oplus h_{k'}(\tau_1), \ldots, x_j \oplus h_{k'}(\tau_j)\}$,
        \item $\textsf{Bad}_{j,2}$: A collision occurs in the set $\{y_1 \oplus h_{k'}(\tau_1), \ldots, y_j \oplus h_{k'}(\tau_j)\}$.
    \end{itemize}
    Using the hybrid technique, we analyze under the condition that these bad events do not occur:
    \begin{align*}
        &\left| \Pr[\A(\Hyb_0)=1] - \Pr[\A(\Hyb_{q_C})=1] \right| \\
        &= \left| \Pr[\A(\Hyb_0)=1] - \Pr[\A(\Hyb_{q_C})=1 \land \neg \textsf{Bad}_{q_C}] - \Pr[\A(\Hyb_{q_C})=1 \land \textsf{Bad}_{q_C}] \right| \\
        &\leq \left| \Pr[\A(\Hyb_0)=1] - \Pr[\A(\Hyb_{q_C})=1 \land \neg \textsf{Bad}_{q_C}] \right| + \Pr[\textsf{Bad}_{q_C}] \\
        &\leq \sum_{j=0}^{q_C-1} \left| \Pr[\A(\Hyb_j)=1 \land \neg \textsf{Bad}_j] - \Pr[\A(\Hyb_{j+1})=1 \land \neg \textsf{Bad}_{j+1}] \right| + \Pr[\textsf{Bad}_{q_C}].
    \end{align*}
    This way to introducing bad events was first employed in \cite{hosoyamada2025post}. Detailed proof for \lrw is in \expref{Appendix}{app:lrw-proof}.
    \qed
\end{proof}

\subsection{Tightness} \label{sec:attacks}


We summarize the best-known attacks on standardized \lrw and \xext, with detailed presentations provided in \expref{Appendix}{app:attacks}:
\begin{itemize}
\item \textsf{Classical Distinguishing Attack:} The best-known distinguishing attacks on \lrw and \xext, which aim to differentiate them from a block cipher, exploit either birthday collisions or \emr-related structures. These attacks generally require only $O(2^{n/2})$ online queries. This matches the $q_C^2 \cdot 2^{-n}$ term in the security bound.
\item \textsf{Classical Complete Key Recovery Attack:} Distinguishing attacks are first employed to identify the tweak key, followed by an exhaustive key search for the cipher key. This process requires a total of $O(2^{n/2} + 2^m)$ online queries and $O(2^m)$ offline queries.
\item \textsf{Quantum Complete Key Recovery Attack (shorter cipher key):} 
The offline-Simon attack~\cite{bonnetain2019quantum} satisfies the trade-off
$q_C q_Q^2 = 2^{m+n}$. 
For block ciphers with shorter key lengths ($m \le 2n$), 
this attack, recognized as the best known attack, 
recovers the secret key using approximately 
$\widetilde{O}(2^{(m+n)/3})$ online (classical) 
and $\widetilde{O}(2^{(m+n)/3})$ offline (quantum) queries. 
Up to poly-logarithmic factors, this aligns with the term 
$\sqrt{q_C}\, q_Q \cdot 2^{(m+n)/2}$ in our security bound.

\item \textsf{Quantum Complete Key Recovery Attack (longer cipher key):} For longer key lengths ($m > 2n$), we can combine Grover's algorithm with the Kuwakado-Morii attack~\cite{kuwakado2012security}. The query complexity is $q_Q = O(2^{m/2 + n/3})$ and  $q_C = O(2^{n/3})$. Notably, this matches the offline-Simon attack's trade-off, satisfying $q_Q^2 q_C = O(2^{(m+n)/2})$. Alternatively, we can apply Grover's algorithm alone to recover both the cipher key and tweak key, resulting in $q_Q = O(2^{(m+n)/2})$. The choice of attack depends on whether we prioritize minimizing quantum or classical queries. Similarly, it aligns with the term 
$\sqrt{q_C}\, q_Q \cdot 2^{(m+n)/2}$ in our security bound.
\end{itemize}

\section{Block Cipher Modes}
\label{sec:BCM}
\subsection{Definitions}
\begin{definition} [Distinguishing Advantage] \label{def:dis-adv}
    Let $E: \{0,1\}^m \times \{0,1\}^n \rightarrow \{0,1\}^n$ be an efficient, keyed permutation. Consider an adversary $\A$ (which can be either quantum or classical). The advantage of $\A$ in attacking $E$ is measured by
    \begin{align*}
        \textsf{\textsf{Adv}}_E(\A):=\left|\Pr_{k \leftarrow\{0,1\}^m}[\A^{E_k(\cdot),E^{-1}_k(\cdot)}=1] - \Pr_{P \leftarrow \mathcal{P}_n}[\A^{P(\cdot),P^{-1}(\cdot)}=1]\right|,
    \end{align*}
    where the probabilities are also taken over the randomness of $\A$.
\end{definition}

For the rest of the paper, we define $T_\A$ as the time an adversary $\A$ spends and $Q_\A$ as the number of \textbf{classical} queries that $\A$ makes to the oracles.  

\begin{definition} [\textsf{SPRP(-PQ)} Security] \label{def:SPRP(-Q1)-security}
    Let $E:\{0,1\}^m\times \{0,1\}^n \rightarrow \{0,1\}^n$ be an efficient keyed permutation. The \textsf{(PQ) strong pseudorandom permutation} security of $E$ is measured by the maximum advantage over all (quantum) adversaries $\A$:
    \begin{align*}
        \textsf{Adv}_E^{\textsf{SPRP(-PQ)}}(q,t)&:=\max_{\substack{\A: \, T_{\A}\leq t;\\ \ \ Q_{\A} \leq q }}\textsf{Adv}_E(\A) \\
        &=\max_{\substack{\A: \, T_{\A}\leq t;\\ \ \ Q_{\A} \leq q }}\left|\Pr_{k \leftarrow\{0,1\}^m}[\A^{E_k(\cdot),E^{-1}_k(\cdot)}=1] - \Pr_{P \leftarrow \mathcal{P}_n}[\A^{P(\cdot),P^{-1}(\cdot)}=1]\right|,
    \end{align*}
    where the probabilities are also taken over the randomness of $\A$.
\end{definition}

\begin{definition} [\textsf{PRF(-PQ)} Security] \label{def:PRF(-Q1)-security}
    Let $E:\{0,1\}^m\times \{0,1\}^* \rightarrow \{0,1\}^*$ be an efficient keyed function. The \textsf{(PQ) pseudorandom function} security of $E$ is measured by the maximum advantage over all (quantum) adversaries $\A$: 
    \begin{align*}
        \textsf{Adv}_E^{\textsf{PRF(-PQ)}}(q,t)&:=\max_{\substack{\A: \, T_{\A}\leq t;\\ \ \ Q_{\A} \leq q }} \left|\Pr_{k \leftarrow\{0,1\}^m}[\A^{E_k(\cdot)}=1] - \Pr_{f}[\A^{f(\cdot)}=1]\right|,
    \end{align*}
    where $f$ is chosen uniformly from the set of functions mapping $\{0,1\}^*$ to $\{0,1\}^*$; the probabilities are also taken over the randomness of $\A$. 
\end{definition}

\subsection{\qprp Security of an Ideal Cipher} \label{sec:q1-sec-ic}

We consider quantum information-theoretic adversaries that are not constrained by computational resources, such as time or the number of available qubits. The only restriction is on the number of queries that the adversary can make to its oracles.
\begin{definition} \label{def:q1-sec-ic}
    Let $E$ be the ideal cipher in the quantum ideal cipher model, and let $\A$ be a quantum adversary. The advantage of $\A$ in attacking $E$ is measured by:
    \begin{align*}
        \textsf{Adv}_{E}(\A) := \left| \Pr_{k \leftarrow \{0,1\}^m} \left[ \A^{E_k(\cdot), E_k^{-1}(\cdot), \ket{E(\cdot, \cdot)},\ket{E^{-1}(\cdot, \cdot)}} = 1 \right] - \Pr_{P \leftarrow \mathcal{P}_n} \left[ \A^{P, P^{-1}, \ket{E},\ket{E^{-1}}} = 1 \right] \right|.
    \end{align*}
\end{definition}

For simplicity of presentation, we sometimes later will omit the inverse oracles when considering $\sprp\textsf{(-PQ)}$-security. Similarly, \qsprp security of an ideal cipher $E$ is the maximum advantage taken over all possible quantum adversaries $\A$:
\begin{align*}
    \textsf{Adv}_{E}^{\qsprp}(q_1, q_2) := \max_{\substack{\A: \, Q_{1,\A} \leq q_1;\\ \ \ Q_{2,\A} \leq q_2}} \textsf{Adv}_{E}(\A),
\end{align*}
where $Q_{1,\A}$ and $Q_{2,\A}$ represent the queries that $\A$ makes to $E^{\pm}_k(\cdot)$ and $\ket{E^{\pm}(\cdot,\cdot)}$, respectively.

We emphasize that no block cipher can offer greater security than an ideal cipher. Now we start examining the \qsprp security of an ideal cipher. To establish a bound on the distinguishing probability, we reduce to the \textsf{UNIQUE-SEARCH} problem from the distinguishing problem.

\begin{definition} ($\textsf{UNIQUE-SEARCH}_n$)
  Given a function $f: \{0,1\}^n \rightarrow  \{0,1\}$, such that $f$ maps at most one element to $1$, output \textsf{YES} if $f^{-1}(1)$ is non-empty and \textsf{NO} otherwise. 
\end{definition}

\begin{theorem} \label{thm:security-ideal-cipher}
    Let $E$ be an ideal cipher (in the quantum ideal cipher model) and assume an adversary $\A$ making $q_C$ classical queries to $E_k$ and $q_Q$ quantum queries to $E$. Then it holds that 
    \begin{align*}
        \textsf{Adv}^{\qsprp\textsf{-IC}}_E(q_C,q_Q) &=\max_{\A} \left|\Pr_{k \leftarrow\{0,1\}^m}[\A^{E_k,\ket{E}}=1] - \Pr_{P \leftarrow \mathcal{P}_n}[\A^{P,\ket{E}}=1]\right|\\
        &\leq \max_{\A'} \textsf{Adv}^{\textsf{UNIQUE-SEARCH}}_{f}(\A').
    \end{align*}
\end{theorem}
\begin{proof}
    Given an adversary $\A$ that attacks $E$, we now construct an adversary $\A'^{\ket{f}}$ attacking \textsf{UNIQUE-SEARCH} problem:
    \begin{enumerate}
        \item $\A'$ generates a permutation $P \in \mathcal{P}_n$ and an ideal block cipher $E \in \mathcal{E}(m,n)$. 
        \item $\A'$ then constructs a function $E'$ by querying its own oracle $f$: 
        $$
        E'(k,\tau)= \begin{cases}P(\tau) & \text { if } f(k)=1 \\ E(k,\tau) & \text { if } f(k) \neq 1\end{cases}.
        $$
        $\A'$  subsequently provides $\A$ with classical oracle access to $P,P^{-1}$, as well as quantum oracle access to $E',E'^{-1}$. The construction of these quantum oracles is detailed in \expref{Section}{sec:oracle-for-E-prime}. 
        \item $\A'$ outputs what $\A$ outputs.
    \end{enumerate}
    
    We now argue that if $\A$ successfully distinguishes between having oracle access to $(P, \ket{E})$ or $(E_k, \ket{E})$, then $\A'$ can successfully attack the \textsf{UNIQUE-SEARCH} problem, specifically determining whether $f^{-1}$ is empty or not. We will present our argument by considering the following cases:

    \medskip \noindent\noindent \textbf{\textsf{YES} case.} There exists a $k^*$ such that $f(k^*) = 1$, which implies that $E'(k^*, \tau) = P(\tau)$ for any $t$. From the perspective of $\A$, for this specific $k^*$, the behavior of the second oracle $E'$ when queried with $(k^*, \cdot)$ matches the behavior of the first oracle $P$ when queried with $(\cdot)$. This situation is equivalent to what $\A$ expects to observe in the interaction with the oracles $(E_k,\ket{E})$, i.e. $\A^{E_k, \ket{E}}$.

    \medskip \noindent\noindent \textbf{\textsf{NO} case.} There does not exist any $k$ such that $f(k) = 1$, which implies that for all $k$ and any $\tau$, we have $E'(k, \tau) = E(k, \tau)$. Consequently, from the perspective of $\A$, the interaction with the oracles $(P,\ket{E'})$ is just $(P,\ket{E})$, which gives $\A^{P, \ket{E}}$.

    Therefore, if $\A$ can distinguish the real world $\left(\A^{E_k, \ket{E}}\right)$ and the ideal world $\left(\A^{P, \ket{E}}\right)$, then $\A'$ can distinguish \textsf{YES} and \textsf{NO} case in \textsf{UNIQUE-SEARCH}:
    \begin{align*}
        \textsf{Adv}^{\textsf{UNIQUE-SEARCH}}_{f}(\A') \geq \textsf{Adv}_{E}(\A).
    \end{align*}
    We note that the number of queries $\A'$ makes to $f$ is determined only by $q_Q$, as specified by the construction of $E'$.
    \qed
\end{proof}

The \textsf{UNIQUE-SEARCH} bound against quantum adversaries~\cite{bennett1997strengths} yields:
\begin{corollary} \label{thm:pq-security-ideal-cipher}
    Let $E$ be an ideal cipher (in QICM) and assume an adversary $\A$ making $q_C$ queries to $E_k$ and $q_Q$ \textbf{quantum} queries to $E$. Then it holds that 
    \begin{align*}
        \textsf{Adv}^{\qsprp\textsf{-IC}}_E(q_C,q_Q) &\leq \frac{q_Q^2}{2^m}.
    \end{align*}
\end{corollary}

We note that the bound is optimal, and Grover's search is the matching algorithm, which involves querying $E$ $2^{m/2}$ times to find the key. It is evident from the bound that only $q_Q$, i.e., the number of queries to $E$, is involved. This implies that the security lower bound stays the same in the \textsf{Q2} model, in which the adversary can query both of its oracles quantumly.

\subsection{Security of Block Cipher Modes}

\subsubsection{General Theorem}

\textbf{Construction $\textsf{Con}$.} We begin by defining a polynomial-time algorithm, denoted as $\textsf{Con}$, which uses a (black-box) permutation as input to produce a specific construction. This permutation can take on different forms: when the construction is instantiated with a specific block cipher $E$, we denote it as $\textsf{Con}[E_{\_}]$. Then with a key $k \leftarrow \{0,1\}^m$ is uniformly sampled first, $E_k$ then serves as the (black-box) permutation, denoted as $\textsf{Con}[E_k]$. For $\textsf{Con}$ that is invertible, we denote its inverse as $\textsf{DeCon}$ and with different underlying permutations, it is denoted as $\textsf{DeCon}\left[E_k^{-1}\right]$ and $\textsf{DeCon}\left[P^{-1}\right]$.

To illustrate, consider $\textsf{Con}$ as the $\cbc\textsf{-MAC}$ construction. When the permutation is instantiated with a block cipher $E$, the construction first samples a uniformly random key $k \leftarrow \{0,1\}^m$ and uses it as the block cipher key for $E$, which gives $\textsf{Con}[E_k] = \cbc\textsf{-MAC}[E_k]$. The output of this construction is given by:
\begin{align*}
    \cbc\textsf{-MAC}[E_k](x_1,\ldots,x_{\ell}) = E_k(E_k(\ldots E_k(E_k(x_1)\oplus x_2) \oplus \ldots)\oplus x_{\ell}).
\end{align*}
Alternatively, when the permutation-like oracle is instantiated as a true random permutation $P$, the construction is denoted by $\textsf{Con}[P] = \cbc\textsf{-MAC}[P]$. In this scenario, the block cipher $E_k$ is replaced entirely by the permutation $P$ within the construction, resulting in the following definition:
\begin{align*}
    \cbc\textsf{-MAC}[P](x_1,\ldots,x_{\ell}) = P(P(\ldots P(P(x_1)\oplus x_2) \oplus \ldots)\oplus x_{\ell}).
\end{align*}

\noindent \textbf{Experiment $\textsf{Exp}$.} The security of a cryptographic scheme $\textsf{Con}$ is typically analyzed through a well-defined probabilistic experiment, denoted as $\textsf{Exp}$. This experiment captures the interaction between an adversary $\A$ and a challenger, modeled as a classical polynomial-time algorithm. The challenger has access to certain resources, such as a black-box permutation oracle, which it can query as needed during the experiment. The goal of the adversary $\A$ is to exploit these interactions and the oracle queries to either compromise a specific security property of $\textsf{Con}$ or extract secret information. The experiment $\textsf{Exp}(\A)$ proceeds as follows: 
\begin{enumerate}
    \item \textbf{Initialization and Oracle Access}: The challenger $\textsf{Chall}$ is possible to be given access to one or more oracles, such as black-box permutation oracles, which it may use to process information or respond to the adversary's queries.  $\textsf{Chall}$ may also grant the adversary direct access to some oracles under predefined rules. $\textsf{Con}$ can also obtain oracle access to these oracles. We assume that the number of queries to these oracles only depends on the number of queries the adversary makes to the construction \textsf{Con}.
    \item \textbf{Interaction with the Adversary}: The challenger $\textsf{Chall}$ and the adversary $\A$ engage in a protocol where information is exchanged. The adversary may send queries or responses to the challenger, aiming to gather information or manipulate the interaction to achieve its goals. The interaction is always classical.
    \item \textbf{Output of the Experiment}: After a sequence of interactions and oracle queries, the challenger outputs a single bit $b \in \{0, 1\}$. This bit represents the outcome of the experiment, where its distribution reflects the adversary’s success in breaking the scheme $\textsf{Con}$ or achieving a specific objective.
\end{enumerate}
The success of the adversary $\A$ is quantified by the probability that it achieves a favorable outcome in $\textsf{Exp}$. This is measured in terms of the adversary's \textbf{advantage}, denoted as $\textsf{Adv}^{\textsf{Exp}}_{\textsf{Con}}(\A)$. The advantage reflects the extent to which the adversary outperforms random guessing in achieving its goals. Finally, the security of the scheme $\textsf{Con}$ is determined by evaluating the maximum advantage achievable by any adversary $\A$. This is expressed as:
\begin{align*}
\textsf{Adv}^{\textsf{Exp}}_{\textsf{Con}}(\cdot) = \max_{\mathcal{A}} \textsf{Adv}^{\textsf{Exp}}_{\textsf{Con}}(\mathcal{A}).
\end{align*}

Here are some example experiments. The \prf experiment is a standard framework to evaluate the pseudorandom function \prf security of a construction $\textsf{Con}$. This experiment proceeds as follows:
\begin{enumerate}
    \item \textbf{Initialization and Oracle Access}: The challenger initializes the experiment by setting up the construction $\textsf{Con}_k(\cdot)$, where $k \in \{0, 1\}^\kappa$ is a secret key sampled uniformly at random. A bit $d \xleftarrow[]{\$} \{0,1\}$ is sampled uniformly at random. This bit determines the behavior of the oracle provided to the adversary $\A$:
    \begin{itemize}
        \item If $d = 0$, the challenger grants $\A$ access to the oracle implementing the function $\textsf{Con}_k(\cdot)$.
        \item If $d = 1$, for each query $m$ from $\A$, the challenger responds with a value $y \xleftarrow[]{\$} \mathcal{Y}$ sampled uniformly at random from the output range $\mathcal{Y}$, independent of $m$, mimicking a random oracle.
    \end{itemize}
    \item \textbf{Interaction with the Adversary}: The adversary $\A$ is allowed to interact with the oracle provided by the challenger. It may issue a sequence of queries $m$ and receive corresponding responses, either from $\textsf{Con}_k(\cdot)$ (if $d=0$) or the random oracle (if $d=1$).
    \item \textbf{Output of the Experiment}: After a sequence of interactions, $\A$ outputs a bit $d'$, representing its guess about whether the oracle is $\textsf{Con}_k(\cdot)$ ($d=0$) or a random oracle ($d=1$). The challenger then computes the output of the experiment as a bit $b = \overline{d \oplus d'}$, where $b = 1$ indicates that $\A$ correctly distinguished the oracle type and $b = 0$ indicates that $\A$ guessed incorrectly.
\end{enumerate}
The adversary's goal is to maximize the probability of outputting the correct guess, $d' = d$, thus making $b = 1$. The \prf security of the construction $\textsf{Con}$ is defined as the maximum distinguishing advantage over all adversaries $\A$ running within the given resource constraints ($q$), denoted as $\textsf{Adv}^{\prf}_{\textsf{Con}}(q)$.

The distinguishing advantage of the adversary $\A$ in the \prf experiment is defined as:
\begin{align*}
    \textsf{Adv}^{\prf}_{\textsf{Con}}(\A) = \left| \Pr[b = 1] - \frac{1}{2} \right|.
\end{align*}
The \prf security of the construction $\textsf{Con}$ is then defined as the maximum distinguishing advantage over all adversaries $\A$ running within the given resource constraints ($q,t$):
\begin{align*}
    \textsf{Adv}^{\prf}_{\textsf{Con}}(q,t) = \max_{\A} \textsf{Adv}^{\prf}_{\textsf{Con}}(\A).
\end{align*}

Consider another example, the \forgery experiment.
The \forgery experiment evaluates the unforgeability of a construction $\textsf{Con}$, demonstrating that $\textsf{Con}$ is a \mac (Message Authentication Code). It is structured as follows:
\begin{enumerate}
    \item \textbf{Initialization and Oracle Access}: The challenger initializes the experiment by setting up the construction $\textsf{Con}_k(\cdot)$, where $k \in \{0, 1\}^\kappa$ is a secret key sampled uniformly at random. The oracle $\textsf{Con}_k(\cdot)$ is made available to the adversary $\A$.
    \item \textbf{Interaction with the Adversary}: The adversary $\A$ interacts with the oracle $\textsf{Con}_k(\cdot)$. On a query $m$, the oracle computes the tag $\textsf{tag}= \textsf{Con}_k(m)$ and returns $\textsf{tag}$ to $\A$.
    \item \textbf{Output of the Experiment}: After a sequence of interactions, $\A$ outputs a pair $(m,\textsf{tag})$. Let $\mathcal{Q}$ denote the set of all queries that $\A$ asked its oracle. The challenger outputs $b=1$ if and only if $\textsf{Con}_k(m)=\textsf{tag}$ and $m \not \in \mathcal{Q}$, the adversary produces a valid tag for a message $m$ that it has not queried during the interaction. Otherwise, the challenger outputs $0$.
\end{enumerate}
The adversary’s goal is to maximize the probability of producing a valid, unqueried message-tag pair $(m, \textsf{tag})$ such that the experiment outputs $1$. The unforgeability of the construction $\textsf{Con}$ is quantified by the adversary's advantage:
\begin{align*}
\textsf{Adv}^{\forgery}_{\textsf{Con}}(\mathcal{A}) = \Pr[b=1].
\end{align*}
The unforgeability of $\textsf{Con}$ is then defined as the maximum advantage over all adversaries $\A$ constrained by a certain number of queries $q$ and time $t$:
\begin{align*}
\textsf{Adv}^{\forgery}_{\textsf{Con}}(q,t) = \max_{\A} \textsf{Adv}^{\forgery}_{\textsf{Con}}(\mathcal{A}).
\end{align*}

For \forgery experiment with decryption oracle access, denoted as $\forgery \pm$, it is structured as follows for a construction $\textsf{Con}$:
\begin{enumerate}
    \item \textbf{Initialization and Oracle Access}: The challenger $\textsf{Chall}$ initializes the experiment by setting up the construction $\textsf{Con}_k(\cdot)$ and $\textsf{DeCon}_k(\cdot)$, where $k \in \{0, 1\}^\kappa$ is a secret key sampled uniformly at random. Both oracles are made available to the adversary $\A$.
    \item \textbf{Interaction with the Adversary}: The adversary $\A$ interacts with the oracle $\textsf{Con}_k(\cdot)$ and $\textsf{DeCon}_k(\cdot)$. When $\textsf{Con}_k(\cdot)$ is queried on a query $m$, the oracle computes the tag $\textsf{tag} = \textsf{Con}_k(m)$ and returns $t$ to $\A$. When $\textsf{DeCon}_k(\cdot)$ is queried on a query $t$, the oracle computes the tag $m= \textsf{DeCon}_k(t)$ and returns $m$ to $\A$. 
    \item \textbf{Output of the Experiment}: Same as in \forgery experiment.
\end{enumerate}

\textbf{General Theorem.} We now focus on a collection of constructions whose security is analyzed in a typical reduction-based approach. The classical security of the construction $\textsf{Con}[E_{\_}]$, for a particular block cipher $E$ in a specific experiment $\textsf{Exp}$, is established through the following steps:
\begin{enumerate}
    \item \textbf{Reduction to the Block Cipher \sprp Security:} Replace the block cipher $E_k$ with a truly random permutation $P \xleftarrow[]{\$} \mathcal{P}_n$, incurring a loss equivalent to the \sprp security of $E$.
    \item \textbf{Information-Theoretic Security of $\sf{Con}[P]$:} When $E$ is replaced by $P$, it can be shown that the construction $\textsf{Con}[P]$ achieves a certain level of security in an information-theoretic sense.  This means that the security is determined solely by the number of queries $q$ made to the construction oracle $\textsf{Con}[P]$ and its inverse, independent of the adversary's computational power.
\end{enumerate}
Therefore, the final security of $\textsf{Con}[E_{\_}]$ in a given experiment $\textsf{Exp}$ against a classical probabilistic adversary making $q$ queries in time $t$ is given by:
\begin{align*}
    \textsf{Adv}^{\textsf{Exp}}_{\textsf{Con}[E_{\_}]}(q,t) \leq \textsf{Adv}_E^{\sprp}(q',t) + \delta(q),
\end{align*}
Here, $q' = f(q)$, a function of $q$, represents the number of queries made by $\textsf{Con}$ and $\textsf{Chall}$ to $E_{\_}(\cdot)$. Additionally, $t$ denotes the computational resources required for the reduction to the security of the underlying block cipher. 

This approach combines the underlying computational security of the block cipher $E$ and the inherent information-theoretic security of the construction when using a random permutation $P$. When $\textsf{Con}$ is instantiated with a specific block cipher $E$, we can extend the classical security analysis of $\textsf{Con}$ in a given experiment $\textsf{Exp}$ to establish its post-quantum security in the same experiment $\textsf{Exp}$ as follows:
\begin{theorem} \label{thm:general-qprp-theorem}
Let $\textsf{Exp}$ be a security experiment, and let $\textsf{Con}$ be a construction as defined above, instantiated with a block cipher $E$. Then the post-quantum security of $\textsf{Con}$ is bounded as follows:
\begin{align*}
    \textsf{Adv}^{\textsf{Exp-PQ}}_{\textsf{Con}[E_{\_}]}(q,t) \leq \textsf{Adv}_E^{\qsprp}(q',t) + \delta(q),
\end{align*}
where $q$ is the number of queries made to the keyed construction and its inverses, $\textsf{Con}[E_k](\cdot)$ and $\textsf{DeCon}\left[E_k^{-1}\right](\cdot)$; $q'=f(q)$ represents the number of queries made by $\textsf{Con}$ and the challenger $\textsf{Chall}$ to $E_{\_}(\cdot)$; $t$ denote the time resources available to quantum adversaries in the experiments $\textsf{Exp}$. $\delta(q)$ is  the classical-model information-theoretic security of $\textsf{Con}$ instantiated with an ideal block cipher. 
\end{theorem} 


\begin{proof}
    Suppose $\A$ is a quantum adversary attacking $\textsf{Con}[E_{\_}]$ in a specified experiment $\textsf{Exp}$. The corresponding security is:
    \begin{align*}
        \textsf{Adv}^{\textsf{Exp-PQ}}_{\textsf{Con}[E_{\_}]}(q,t) &= \max_{\A}\textsf{Adv}^{\textsf{Exp-PQ}}_{\textsf{Con}[E_{\_}]}(\A) \\
        &= \max_{\A}\left|\Pr_{k \leftarrow\{0,1\}^m}\left[\textsf{Exp}\left(\A^{\textsf{Con}[E_k],\textsf{DeCon}\left[E_k^{-1}\right]}\right)=1\right] - \Pr_{P \leftarrow \mathcal{P}_n}\left[\textsf{Exp}\left(\A^{P,P^{-1}}\right)=1\right]\right|.
    \end{align*}

    We then consider a special class of quantum adversaries, denoted as $\B$. As defined in \expref{Definition}{def:dis-adv}, $\B$ has classical oracle access to $\mathcal{O}$ and $\mathcal{O}^{-1}$, where $\mathcal{O}$ represents either $E_k(\cdot)$ or $P(\cdot)$. The goal of $\B$ is to simulate a complete interaction between a specific adversary $\A$ and the challenger in the experiment $\textsf{Exp}$ as follows:
    \begin{enumerate}
        \item \textbf{Construct Oracles:} $\B$ constructs the oracles $\textsf{Con}[\mathcal{O}]$ and $\textsf{DeCon}\left[\mathcal{O}^{-1}\right]$ based on its own access to $\mathcal{O}$ and $\mathcal{O}^{-1}$. 
        \item \textbf{Simulate Interaction:} $\B$ invokes the specific adversary $\A$ and runs it in the experiment $\textsf{Exp}$ while acting as the challenger; $\B$ also provides $\A$ with the constructed oracles $\textsf{Con}[\mathcal{O}]$ and $\textsf{DeCon}\left[\mathcal{O}^{-1}\right]$ when necessary.
        \item \textbf{Output Result:} $\mathcal{B}$ outputs a bit: $1$ if $\A$ succeeds in the experiment $\textsf{Exp}$, and $0$ otherwise.
    \end{enumerate}

     We note that for certain specific experiments $\textsf{Exp}$, where the adversary $\A$ is tasked with distinguishing whether the construction is based on $E$ or $P$, in the final step, $\B$ simply outputs whatever $\A$ outputs. From the reduction, it is clear that 
    \begin{align} \label{eq:general-eq-1}
        \textsf{Adv}_E(\B)&=\left|\Pr_{k \leftarrow\{0,1\}^m}\left[\B^{E_k(\cdot),E^{-1}_k(\cdot)}=1\right] - \Pr_{P \leftarrow \mathcal{P}_n}\left[\B^{P(\cdot),P^{-1}(\cdot)}=1\right]\right| \nonumber\\
        &\geq  \left|\Pr_{k \leftarrow\{0,1\}^m}\left[\textsf{Exp}\left(\A^{\textsf{Con}[E_k],\textsf{DeCon}\left[E_k^{-1}\right]}\right)=1\right] - \Pr_{k \leftarrow\{0,1\}^m}\left[\textsf{Exp}\left(\A^{\textsf{Con}[P],\textsf{DeCon}\left[P^{-1}\right]}\right)=1\right] \right|.
    \end{align}

    We also note that after replacing $E_k$ with $P$, the distinguishing advantage of $\A$ attacking $\textsf{Con}[E_k]$ with classical oracle access is given by
    \begin{align} \label{eq:general-eq-2}
         \left| \Pr_{k \leftarrow \{0,1\}^m}\left[\textsf{Exp}\left(\A^{\textsf{Con}[P],\textsf{DeCon}\left[P^{-1}\right]}\right)=1\right] - \Pr_{P \leftarrow \mathcal{P}_n}\left[\textsf{Exp}\left(\A^{P,P^{-1}}\right)=1\right] \right| \leq \delta(q),
    \end{align}
    which is information-theoretic. This indicates that even if $\A$ possesses local quantum capabilities, as long as it has only classical oracle access, the distinguishing advantage remains bounded by $\delta(q)$.

    By combining the above three equations, we obtain:
    \begin{align} \label{eq:general-eq-3}
        & \textsf{Adv}^{\textsf{Exp-PQ}}_{\textsf{Con}[E_{\_}]}(\A) \nonumber\\
        &\leq  \left|\Pr_{k \leftarrow\{0,1\}^m}\left[\textsf{Exp}\left(\A^{\textsf{Con}[E_k],\textsf{DeCon}\left[E_k^{-1}\right]}\right)=1\right] - \Pr_{k \leftarrow\{0,1\}^m}\left[\textsf{Exp}\left(\A^{\textsf{Con}[P],\textsf{DeCon}\left[P^{-1}\right]}\right)=1\right] \right| \nonumber \\
        & \quad +  \left| \Pr_{k \leftarrow \{0,1\}^m}\left[\textsf{Exp}\left(\A^{\textsf{Con}[P],\textsf{DeCon}\left[P^{-1}\right]}\right)=1\right] - \Pr_{P \leftarrow \mathcal{P}_n}\left[\textsf{Exp}\left(\A^{P,P^{-1}}\right)=1\right] \right| \nonumber \\
        &\leq \textsf{Adv}_E(\B) + \delta(q).
    \end{align}
    
    By taking the maximum over adversaries $\A$, which is equivalent to taking the maximum over adversaries $\B$, it follows that:
    \begin{align*}
        \textsf{Adv}^{\textsf{Exp-PQ}}_{\textsf{Con}[E_{\_}]}(q,t) &\leq \max_{\B}\textsf{Adv}_E(\B) + \delta(q)\\
        &= \textsf{Adv}^{\qsprp}_{E}(q',t) + \delta(q).
    \end{align*}
    \qed
\end{proof}

\noindent\textbf{Remark.} The current statement applies directly only to constructions $\textsf{Con}$ based on a block cipher $E$ with a single key. For constructions where $E$ uses multiple (say $c$) independently sampled keys, the argument can be extended with minor modifications:
\begin{align*}
    \textsf{Adv}^{\textsf{Exp}}_{\textsf{Con}[E_{\_}]}(q,t) \leq c \cdot \textsf{Adv}_E^{\sprp}(q',t) + \delta(q) \Longrightarrow
    \textsf{Adv}^{\textsf{Exp-PQ}}_{\textsf{Con}[E_{\_}]}(q,t) \leq c \cdot \textsf{Adv}_E^{\qsprp}(q',t) + \delta(q).
\end{align*}

We also examine the security of $\textsf{Con}[E_{\_}]$ in the QICM. As discussed in \expref{Section}{sec:resampling_ic}, in this model, $E$ is treated as an ideal cipher, and $\textsf{Con}[E_{\_}]$ is constructed using such $E$. Unlike the previous version, here, the challenger $\textsf{Chall}$ and $\textsf{Con}$ are only granted \textbf{classical} oracle access to the \textit{ideal cipher} $E(\cdot, \cdot)$, rather than access to the full scripts as in the previous version. The number of queries they make is recorded as $q'_C$. The adversary $\A$, however, is granted \textbf{quantum} oracle access to $E(\cdot,\cdot)$, recorded as $q_Q$, in addition to \textbf{classical} oracle access to $\textsf{Con}[E_k]$ and $\textsf{DeCon}\left[E_k^{-1}\right]$, recorded as $q_C$. It is important to note that, in this model, we consider adversaries with unbounded computational power. Thus, we are effectively analyzing the security of $\textsf{Con}[E_{\_}]$ in the post-quantum model.

\begin{theorem}\label{thm:general-qprp-theorem-ideal-cipher}
    Let $\textsf{Exp}$ be a security experiment, and let $\textsf{Con}$ be a construction as defined before. Then the post-quantum security of $\textsf{Con}[E_{\_}]$ in \textit{the ideal cipher model} is given by:
    \begin{align*}
    \textsf{Adv}^{\textsf{Exp-PQ-IC}}_{\textsf{Con}[E_{\_}]}(q_C,q_Q) &\leq  \textsf{Adv}_E^{\qsprp\textsf{-IC}}(q_C',q_Q) + \delta(q_C)\\
    &=\frac{q_Q^2}{2^m} + \delta(q_C)
\end{align*}
where $q_C,q'_C$ and $q_Q$ is defined as above; $t$ denote the time resources available to quantum adversaries in the experiments $\textsf{Exp}$. $\delta(q_C)$ is  the classical-model info-theoretic security of $\textsf{Con}$ instantiated with an ideal block cipher.
\end{theorem}

\begin{proof}
    Suppose $\A$ is a quantum adversary attacking $\textsf{Con}[E_{\_}]$ in a specified experiment $\textsf{Exp}$ in the quantum ideal cipher model. The corresponding security is:
    \begin{align*}
        &\textsf{Adv}_{\textsf{Con}[E_{\_}]}(q_C,q_Q) \\
        &= \max_{\A}\left|\Pr_{\substack{k \leftarrow\{0,1\}^m\\E \leftarrow \mathcal{E}(m,n)}}\left[\textsf{Exp}\left(\A^{\textsf{Con}[E_k],\textsf{DeCon}\left[E_k^{-1}\right],E}\right)=1\right] - \Pr_{\substack{P \leftarrow \mathcal{P}_n\\E \leftarrow \mathcal{E}(m,n)}}\left[\textsf{Exp}\left(\A^{P,P^{-1},E}\right)=1\right]\right|.
    \end{align*}

    We note that the presence of an additional oracle $E(\cdot,\cdot)$ does not affect the validity of \expref{Equation}{eq:general-eq-1}. Additionally, \expref{Equation}{eq:general-eq-2} also remains valid in the presence of the extra oracle $E$, as the oracles $\textsf{Con}[P]$, $\textsf{DeCon}\left[P^{-1}\right]$, $P$, $P^{-1}$ operate independently from $E$:
    \begin{align*}
         \Pr_{\substack{k \leftarrow\{0,1\}^m\\E \leftarrow \mathcal{E}(m,n)}}\left[\textsf{Exp}\left(\A^{\textsf{Con}[P],\textsf{DeCon}\left[P^{-1}\right],E}\right)=1\right] &= \Pr_{k \leftarrow\{0,1\}^m}\left[\textsf{Exp}\left(\A^{\textsf{Con}[P],\textsf{DeCon}\left[P^{-1}\right]}\right)=1\right]\\
         \Pr_{\substack{P \leftarrow \mathcal{P}_n\\E \leftarrow \mathcal{E}(m,n)}}\left[\textsf{Exp}\left(\A^{P,P^{-1},E}\right)=1\right] &= \Pr_{P \leftarrow \mathcal{P}_n}\left[\textsf{Exp}\left(\A^{P,P^{-1}}\right)=1\right].
    \end{align*}

    Therefore \expref{Equation}{eq:general-eq-3} holds in the quantum ideal cipher model. Then combining with the security of the ideal cipher in the quantum ideal cipher model as stated in \expref{Theorem}{thm:pq-security-ideal-cipher}:
    \begin{align*}
        \textsf{Adv}^{\textsf{Exp-PQ-IC}}_{\textsf{Con}[E_{\_}]}(q_C,q_Q) &= \textsf{Adv}^{\qsprp\textsf{-IC}}_{E}(q'_C,q_Q) + \delta(q_C)\\
        &\leq \frac{q_Q^2}{2^m} + \delta(q_C). 
    \end{align*}
    \qed
\end{proof}

The theorem also applies to constructions $\textsf{Con}$ whose classical security can be reduced to the block cipher's \prp security. This yields the following bound:
\begin{align*}
\textsf{Adv}^{\textsf{Exp-PQ}}_{\textsf{Con}[E_{\_}]}(q,t) \leq \textsf{Adv}_E^{\qprp}(q',t') + \delta(q).
\end{align*}
Similarly, in QICM, we can derive:
\begin{align*}
\textsf{Adv}^{\textsf{Exp-PQ-IC}}_{\textsf{Con}[E_{\_}]}(q_Q,q_C) \leq \frac{q_Q^2}{2^m} + \delta(q_C).
\end{align*}

\subsubsection{Applications of \expref{Theorem}{thm:general-qprp-theorem} and \expref{Theorem}{thm:general-qprp-theorem-ideal-cipher}.} \label{sec:general-thm-app} 

Let $m$ and $n$ be positive integers. Let $E: \{0,1\}^m \times \{0,1\}^{n} \rightarrow \{0,1\}^n$ be a block cipher. A summary is provided in \expref{Appendix}{app:modes-table}. 

\medskip \noindent\noindent\textbf{\cbc~\cite{ehrsam1978message}.}
\begin{itemize}
    \item Construction:
    $\textsf{Con}[E_k] = \cbc^{E}_{k}(x_1,\ldots,x_{\ell}) = E_k(E_k(\ldots E_k(E_k(x_1)\oplus x_2) \oplus \ldots)\oplus x_{\ell})$
    \item Classical security~\cite{katz2007introduction}: For all probabilistic adversaries $\A$, it holds that:
   \begin{align*}
       \textsf{Adv}^{\prf}_{\cbc}(q,t) &:= \max_{\A} \left|\Pr_{k \leftarrow\{0,1\}^m}[\A^{\cbc[E_k]}=1] - \Pr_f[\A^{f}=1]\right|\\
       &\leq \textsf{Adv}^{\prp}_{E}(q\ell, t) + \frac{q^2\ell^2}{2^n},
   \end{align*}
   where $f$ is chosen uniformly from the set of functions mapping $(\{0,1\}^n)^{\ell}$ to $\{0, 1\}^n$, $q$ is the number of queries to \cbc, $t$ is the time complexity and $\ell$ is the number of blocks of the longest query input.
    \item Post-quantum security: The post-quantum security of a construction based on another concrete block cipher, such as \textsf{AES}, in the $\textsf{PQ}$ model is closely related to its classical security since no second oracle is available to the adversary. The equivalence holds with $\textsf{Adv}^{\prf}_E$ replaced by $\textsf{Adv}^{\qprp}_E$ as stated in \expref{Theorem}{thm:general-qprp-theorem}:
    \begin{align*}
        \textsf{Adv}_{\cbc}^{\prf\textsf{-PQ}}(q,t) \leq \textsf{Adv}_E^{\qprp}(q\ell,t)+\frac{q^2\ell^2}{2^n}.
    \end{align*}
    \item Post-quantum security in QICM: By applying \expref{Theorem}{thm:general-qprp-theorem-ideal-cipher}, we get
    \begin{align*}
        \textsf{Adv}_{\cbc}^{\prf\textsf{-PQ-IC}}(q_C,q_Q) &= \max_{\A}\left| \Pr_{k \leftarrow \{0,1\}^m} \left[ \A^{\cbc[E_k],E} = 1\right]  - \Pr_{f} \left[ \A^{P,E} = 1\right] \right| \\
        & \leq \frac{q_Q^2\ell^2}{2^m}+\frac{q_C^2\ell^2}{2^n}.
    \end{align*}
\end{itemize}

Similarly as explained in~\cite{katz2007introduction}, when the tag is set to be $\cbc_k(m)$, this is exactly \cbc\textsf{-MAC}, and post-quantum security of \cbc implies that basic \cbc\textsf{-MAC} is post-quantum secure for messages of any fixed length. However, it is not inherently secure for variable-length messages. To address this limitation, \textit{encrypt-last-block} offers a method to extend \cbc\textsf{-MAC} for use with variable-length inputs. It is formally called \textsf{Encrypt-last-block CBC-MAC}, i.e., \ecbc\textsf{-MAC}.

\medskip \noindent\noindent\textbf{\ecbc~\cite{petrank2000cbc}.}
\begin{itemize}
    \item Construction:
    $\textsf{Con}[E_k](x) = \ecbc^{E}_{k_1,k_2}(x) = E_{k_2}(\cbc^E_{k_1}(x))$
    \item Classical security~\cite{petrank2000cbc,vaudenay2000decorrelation}: For all probabilistic  adversaries $\A$, it holds that:
   \begin{align*}
       \textsf{Adv}^{\prf}_{\ecbc}(q,t) &:= \max_{\A} \left|\Pr_{k \leftarrow\{0,1\}^m}[\A^{\ecbc[E_k]}=1] - \Pr_f[\A^{f}=1]\right|\\
       &\leq 2\textsf{Adv}^{\prp}_{E}(q\ell,t) + \frac{4q^2\ell^2}{2^n}.
   \end{align*}
   where $f$ is chosen uniformly from the set of functions mapping $(\{0,1\}^n)^{\ell}$ to $\{0, 1\}^n$, $q$ is the number of queries to \cbc, $t$ is the time complexity and $\ell$ is the number of blocks of the longest query input.
    \item Post-quantum security: By applying \expref{Theorem}{thm:general-qprp-theorem}, we get
    \begin{align*}
        \textsf{Adv}_{\cbc}^{\prf\textsf{-PQ}}(q,t)\leq 2 \textsf{Adv}_E^{\qprp}(q\ell,t)+\frac{4q^2\ell^2}{2^n}.
    \end{align*}
    \item Post-quantum security in QICM: By applying \expref{Theorem}{thm:general-qprp-theorem-ideal-cipher}, we get
    \begin{align*}
        \textsf{Adv}_{\cbc}^{\prf\textsf{-PQ-IC}}(q_C,q_Q) &= \max_{\A}\left| \Pr_{k \leftarrow \{0,1\}^m} \left[ \A^{\cbc[E_k],E} = 1\right]  - \Pr_{f} \left[ \A^{P,E} = 1\right] \right| \\
        & \leq \frac{2q_Q^2\ell^2}{2^m}+\frac{4q_C^2\ell^2}{2^n}.
    \end{align*}
\end{itemize}

\medskip \noindent\textbf{\cmac~\cite{iwata2003omac}.}
\begin{itemize}
    \item Construction:
    $\textsf{Con}[E_k](x) = \cmac^{E}_{k}(x) = E_{k_1}(\cbc^E_{k_1}(x_1\|\ldots \| x_{\ell-1})\oplus x_{\ell} \oplus k_2)$
    \item Classical security~\cite{iwata2003omac}: For all probabilistic adversaries $\A$, it holds that:
   \begin{align*}
       \textsf{Adv}^{\forgery}_{\cmac}(q,t) \leq \textsf{Adv}^{\prp}_{E}(q\ell+1,t+O(q\ell)) + \frac{5(\ell^2+1)q^2}{2^n}.
   \end{align*}
   where $q$ is the number of queries in the \forgery experiment, $t$ is the time complexity and $\ell$ is the number of blocks of the longest query input.
    \item Post-quantum security: By applying \expref{Theorem}{thm:general-qprp-theorem}, we get
    \begin{align*}
        \textsf{Adv}_{\cmac}^{\forgery\textsf{-PQ}}(q,t)\leq \textsf{Adv}_E^{\qprp}(q\ell+1,t+O(q\ell))+\frac{5(\ell^2+1)q^2}{2^n}.
    \end{align*}
    \item Post-quantum security in QICM: By applying \expref{Theorem}{thm:general-qprp-theorem-ideal-cipher}, we get
    \begin{align*}
        \textsf{Adv}_{\cmac}^{\forgery\textsf{-PQ-IC}}(q_C,q_Q) \leq \frac{(q_Q\ell+1)^2}{2^m}+\frac{5(\ell^2+1)q^2}{2^n}.
    \end{align*}
\end{itemize}

\medskip \noindent\textbf{\gcm~\cite{mcgrew2004security}.} 

\begin{itemize}
    \item Classical security~\cite{mcgrew2004security,iwata2012breaking,niwa2015gcm}: We begin by noting that the oracle in the \forgery experiment for \gcm also includes a decryption oracle, denoted as $\gcm \pm$.  For all probabilistic adversaries $\A$, it holds:
   \begin{align*}
       \textsf{Adv}^{\prf}_{\gcm}(q,t) &\leq \textsf{Adv}^{\prp}_{E}(q,t) + \frac{(\sigma+q+1)^2}{2^{n+1}}+ \frac{\sigma+q}{2^{n-1}}\\
       \textsf{Adv}^{\pm\forgery}_{\gcm}(q,t) &\leq \textsf{Adv}^{\prp}_{E}(q,t) + \frac{(\sigma+q+q'+1)^2}{2^{n+1}}+ \frac{\sigma+q+q'}{2^{n-1}} + \frac{q'(\ell+1)}{2^s}.
   \end{align*}
   where $q,q'$ are the number of queries made to the encryption and decryption oracles, respectively; $t$ is the time complexity and $s$ is the authentication tag size. Additionally $\ell$ is the maximum number of blocks in the input-data for any individual query, and $\sigma$ is the total number of blocks of input-data across all queries.
   
   From \cite{rogaway2002authenticated}, an AEAD-scheme $\Pi$ is said to be “secure” if $\textsf{Adv}^{\prf}_{\Pi}$ and $\textsf{Adv}^{\forgery}_{\Pi}$ are “small” for any “reasonable” adversary $\A$. \gcm is a widely used secure AEAD-scheme.
    \item Post-quantum security: By applying \expref{Theorem}{thm:general-qprp-theorem}, we get
    \begin{align*}
        \textsf{Adv}_{\gcm}^{\prf\textsf{-PQ}}(q,t) &\leq \textsf{Adv}_E^{\qprp}(q,t) + \frac{(\sigma+q+1)^2}{2^{n+1}}+ \frac{\sigma+q}{2^{n-1}}\\
        \textsf{Adv}^{\pm\forgery\textsf{-PQ}}_{\gcm}(q,t) &\leq \textsf{Adv}^{\qprp}_{E}(q,t)+ \frac{(\sigma+q+q'+1)^2}{2^{n+1}}+ \frac{\sigma+q+q'}{2^{n-1}} + \frac{q'(\ell+1)}{2^s}.
    \end{align*}
    \item Post-quantum security in QICM: By applying \expref{Theorem}{thm:general-qprp-theorem-ideal-cipher}, we get
    \begin{align*}
        \textsf{Adv}_{\gcm}^{\prf\textsf{-PQ-IC}}(q_C,q_Q)  &\leq \frac{q_Q^2\ell^2}{2^m}  + \frac{(\sigma+q_C+1)^2}{2^{n+1}}+ \frac{\sigma+q_C}{2^{n-1}}\\
        \textsf{Adv}^{\pm\forgery\textsf{-PQ-IC}}_{\gcm}(q_C,q_Q)  &\leq \frac{q_Q^2\ell^2}{2^m}  + \frac{(\sigma+q_C+q_C'+1)^2}{2^{n+1}}+ \frac{\sigma+q_C+q_C'}{2^{n-1}} + \frac{q_C'(\ell+1)}{2^s}.
    \end{align*}
\end{itemize}

\medskip \noindent\textbf{\textsf{GCM-SST}}~\cite{campagna2023galois,inoue2024generic}. When analyzing security of \textsf{GCM-SST}, the experiment considers Nonce-Misuse Resilience for both privacy and authenticity, denoted as \textsf{nml-PRF} and \textsf{nml-Forgery}, respectively. The \textsf{nml} security notions give stronger notions of security, where \gcm fails to meet the requirements, but \textsf{GCM-SST} satisfies them. 
\begin{itemize}
    \item Classical security~\cite{inoue2024generic}: For all probabilistic adversaries $\A$, it holds that:
   \begin{align*}
       \textsf{Adv}^{\textsf{nml-PRF}}_{\gcm}(q,t) &\leq \textsf{Adv}^{\prp}_{E}(q,t) + \frac{(\sigma+3q)^2}{2^{n+1}}\\
       \textsf{Adv}^{\textsf{nml-Forgery}}_{\gcm}(q,t) &\leq \textsf{Adv}^{\prp}_{E}(q,t) + \frac{(\sigma+3(q+q'))^2}{2^{n+1}}+ \frac{q'\ell}{2^{n}} + \frac{q'}{2^s}.
   \end{align*}
   where $q,q'$ are the number of queries made to the encryption and decryption oracles, respectively; $t$ is the time complexity and $s$ is the authentication tag size. Additionally $\ell$ is the maximum number of blocks in the input-data for any individual query, and $\sigma$ is the total number of blocks of input-data across all queries.
      
    \item Post-quantum security: By applying \expref{Theorem}{thm:general-qprp-theorem}, we get
    \begin{align*}
        \textsf{Adv}_{\gcm}^{\textsf{nml-PRF-PQ}}(q,t) &\leq \textsf{Adv}_E^{\qprp}(q,t)  + \frac{(\sigma+3q)^2}{2^{n+1}}\\
        \textsf{Adv}^{\textsf{nml-Forgery-PQ}}_{\gcm}(q,t) &\leq \textsf{Adv}^{\qprp}_{E}(q,t) + \frac{(\sigma+3(q+q'))^2}{2^{n+1}}+ \frac{q'\ell}{2^{n}} + \frac{q'}{2^s}.
    \end{align*}
    \item Post-quantum security in QICM: By applying \expref{Theorem}{thm:general-qprp-theorem-ideal-cipher}, we get
    \begin{align*}
        \textsf{Adv}_{\gcm}^{\textsf{nml-PRF-PQ-IC}}(q_C,q_Q)  &\leq \frac{q_Q^2\ell^2}{2^m} + \frac{(\sigma+3q_C)^2}{2^{n+1}}\\
        \textsf{Adv}^{\textsf{nml-Forgery-PQ-IC}}_{\gcm}(q_C,q_Q)  &\leq \frac{q_Q^2\ell^2}{2^m} + + \frac{(\sigma+3(q_C+q_C'))^2}{2^{n+1}}+ \frac{q_C'\ell}{2^{n}} + \frac{q_C'}{2^s}.
    \end{align*}
\end{itemize}

\subsection{Security of \lrw and \xext using \expref{Theorem}{thm:general-qprp-theorem}}

The result from \expref{Theorem}{thm:general-qprp-theorem} extends to tweakable block cipher constructions. As a direct result, the \qsprp security of \lrw also follows from its classical security bound (\expref{Theorem}{thm:lrw-security}).
\begin{corollary} \label{thm:Q1-secure-LRW-2-AES}
    Let \lrw be as defined in \expref{Definition}{def:lrw}, constructed from a block cipher $E$ and an $\varepsilon$-XOR universal hash function family $h$. Then the post-quantum security of \lrw is given by:
    \begin{align*}
        \textsf{Adv}_{\lrw}^{\qsprp}(q,t) & \leq \textsf{Adv}^{\qsprp}_E(q,t) + q^2\varepsilon.
    \end{align*}
\end{corollary}

We analyze the \qsprp security of \lrw. The \lrw construction is built upon such $E$, and consequently, the adversary $\A$ is granted quantum oracle access to $E(\cdot,\cdot)$.
\begin{corollary} \label{cor:Q1-secure-LRW-2-ideal}
    Let \lrw be as defined in \expref{Definition}{def:lrw}. Assuming $h$ is XOR-universal, the post-quantum security of \lrw in QICM is given by:
    \begin{align*}
        \textsf{Adv}_{\lrw}^{\qsprp\textsf{-IC}}(q_C,q_Q) &= \left| \Pr_{\substack{k \leftarrow \{0,1\}^m; k' \leftarrow\{0,1\}^{\kappa} \\  E \leftarrow  \mathcal{E}(m,n)}} \left[ \A^{\pm \lrw_{k,k'}[E],\ket{\pm E}} = 1\right]  - \Pr_{\substack{\widetilde{\Pi} \leftarrow  \mathcal{E}(T,n); \\ E \leftarrow  \mathcal{E}(m,n)}} \left[ \A^{\pm\widetilde{\Pi},\ket{\pm E}} = 1\right] \right| \\
        & \leq \frac{q_Q^2}{2^m} + \frac{q_C^2}{2^n}.
    \end{align*}
\end{corollary}
Similarly to the classical case, the \qsprp\ security of \xext\ reduces to that of \lrw.
\begin{corollary} \label{thm:Q1-secure-XEX-2-AES}
    Let \xext be as defined in \expref{Definition}{def:xex2}. Then the post-quantum security of \xext in QICM is given by:
    \begin{align*}
        \textsf{Adv}_{\xext}^{\qsprp-\textsf{IC}}(q_C,q_Q) & \leq \frac{2q_Q^2}{2^m} + \frac{q_C^2}{2^n-1}.
    \end{align*}

\end{corollary}
\begin{proof}
Immediate consequence of \cref{prop:xext2idhashxext,lem:idhash-is-eps-XOR,thm:Q1-secure-LRW-2-AES}.
\qed
\end{proof}

\subsection{Comparison of \lrw Security Bound} 

Recall that the bound of \lrw proved in \expref{Theorem}{thm:Q1-secure-LRW-1} is:
$$
 \frac{6 q_C^2}{2^n} + \frac{4 \cdot \left( q_C \sqrt{q_Q}+ q_Q \sqrt{q_C}\right)}{2^{(m+n)/2}}.
$$ 
The comparison between the bounds from \expref{Corollary}{cor:Q1-secure-LRW-2-ideal} and \expref{Theorem}{thm:Q1-secure-LRW-1} is summarized earlier in \expref{Table}{tab:lrw-comp}. A detailed discussion follows below, omitting multiplicative constants. Note that in applying the general theorem, we do not need to assume that $h$ is uniform.
\begin{itemize}
   \item When $m\gg n$: The two bounds converge to $\frac{q_C^2}{2^n}$, and it aligns with the classical optimal bound (and classical matching attack) in \expref{Theorem}{thm:lrw-security}.
   \item When $q_C \gg q_Q$: The bounds also align at $\frac{q_C^2}{2^n}$, and it aligns with the classical optimal bound (and classical matching attack) in \expref{Theorem}{thm:lrw-security}.
   \item When $q_C \ll q_Q$: This scenario is more practical since online queries ($q_C$) are typically more scarce resources. Under these conditions, the bound in \expref{Theorem}{thm:Q1-secure-LRW-1} simplifies to approximately $\frac{q_Q \sqrt{q_C}}{2^{(m+n)/2}}$, whereas the bound in \expref{Corollary}{cor:Q1-secure-LRW-2-ideal} becomes $\frac{q_Q^2}{2^m}$. The bound from \expref{Theorem}{thm:Q1-secure-LRW-1} is more favorable in this case. 
   
   Notably, this bound aligns with two quantum complete key recovery attacks: the offline-Simon attack with the trade-off $q_Q^2 q_C = 2^{(m+n)/2}$, and the Grover+Kuwakado--Morii attack with $q_Q = 2^{m/2 + n/3}$ and $q_C = 2^{n/3}$.
   \item When $q_C \approx q_Q \approx q$: The bound in \expref{Theorem}{thm:Q1-secure-LRW-1} simplifies to approximately $\frac{q^2}{2^n} + \frac{q^{3/2}}{2^{(m+n)/2}}$, where the second term corresponds to the offline-Simon attack described in \expref{Section}{sec:lrw-Q1-attack}, assuming $m \leq 2n$. Furthermore, when $m \approx n$, both bounds reduce to $\frac{q^2}{2^n}$, aligning with the classical collision attack.
\end{itemize}

A summary is provided in \expref{Table}{tab:lrw-comp}. From the case analysis, the bound in \expref{Theorem}{thm:Q1-secure-LRW-1} is typically tighter, with matching attacks. Nonetheless, the analysis in \expref{Theorem}{thm:Q1-secure-LRW-2-AES} and \expref{Corollary}{cor:Q1-secure-LRW-2-ideal} provides a simpler and more general method for deriving a post-quantum security bound.

\begin{table}
    \centering
    \bgroup
\def\arraystretch{1.8}
\begin{tabular}{ |p{4cm}||p{3cm}||p{3cm}||p{3cm}|}
 \hline
 Scenarios & Classical & \textsf{Q1} from GT & \textsf{Q1} from HT\\
 \hline
 $m \gg n$  & $\frac{q_{C}^2}{2^n}$ &$\frac{q_C^2}{2^n}$ & $\frac{q_C^2}{2^n}$\\
 $q_C \gg q_Q$ & $\frac{q_C^2}{2^n}$ &  $\frac{q_C^2}{2^n}$ & $\frac{q_C^2}{2^n}$\\
 $\underline{q_C \ll q_Q}$ & $\frac{q_Q}{2^{m}}$ & $\frac{q_Q^2}{2^m}$ & $\underline{\frac{q_Q \sqrt{q_C}}{2^{(m+n)/2}}}$\\
 $\underline{q_C \approx q_Q = q}$ & $\frac{q^2}{2^n}+\frac{q}{2^{m}} $ & $ \frac{q^2}{2^n} + \frac{q^2}{2^m}$  & $\underline{\frac{q^2}{2^n} + \frac{q^{3/2}}{2^{(m+n)/2}}}$ \\
 \hline
\end{tabular}
\egroup
\vspace{2mm}
    \caption{Comparison of the classical bound (\expref{Theorem}{thm:lrw-security}), the post-quantum security bounds of \lrw using the General Theorem (GT, \expref{Corollary}{cor:Q1-secure-LRW-2-ideal}) and the Hybrid Technique (HT, \expref{Theorem}{thm:Q1-secure-LRW-1}). In the classical case, $q_Q$ denotes the number of classical queries made to $E$, denoted for straightforward comparison.}\label{tab:lrw-comp}
\end{table}


\section*{Acknowledgments}
Gorjan Alagic and Kaiyan Shi were supported in part by NSF award CNS-21547.
Chen Bai acknowledges support from the Presidential Postdoctoral Fellowship at Virginia Tech. 
Christian Majenz
acknowledges support by the Independent Research Fund Denmark via a DFF Sapere Aude grant (IM-3PQC,
grant ID 10.46540/2064-00034B).

\bibliographystyle{alpha}

\bibliography{abbrev3,crypto,ArXiv/reference}

\newcommand{\etalchar}[1]{$^{#1}$}
\begin{thebibliography}{BHNP{\etalchar{+}}19}

\bibitem[ABK{\etalchar{+}}24]{EC:ABKMS24}
Gorjan Alagic, Chen Bai, Jonathan Katz, Christian Majenz, and Patrick Struck.
\newblock Post-quantum security of tweakable {Even}-{Mansour}, and applications.
\newblock In Marc Joye and Gregor Leander, editors, {\em EUROCRYPT~2024, Part~I}, volume 14651 of {\em {LNCS}}, pages 310--338. Springer, Cham, May 2024.

\bibitem[ABKM22]{EC:ABKM22}
Gorjan Alagic, Chen Bai, Jonathan Katz, and Christian Majenz.
\newblock Post-quantum security of the {Even}-{Mansour} cipher.
\newblock In Orr Dunkelman and Stefan Dziembowski, editors, {\em EUROCRYPT~2022, Part~III}, volume 13277 of {\em {LNCS}}, pages 458--487. Springer, Cham, May~/~June 2022.

\bibitem[ADK{\etalchar{+}}14]{albrecht2014block}
Martin~R Albrecht, Benedikt Driessen, Elif~Bilge Kavun, Gregor Leander, Christof Paar, and Tolga Yal{\c{c}}{\i}n.
\newblock Block ciphers--focus on the linear layer (feat. pride).
\newblock In {\em Annual Cryptology Conference}, pages 57--76. Springer, 2014.

\bibitem[BBBV97]{bennett1997strengths}
Charles~H Bennett, Ethan Bernstein, Gilles Brassard, and Umesh Vazirani.
\newblock Strengths and weaknesses of quantum computing.
\newblock {\em SIAM journal on Computing}, 26(5):1510--1523, 1997.

\bibitem[BBC{\etalchar{+}}25]{basak2025post}
Jyotirmoy Basak, Ritam Bhaumik, Amit~Kumar Chauhan, Ravindra Jejurikar, Ashwin Jha, Anandarup Roy, Andr{\'e} Schrottenloher, and Suprita Talnikar.
\newblock Post-quantum security of key-alternating feistel ciphers.
\newblock {\em Cryptology ePrint Archive}, 2025.

\bibitem[BCG{\etalchar{+}}12]{borghoff2012prince}
Julia Borghoff, Anne Canteaut, Tim G{\"u}neysu, Elif~Bilge Kavun, Miroslav Knezevic, Lars~R Knudsen, Gregor Leander, Ventzislav Nikov, Christof Paar, Christian Rechberger, et~al.
\newblock Prince--a low-latency block cipher for pervasive computing applications.
\newblock In {\em International conference on the theory and application of cryptology and information security}, pages 208--225. Springer, 2012.

\bibitem[BEM24]{bai2024quantum}
Chen Bai, Mehdi Esmaili, and Atul Mantri.
\newblock Quantum security analysis of the key-alternating ciphers.
\newblock {\em arXiv preprint arXiv:2412.05026}, 2024.

\bibitem[BHN{\etalchar{+}}19]{AC:BHNSS19}
Xavier Bonnetain, Akinori Hosoyamada, Mar{\'i}a {Naya-Plasencia}, Yu~Sasaki, and Andr{\'e} Schrottenloher.
\newblock Quantum attacks without superposition queries: The offline {Simon}'s algorithm.
\newblock In Steven~D. Galbraith and Shiho Moriai, editors, {\em ASIACRYPT~2019, Part~I}, volume 11921 of {\em {LNCS}}, pages 552--583. Springer, Cham, December 2019.

\bibitem[BHNP{\etalchar{+}}19]{bonnetain2019quantum}
Xavier Bonnetain, Akinori Hosoyamada, Mar{\'\i}a Naya-Plasencia, Yu~Sasaki, and Andr{\'e} Schrottenloher.
\newblock Quantum attacks without superposition queries: the offline simon’s algorithm.
\newblock In {\em International Conference on the Theory and Application of Cryptology and Information Security}, pages 552--583. Springer, 2019.

\bibitem[BHT97]{brassard1997quantum}
Gilles Brassard, Peter Hoyer, and Alain Tapp.
\newblock Quantum algorithm for the collision problem.
\newblock {\em arXiv preprint quant-ph/9705002}, 1997.

\bibitem[BSS22]{EC:BonSchSib22}
Xavier Bonnetain, Andr{\'e} Schrottenloher, and Ferdinand Sibleyras.
\newblock Beyond quadratic speedups in quantum attacks on symmetric schemes.
\newblock In Orr Dunkelman and Stefan Dziembowski, editors, {\em EUROCRYPT~2022, Part~III}, volume 13277 of {\em {LNCS}}, pages 315--344. Springer, Cham, May~/~June 2022.

\bibitem[CHL{\etalchar{+}}25]{cojocaru2025quantum}
Alexandru Cojocaru, Minki Hhan, Qipeng Liu, Takashi Yamakawa, and Aaram Yun.
\newblock Quantum lifting for invertible permutations and ideal ciphers.
\newblock In {\em Annual International Cryptology Conference}, pages 481--512. Springer, 2025.

\bibitem[CMM23]{campagna2023galois}
Matthew Campagna, Alexander Maximov, and John~Preu{\ss} Mattsson.
\newblock Galois counter mode with secure short tags (gcm-sst).
\newblock 2023.

\bibitem[CNPS17]{chailloux2017efficient}
Andr{\'e} Chailloux, Mar{\'\i}a Naya-Plasencia, and Andr{\'e} Schrottenloher.
\newblock An efficient quantum collision search algorithm and implications on symmetric cryptography.
\newblock In {\em International Conference on the Theory and Application of Cryptology and Information Security}, pages 211--240. Springer, 2017.

\bibitem[CSR{\etalchar{+}}08]{clunie2008public}
David Clunie, Rich Shroeppel, Phillip Rogaway, Vijay Bharadwaj, and Neils Ferguson.
\newblock Public comments on the xts-aes mode.
\newblock {\em Collected email comments released by NIST, available from their web page}, 2008.

\bibitem[DBN{\etalchar{+}}01]{dworkin2001advanced}
Morris~J Dworkin, Elaine Barker, James~R Nechvatal, James Foti, Lawrence~E Bassham, E~Roback, James~F Dray~Jr, et~al.
\newblock Advanced encryption standard (aes).
\newblock 2001.

\bibitem[Din15]{dinur2015cryptanalytic}
Itai Dinur.
\newblock Cryptanalytic time-memory-data tradeoffs for fx-constructions with applications to prince and pride.
\newblock In {\em Annual International Conference on the Theory and Applications of Cryptographic Techniques}, pages 231--253. Springer, 2015.

\bibitem[Dwo01]{dworkin2001recommendation}
Morris Dworkin.
\newblock Recommendation for block cipher modes of operation.
\newblock {\em NIST special publication}, 800:38B, 2001.

\bibitem[Dwo07]{dworkin2007sp}
Morris~J Dworkin.
\newblock {\em Sp 800-38d. recommendation for block cipher modes of operation: Galois/counter mode (gcm) and gmac}.
\newblock National Institute of Standards \& Technology, 2007.

\bibitem[Dwo10a]{dworkin2010recommendation}
Morris~J Dworkin.
\newblock Recommendation for block cipher modes of operation: The xts-aes mode for confidentiality on storage devices.
\newblock 2010.

\bibitem[{Dwo}10b]{NIST.SP.800-38E}
{Dworkin, Morris}.
\newblock {Recommendation for Block Cipher Modes of Operation: The XTS-AES Mode for Confidentiality on Storage Devices}.
\newblock Technical Report SP 800-38E, {National Institute of Standards and Technology (NIST)}, Gaithersburg, MD, February 2010.

\bibitem[EMST78]{ehrsam1978message}
William~F Ehrsam, Carl~HW Meyer, John~L Smith, and Walter~L Tuchman.
\newblock Message verification and transmission error detection by block chaining, February~14 1978.
\newblock US Patent 4,074,066.

\bibitem[GHHM21]{grilo2021tight}
Alex~B Grilo, Kathrin H{\"o}velmanns, Andreas H{\"u}lsing, and Christian Majenz.
\newblock Tight adaptive reprogramming in the qrom.
\newblock In {\em International Conference on the Theory and Application of Cryptology and Information Security}, pages 637--667. Springer, 2021.

\bibitem[GLL17]{gueron2017aes}
Shay Gueron, Adam Langley, and Yehuda Lindell.
\newblock Aes-gcm-siv: specification and analysis.
\newblock {\em Cryptology ePrint Archive}, 2017.

\bibitem[GT12]{EC:GazTes12}
Peter Ga{\v{z}}i and Stefano Tessaro.
\newblock Efficient and optimally secure key-length extension for block ciphers via randomized cascading.
\newblock In David Pointcheval and Thomas Johansson, editors, {\em EUROCRYPT~2012}, volume 7237 of {\em {LNCS}}, pages 63--80. Springer, Berlin, Heidelberg, April 2012.

\bibitem[Hos25]{hosoyamada2025post}
Akinori Hosoyamada.
\newblock Post-quantum security of keyed sponge-based constructions through a modular approach.
\newblock {\em Cryptology ePrint Archive}, 2025.

\bibitem[HS18]{hosoyamada2018cryptanalysis}
Akinori Hosoyamada and Yu~Sasaki.
\newblock Cryptanalysis against symmetric-key schemes with online classical queries and offline quantum computations.
\newblock In {\em Cryptographers’ Track at the RSA Conference}, pages 198--218. Springer, 2018.

\bibitem[HY18]{AC:HosYas18}
Akinori Hosoyamada and Kan Yasuda.
\newblock Building quantum-one-way functions from block ciphers: {Davies}-{Meyer} and {Merkle}-{Damg{\r a}rd} constructions.
\newblock In Thomas Peyrin and Steven Galbraith, editors, {\em ASIACRYPT~2018, Part~I}, volume 11272 of {\em {LNCS}}, pages 275--304. Springer, Cham, December 2018.

\bibitem[HY24]{hhan2024pseudorandom}
Minki Hhan and Shogo Yamada.
\newblock Pseudorandom function-like states from common haar unitary.
\newblock {\em arXiv preprint arXiv:2411.03201}, 2024.

\bibitem[IJMM24]{inoue2024generic}
Akiko Inoue, Ashwin Jha, Bart Mennink, and Kazuhiko Minematsu.
\newblock Generic security of gcm-sst.
\newblock {\em Cryptology ePrint Archive}, 2024.

\bibitem[IK03]{iwata2003omac}
Tetsu Iwata and Kaoru Kurosawa.
\newblock Omac: One-key cbc mac.
\newblock In {\em Fast Software Encryption: 10th International Workshop, FSE 2003, Lund, Sweden, February 24-26, 2003. Revised Papers 10}, pages 129--153. Springer, 2003.

\bibitem[IM19]{isobe2019plaintext}
Takanori Isobe and Kazuhiko Minematsu.
\newblock Plaintext recovery attacks against xts beyond collisions.
\newblock In {\em International Conference on Selected Areas in Cryptography}, pages 103--123. Springer, 2019.

\bibitem[IOM12]{iwata2012breaking}
Tetsu Iwata, Keisuke Ohashi, and Kazuhiko Minematsu.
\newblock Breaking and repairing gcm security proofs.
\newblock In {\em Advances in Cryptology--CRYPTO 2012: 32nd Annual Cryptology Conference, Santa Barbara, CA, USA, August 19-23, 2012. Proceedings}, pages 31--49. Springer, 2012.

\bibitem[JN20]{jha2020tight}
Ashwin Jha and Mridul Nandi.
\newblock Tight security of cascaded lrw2.
\newblock {\em Journal of Cryptology}, 33(3):1272--1317, 2020.

\bibitem[JST21]{jaeger2021quantum}
Joseph Jaeger, Fang Song, and Stefano Tessaro.
\newblock Quantum key-length extension.
\newblock In {\em Theory of Cryptography Conference}, pages 209--239. Springer, 2021.

\bibitem[KL07]{katz2007introduction}
Jonathan Katz and Yehuda Lindell.
\newblock {\em Introduction to modern cryptography: principles and protocols}.
\newblock Chapman and hall/CRC, 2007.

\bibitem[KLLN16]{C:KLLN16}
Marc Kaplan, Ga{\"e}tan Leurent, Anthony Leverrier, and Mar{\'i}a {Naya-Plasencia}.
\newblock Breaking symmetric cryptosystems using quantum period finding.
\newblock In Matthew Robshaw and Jonathan Katz, editors, {\em CRYPTO~2016, Part~II}, volume 9815 of {\em {LNCS}}, pages 207--237. Springer, Berlin, Heidelberg, August 2016.

\bibitem[KM12]{kuwakado2012security}
Hidenori Kuwakado and Masakatu Morii.
\newblock Security on the quantum-type even-mansour cipher.
\newblock In {\em 2012 international symposium on information theory and its applications}, pages 312--316. IEEE, 2012.

\bibitem[KR96]{C:KilRog96}
Joe Kilian and Phillip Rogaway.
\newblock How to protect {DES} against exhaustive key search.
\newblock In Neal Koblitz, editor, {\em CRYPTO'96}, volume 1109 of {\em {LNCS}}, pages 252--267. Springer, Berlin, Heidelberg, August 1996.

\bibitem[KR01]{JC:KilRog01}
Joe Kilian and Phillip Rogaway.
\newblock How to protect {DES} against exhaustive key search (an analysis of {DESX}).
\newblock {\em Journal of Cryptology}, 14(1):17--35, January 2001.

\bibitem[LM08]{liskov2008comments}
Moses Liskov and Kazuhiko Minematsu.
\newblock Comments on xts-aes.
\newblock {\em Comments to NIST, available from their web page}, 2008.

\bibitem[LRW02]{liskov2002tweakable}
Moses Liskov, Ronald~L Rivest, and David Wagner.
\newblock Tweakable block ciphers.
\newblock In {\em Advances in Cryptology—CRYPTO 2002: 22nd Annual International Cryptology Conference Santa Barbara, California, USA, August 18--22, 2002 Proceedings 22}, pages 31--46. Springer, 2002.

\bibitem[Min06]{minematsu2006improved}
Kazuhiko Minematsu.
\newblock Improved security analysis of xex and lrw modes.
\newblock In {\em International Workshop on Selected Areas in Cryptography}, pages 96--113. Springer, 2006.

\bibitem[MV04]{mcgrew2004security}
David~A McGrew and John Viega.
\newblock The security and performance of the galois/counter mode of operation (full version).
\newblock {\em Cryptology ePrint Archive}, 2004.

\bibitem[NOMI15]{niwa2015gcm}
Yuichi Niwa, Keisuke Ohashi, Kazuhiko Minematsu, and Tetsu Iwata.
\newblock Gcm security bounds reconsidered.
\newblock In {\em Fast Software Encryption: 22nd International Workshop, FSE 2015, Istanbul, Turkey, March 8-11, 2015, Revised Selected Papers 22}, pages 385--407. Springer, 2015.

\bibitem[P1606]{P1619D4}
Ieee draft standard for authenticated encryption with length expansion for storage device.
\newblock {\em IEEE Unapproved Draft Std P1619.1/D4, Oct 2007}, 2006.

\bibitem[P1619]{P1619F}
Ieee standard for cryptographic protection of data on block-oriented storage devices.
\newblock {\em IEEE Std 1619-2018 (Revision of IEEE Std 1619-2007)}, pages 1--41, 2019.

\bibitem[PR00]{petrank2000cbc}
Erez Petrank and Charles Rackoff.
\newblock Cbc mac for real-time data sources.
\newblock {\em Journal of Cryptology}, 13:315--338, 2000.

\bibitem[Rog02]{rogaway2002authenticated}
Phillip Rogaway.
\newblock Authenticated-encryption with associated-data.
\newblock In {\em Proceedings of the 9th ACM Conference on Computer and Communications Security}, pages 98--107, 2002.

\bibitem[Rog04]{AC:Rogaway04}
Phillip Rogaway.
\newblock Efficient instantiations of tweakable blockciphers and refinements to modes {OCB} and {PMAC}.
\newblock In Pil~Joong Lee, editor, {\em ASIACRYPT~2004}, volume 3329 of {\em {LNCS}}, pages 16--31. Springer, Berlin, Heidelberg, December 2004.

\bibitem[Rog13]{rogaway2013evaluation}
Phillip Rogaway.
\newblock Evaluation of some blockcipher modes of operation. 2011.
\newblock {\em Cryptography Research and Evaluation Committees (CRYPTREC) for the Government of Japan}, 2013.

\bibitem[Sho94]{shor1994algorithms}
Peter~W Shor.
\newblock Algorithms for quantum computation: discrete logarithms and factoring.
\newblock In {\em Proceedings 35th annual symposium on foundations of computer science}, pages 124--134. Ieee, 1994.

\bibitem[SS19]{sato2019so}
Shingo Sato and Junji Shikata.
\newblock So-cca secure pke in the quantum random oracle model or the quantum ideal cipher model.
\newblock In {\em IMA International Conference on Cryptography and Coding}, pages 317--341. Springer, 2019.

\bibitem[Vau00]{vaudenay2000decorrelation}
Serge Vaudenay.
\newblock Decorrelation over infinite domains: the encrypted cbc-mac case.
\newblock In {\em International Workshop on Selected Areas in Cryptography}, pages 189--201. Springer, 2000.

\bibitem[Zha19]{zhandry2019record}
Mark Zhandry.
\newblock How to record quantum queries, and applications to quantum indifferentiability.
\newblock In {\em Annual International Cryptology Conference}, pages 239--268. Springer, 2019.

\end{thebibliography}

 \appendix
 \newpage
 
\section{Proofs of the new resampling lemma}\label{app:lemma-proofs}
\subsection{Proof of \expref{Lemma}{lem:resampling-ic}} \label{app:arl-ic}

\begin{lemma}[Restatement of \expref{Lemma}{lem:resampling-ic}]
Let $\algo D=(\D_0,\D_1)$ be a quantum distinguisher in the following experiment:
 \begin{description} 
     \item[Phase 1:] Choose uniform $E \in \algo E(m, n)$ and give $\D$ quantum access to~$E$ and~$E^{-1}$.  Then $\algo D_0$ chooses and outputs a distribution $D$ over $\bool^{m+2n}$.
     \item[Phase 2:] Choose $k_0 \in \bool^m$ and $s_0,s_1 \in \bool^n$ according to $D$.
     Let $E^{(0)}=E$  and define $E^{(1)}:\bool^m\times \bool^n \rightarrow \bool^n$ by
        $$
        E_{k^*}^{(1)}(x)=
        \begin{cases}
        E_{k^*}(x) &\text{if } k^* \neq k_0\\
        E_{k^*} \circ \swap{s_0}{s_1}(x)  &\text{if } k^*=k_0 \,.
        \end{cases}
        $$		
	A uniform bit~$b \in \bool$ is chosen, and $\D$ is given $k_0,s_0, s_1$, and quantum access to~$E^{(b)}$. Then $\D$ outputs a guess~$b'$. 
 \end{description}
For a probability distribution~$D$ on $\{0,1\}^{m+2n}$, we define 
$$
\epsilon=\max_{(k_0^*, s_0^*, s_1^*)\in\{0,1\}^{m+2n}}D(k_0,s_0,s_1)
$$
\noindent Then for any $\D$ making at most $q$ queries to $E$ in phase~1,
\begin{align*}
\left|\Pr[\mbox{$\D$ outputs 1} \mid b=1] - \Pr[\mbox{$\D$ outputs 1} \mid b=0]\right| \leq 4\sqrt{2^n\cdot q\cdot \epsilon}\,.
\end{align*}
\end{lemma}

\begin{proof}
The proof of this lemma generalizes Lemma~3 from \cite{hosoyamada2025post}, which in turn builds on techniques from \cite{EC:ABKM22}. Thus, part of the proof is borrowed from these two works.

Let $F=F_{0^m}F_{0^{m-1}1},\ldots, F_{1^m}$ be the internal register of a superposition oracle for an ideal cipher, where each $F_k=F_{k,0^n},\ldots, F_{k,1^n}$ is a database register for a random permutation.  Each $F_k$ is initialized in the initial state $\ket{\phi_0}$ for a random permutation $\pi$, namely, 
\begin{align*}
    \ket{\phi_0}_{F_k}=\left(2^n !\right)^{-1/2}\sum_{\pi\in\permset{n}}\ket{\pi}_{F_k}.
\end{align*}

Moreover, define $\ket{\Phi_0}_F=\ket{\phi_0}^{\otimes 2^m}$. 

Then, quantum queries to the ideal cipher can be evaluated via the unitary operators $O$ and $O^{\mathrm{inv}}$:
\begin{align*}
O_{KXYF}\;\ket{k}_K\ket x_X\ket y_Y\ket {\pi}_{F_k} &=\ket{k}_K\ket{x}_X\ket{y\oplus \pi(x)}_Y\ket {\pi}_{F_k}\\  
\Oinv_{KXYF} \ket{k}_K\ket x_X \ket y_Y \ket {\pi}_{F_k} &= \ket{k}_K\ket {x \xor \left(\oplus_{x': \pi(x')=y} \, x'\right)}_X \ket{y}_Y \ket{\pi}_{F_k}.
\end{align*}
Note that $O$ and $\Oinv$ act only on the register $F_k$, and act as the identity on all other $F$-registers.
By analogy to the proof of \cite{EC:ABKM22,hosoyamada2025post}, define the projectors
\begin{align*}
\left(P_{k_0,s_0,s_1}\right)_{KX} =
    \begin{cases}
        \mathds 1& s_0=s_1\\
        \mathds 1-\proj{k_0}\otimes \left(\proj{s_0}+\proj{s_1}\right)_X&s_0\neq s_1
    \end{cases}
\end{align*}
and 
\begin{align*}
    &\left(P^{\mathrm{inv}}_{k_0,s_0,s_1}\right)_{KYF}\nonumber\\
	=&\begin{cases}
    	\mathds 1& s_0=s_1\\
            \proj{k_0}_K\otimes \sum_{y\in\{0,1\}^n}\proj{y}_Y\otimes\left(\mathds 1-\proj y\right)^{\otimes 2}_{F_{k_0,s_0}F_{k_0,s_1}}&s_0\neq s_1.
	\end{cases}
\end{align*}
Moreover, let $\bar{P}_{k_0,s_0,s_1}=\mathds 1- P_{k_0,s_0,s_1}$ and $\bar{P}_{k_0,s_0,s_1}^{\mathrm{inv}}=\mathds 1- P^{\mathrm{inv}}_{k_0,s_0,s_1}$. Define the swap operator $$\S_{AB}\ket{x}_A\ket{y}_B=\ket{y}_A\ket{x}_B$$.

Without loss of generality, we assume (as in~\cite{hosoyamada2025post}) that $\mathcal{D}_0$ makes exactly $q$ quantum queries and does not perform intermediate measurements. The distribution $D$ is chosen by 
$\mathcal{D}_0$ as follows.

\begin{enumerate}
  \item At the beginning, $\D_0$ sets a specific part of the offline register $E$, which we
  refer to as register $Z$, to be $\ket{0}_Z$. Register $Z$ is untouched until $\D_0$
  has made all its $q$ queries.
  \item After making $q$ quantum queries, $\D_0$ applies a unitary to all its internal registers, including $Z$.
  \item Then, it measures register $Z$ to obtain the distribution $D$.
\end{enumerate}

We denote the projector $\ket{D}\!\bra{D}_Z$ as $P_D$. Let $\ket{\psi_0}= \ket{\psi'_0}_{KXYE}\ket{\phi_0}_F$ be the initial state of $\D_0$.
Let $\ket{\psi}$ be the quantum state right before the measurement to choose the
distribution $D$. Then,
\begin{equation}\label{eqn:psi}
  \ket{\psi}
  = U_q O_q \cdots U_1 O_1 \ket{\psi_0},
\end{equation}
holds, where each $O_i$ (which is either $O_{KXYF}$ or $O^{\mathrm{inv}}_{KXYF}$)
is the unitary operator corresponding to the $i^\text{th}$ quantum query, and $U_i$ is a
unitary operator acting on registers $XYE$ that corresponds to some offline computation by
$\D_0$. In addition, for arbitrary $k_0\in \bool^m$ and $s_0,s_1 \in \{0,1\}^n$, let
\begin{align*}
  \ket{\psi_{\mathrm{good}}(k_0,s_0,s_1)} &:= z \cdot
  U_q O_q P_{k_0,s_0,s_1}^q \cdots U_1 O_1 P_{k_0,s_0,s_1}^1
  \ket{\psi_0}\\
  \ket{\psi_{\mathrm{bad}}(k_0,s_0,s_1)} &:= \ket{\psi} - \ket{\psi_{\mathrm{good}}(k_0,s_0,s_1)},
\end{align*}
where
\begin{align*}
P_{k_0,s_0,s_1}^i :=
\begin{cases}
  (P_{k_0,s_0,s_1})_X & \text{if $O_i = O_{KXYF}$}\\
  (P^{\mathrm{inv}}_{k_0,s_0,s_1})_{KYF} & \text{if $O_i = O^{\mathrm{inv}}_{KXYF}$},
\end{cases}
\end{align*}
and $z$ is a complex number such that $|z| = 1$ and
$\langle \psi \mid \psi_{\mathrm{good}}(k_0,s_0,s_1) \rangle \in \mathbb{R}_{\geq 0}$.
As in \cite{EC:ABKM22,hosoyamada2025post}, the vector $\ket{\psi_{\mathrm{good}}(k_0,s_0,s_1)}$ is invariant under the action of $\S_{F_{k_0,s_0},F_{k_0,s_1}}$, that is,
\begin{equation} \label{eqn:swapgood}
\S_{F_{k_0,s_0},F_{k_0,s_1}} \ket{\psi_{\mathrm{good}}(k_0,s_0,s_1)} 
= \ket{\psi_{\mathrm{good}}(k_0,s_0,s_1)}.
\end{equation}
Moreover, let
\begin{align*}
\epsilon_i(k_0,s_0,s_1) := \left\| \bar{P}^{i}_{k_0,s_0,s_1} U_{i-1} O_{i-1} \cdots U_1 O_1 
  \ket{\psi_0} \right\|_2^2,
\end{align*}
where $\bar{P}^{i}_{k_0,s_0,s_1} := \mathds 1 - P^{i}_{k_0,s_0,s_1}$.

Recall that, at the end of $\D_0$, the register $Z$ (which is a part of
register $E$ for offline computation) is measured to choose a distribution $D$.
The (pure) state just before the measurement is $\ket{\psi}$ of \expref{Equation}{eqn:psi}.
By the measurement, $\ket{\psi}$ changes to the mixed state
\begin{align*}
\sum_{D \in \mathcal{S}} P_D \ket{\psi}\bra{\psi} P_D.
\end{align*}
After that, $k_0$, $s_0$ and $s_1$ are sampled, and the state changes to
\begin{align*}
\rho_0 := \sum_{\substack{D \in \mathcal{S} \\  (k_0,s_0,s_1) \sim D}}
  D(k_0,s_0,s_1) \;
  P_D \ket{\psi}\bra{\psi}_{KXYEF} P_D
  \otimes \ket{k_0,s_0,s_1}\bra{k_0,s_0,s_1}_C ,
\end{align*}
where $C$ is an additional register.

If $b=1$ and the values of the permutation are swapped on $s_0$ and $s_1$, the
state further changes to
\begin{align*}
\rho_1 := \sum_{\substack{D \in \mathcal{S} \\  (k_0,s_0,s_1) \sim D}}
  D(k_0,s_0,s_1) \;
  \S_{F_{k_0,s_0},F_{k_0,s_1}} P_D &\ket{\psi}\bra{\psi} P_D \S_{F_{k_0,s_0},F_{k_0,s_1}} \\
  &\otimes \ket{k_0,s_0,s_1}\bra{k_0,s_0,s_1}.
\end{align*}
From the standard argument for distinguishing advantages,
\begin{equation*}\label{eq:disting2TD}
\bigl|\Pr[b'=1 \mid b=1] - \Pr[b'=1 \mid b=0]\bigr|
  \leq \mathsf{TD}(\rho_0,\rho_1),
\end{equation*}
where $\mathsf{TD}$ is the trace distance.

We observe that the state $\rho_0$ can also be produced from $\ket{\psi}$ as follows.  

\begin{enumerate}
  \item Without measuring register $Z$ (on which the information of $D$ is written),
  create the superposition of $k_0$, $s_0$ and $s_1$ of the form
  \begin{align*}
    \sum_{(k_0,s_0,s_1)\sim D} \sqrt{D(k_0,s_0,s_1)} \;\ket{D}_Z \ket{k_0,s_0,s_1}_C.
  \end{align*}
  For each $D$, obtaining the (pure) state
  \begin{align*}
    \ket{\xi_0} := \sum_{\substack{D \in \mathcal{S} \\  (k_0,s_0,s_1) \sim D}} \sqrt{D(k_0,s_0,s_1)} \;
      P_D \ket{\psi}_{KXYEF} \otimes \ket{k_0,s_0,s_1}_C .
  \end{align*}
  
  \item Measure registers $Z$ and $C$ to obtain the classical information of $D,k_0,s_0,s_1$.
\end{enumerate}

Similarly, the state $\rho_1$ can be produced by measuring the registers $Z$ and $C$ of
the pure state
\begin{align*}
\ket{\xi_1} := \sum_{\substack{D \in \mathcal{S} \\  (k_0,s_0,s_1) \sim D}} \sqrt{D(k_0,s_0,s_1)} \;
  \S_{F_{k_0,s_0},F_{k_0,s_1}} P_D \ket{\psi}_{KXYEF} \otimes \ket{k_0,s_0,s_1}_C .
\end{align*}

Since measurements decrease the trace distance between two states, and the trace
distance between the two pure states $\ket{\xi_0}\!\bra{\xi_0}$ and
$\ket{\xi_1}\!\bra{\xi_1}$ is upper bounded by the norm of the difference
$(\ket{\xi_0}-\ket{\xi_1})$, it holds that
\begin{equation*}\label{eq:purify}
    \mathsf{TD}(\rho_0,\rho_1) \;\leq\;
\mathsf{TD}(\ket{\xi_0}\!\bra{\xi_0}, \ket{\xi_1}\!\bra{\xi_1})
\;\leq\; \|\ket{\xi_0}-\ket{\xi_1}\|_2 .
\end{equation*}

In addition, we have
\begin{align} \label{eqn:TD}
\|\ket{\xi_0}-\ket{\xi_1}\|_2^2
&\overset{(i)}{=}\;
   \sum_{\substack{D \in \mathcal{S} \\ (k_0,s_0,s_1) \sim D}}
   D(k_0,s_0,s_1)\,
   \left\| P_D \bigl(1-\mathsf{Swap}_{F_{k_0,s_0},F_{k_0,s_1}}\bigr) \ket{\psi} \right\|_2^2 \nonumber\\[0.8ex]\nonumber
&\overset{(ii)}{=}\;
   \sum_{\substack{D \in \mathcal{S} \\ (k_0,s_0,s_1) \sim D}}
   D(k_0,s_0,s_1)\,
   \Bigl\| P_D \Bigl(
       \ket{\psi_{\mathrm{bad}}(k_0,s_0,s_1)}
\\[-0.5ex]\nonumber &\hspace{10em}
       -\, \mathsf{Swap}_{F_{k_0,s_0},F_{k_0,s_1}}
         \ket{\psi_{\mathrm{bad}}(k_0,s_0,s_1)}
     \Bigr)\Bigr\|_2^2 \\[0.8ex]\nonumber
&\overset{(iii)}{\le}\;
   4 \sum_{\substack{D \in \mathcal{S} \\ (k_0,s_0,s_1)\in\{0,1\}^{m+2n}}}
   D(k_0,s_0,s_1)\,
   \left\| P_D \ket{\psi_{\mathrm{bad}}(k_0,s_0,s_1)} \right\|_2^2 \\[0.8ex]\nonumber
&\overset{(iv)}{\le}\;
   4 \sum_{\substack{D \in \mathcal{S} \\ (k_0,s_0,s_1)\in\{0,1\}^{m+2n}}}
   \varepsilon\,
   \left\| P_D \ket{\psi_{\mathrm{bad}}(k_0,s_0,s_1)} \right\|_2^2 \\[0.8ex]
   &\overset{(v)}{=}\;
   4 \varepsilon\cdot 2^{m+2n}\, \mathbb{E}_{(k_0,s_0,s_1)\leftarrow\{0,1\}^{m+2n}}
  \left[ \|\ket{\psi_{\mathrm{bad}}(k_0,s_0,s_1)}\|_2^2 \right].
\end{align}

where we used the fact that
$\bigl[P_D,\S_{F_{k_0,s_0},F_{k_0,s_1}}\bigr]=0$ \footnote{$P_D$ acts only on (a part of) register $E$ while 
$\mathsf{Swap}_{F_{k_0,s_0},F_{k_0,s_1}}$ acts only on register $F$.} for (i) and (iii),
\expref{Equation}{eqn:swapgood} for (ii), and the fact that $\sum_DP_D=1$ for (v).

Next, we bound $\mathbb{E}_{(k_0,s_0,s_1)\leftarrow\{0,1\}^{m+2n}}
   \left[ \left\|\ket{\psi_{\mathrm{bad}}(k_0,s_0,s_1)}\right\|_2^2 \right]$. First, using the gentle measurement lemma as in \cite[Section 4.2]{EC:ABKM22}, we get
\begin{equation} \label{eqn:GML}
\left\|\ket{\psi_{\mathrm{bad}}(k_0,s_0,s_1)}\right\|_2^2
  \;\leq\; 2 \sum_{i=1}^q \epsilon_i(k_0,s_0,s_1).
\end{equation}
   
For any normalized state $\ket \psi_{KXE}$, let 
\begin{align*}
	\ket \psi_{KXE}=\sum_{\substack{k\in\{0,1\}^m\\x\in\{0,1\}^n}}\ket k_K\ket x_X\otimes\ket{\psi_{kx}}_E
\end{align*}
be its expansion in the computational basis on $X$. By the definition of $\epsilon$, following \cite[Section 4.2]{EC:ABKM22}, we obtain
\begin{align*} 
	&\mathbb E_{(k_0,s_0,s_1)\leftarrow\{0,1\}^{m+2n}}\left[\|\left(P_{k_0,s_0,s_1}\right)_{KX}\ket{\psi}_{KXE}\|_2^2\right]\nonumber\\
	&=\sum_{\substack{k\in\{0,1\}^m\\x\in\{0,1\}^n}}\|\ket{\psi_{kx}}\|_2^2\mathbb E_{(k_0,s_0,s_1)\leftarrow\{0,1\}^{m+2n}}\left[\|\left(P_{k_0,s_0,s_1}\right)_{KX}\ket k_K\ket x_X\|_2^2\right]\nonumber\\
	&=\sum_{\substack{k\in\{0,1\}^m\\x\in\{0,1\}^n}}\|\ket{\psi_{kx}}\|_2^2\Pr_{(k_0,s_0,s_1)\leftarrow\{0,1\}^{m+2n}}\left[(k,x)\in \{(k_0,s_0),(k_0,s_1)\}\right]\nonumber\\
	&\le 2\cdot 2^{-(n+m)}.
\end{align*}

Similarly, again following \cite[Section 4.2]{EC:ABKM22}
\begin{align*}
	\Exp_{(k_0,s_0, s_1)\sim D}\left[\left\|\left(\bar P^{\mathrm{inv}}_{k_0,s_0,s_1}\right)_{KYF}\ket\psi_{KYEF}\right\|_2^2\right]	\le2\cdot  2^{-(n+m)}.
\end{align*}

Then we have 
\begin{equation} \label{eqn:E}
\mathbb{E}_{(k_0,s_0,s_1) \sim D}
  \bigl[\epsilon_i(k_0,s_0,s_1)\bigr]
  \;\leq\; 2\cdot 2^{-(n+m)}.
\end{equation}
Combining \cref{eqn:TD,eqn:GML,eqn:E}, we obtain the desired bound.
\qed
\end{proof}

 \section{More Details on Post-Quantum Security of \fx}

\subsection{Proof of \expref{Proposition}{prop:E-T_jk-y}} \label{app:E-T_jk-y}

\begin{proposition} [Restatement]
    For any $E \in \mathcal{E}(m,n)$, $K = (k_0, k_1, k_2) \in \{0,1\}^{m+2n}$, $j \in [1,q_C]$, transcript $T_j = \{(x_1, y_1), \ldots, (x_j, y_j)\}$ without repetition, and any $i \in \{1, \ldots, j\}$, it holds that 
    \begin{align*}
        E^{T_j,K}_{k_0}(x_i \oplus k_1) \oplus k_2 = y_i.
    \end{align*} 
\end{proposition} 
\begin{proof}
We prove by strong induction on $i$. We start by the base case $i=j$. 
\begin{align*}
    E^{T_j,K}_{k_0}(x_j \oplus k_1) = \swap{E^{T_{j-1},K}_{k_0}( x_{j}\oplus k_1)}{y_{j}\oplus k_2} \circ E^{T_{j-1},K}_{k_0}(x_j \oplus k_1) = y_j \oplus k_2.
\end{align*}
Assume for all $i \in \{r, \ldots, j\}$ that $E^{T_j,K}_{k_0}(x_i \oplus k_1) \oplus k_2 = y_i$. Expanding it, we obtain
\begin{align*}
    &E^{T_j,K}_{k_0}(x_{i} \oplus k_1)\\
    &= \swap{E^{T_{j-1},K}_{k_0}( x_{j}\oplus k_1)}{y_{j}\oplus k_2} \circ \ldots \circ \swap{E^{T_{i-1},K}_{k_0}( x_{i}\oplus k_1)}{y_{i}\oplus k_2} \circ E^{T_{i-1},K}_{k_0}(x_{i} \oplus k_1)\\
    &= F_i \circ E^{T_{i-1},K}_{k_0}(x_{i} \oplus k_1) = y_i \oplus k_2,
\end{align*}
where we define
\begin{align*}
    F_i = \swap{E^{T_{j-1},K}_{k_0}( x_{j}\oplus k_1)}{y_{j}\oplus k_2} \circ \ldots \circ \swap{E^{T_{i-1},K}_{k_0}( x_{i}\oplus k_1)}{y_{i}\oplus k_2}.
\end{align*}
By rearranging, we get for all $i \in \{r, \ldots, j\}$,
\begin{align*}
    E^{T_{i-1},K}_{k_0}(x_{i} \oplus k_1) = F_i^{-1} (y_i \oplus k_2).
\end{align*}
We now work on case $r-1$, 
\begin{align*}
    &E^{T_j,K}_{k_0}(x_{r-1} \oplus k_1)\\
    &= \swap{E^{T_{j-1},K}_{k_0}( x_{j}\oplus k_1)}{y_{j}\oplus k_2} \circ \ldots \circ \swap{E^{T_{r-2},K}_{k_0}( x_{r-1}\oplus k_1)}{y_{r-1}\oplus k_2} \circ E^{T_{r-2},K}_{k_0}(x_{r-1} \oplus k_1)\\
    &= \swap{E^{T_{j-1},K}_{k_0}( x_{j}\oplus k_1)}{y_{j}\oplus k_2} \circ \ldots \circ \swap{E^{T_{r-1},K}_{k_0}( x_{r}\oplus k_1)}{y_{r}\oplus k_2} (y_{r-1} \oplus k_2)\\
    &= F_r(y_{r-1} \oplus k_2).
\end{align*}
The last equality holds if the following two properties hold.

\medskip \noindent \textbf{Property $1$: $y_{r-1} \not \in \{y_{r}, \ldots, y_{j} \}$.} Since we don't allow repeated queries, then $x_{r-1} \not \in \{x_{r}, \ldots, x_{j} \}$. Since $R$ is a permutation, $R(x_{r-1}) \not \in \{R(x_{r}), \ldots, R(x_{j})\}$ which gives the desired property.

\medskip \noindent \textbf{Property $2$: $y_{r-1} \oplus k_2 \not \in \left\{E^{T_{r-1},K}_{k_0}( x_{r}\oplus k_1), \ldots, E^{T_{j-1},K}_{k_0}( x_{j}\oplus k_1) \right\}$.} We prove step by step. First by contradiction, suppose $y_{r-1} \oplus k_2 = E^{T_{r-1},K}_{k_0}( x_{r}\oplus k_1)$.  By the inductive assumption, we get
\begin{align*}
    y_{r-1} \oplus k_2 &= F_r^{-1}( y_{r}\oplus k_2)\\
    F_r(y_{r-1} \oplus k_2) &=  y_{r}\oplus k_2\\
    E^{T_j,K}_{k_0}(x_{r-1} \oplus k_1) & = E^{T_j,K}_{k_0}(x_{r} \oplus k_1).
\end{align*}
Since $E^{T_j,K}_{k_0}$ is a permutation, this contradicts the fact that $x_{r-1} \neq x_{r}$. Therefore we have $y_{r-1} \oplus k_2 \neq E^{T_{r-1},K}_{k_0}( x_{r}\oplus k_1)$. 
Now, we can write
\begin{align*}
    &E^{T_j,K}_{k_0}(x_{r-1} \oplus k_1)\\
    &= \swap{E^{T_{j-1},K}_{k_0}( x_{j}\oplus k_1)}{y_{j}\oplus k_2} \circ \ldots \circ \swap{E^{T_{r},K}_{k_0}( x_{r+1}\oplus k_1)}{y_{r+1}\oplus k_2} (y_{r-1} \oplus k_2)\\
    &= F_{r+1}(y_{r-1} \oplus k_2).
\end{align*}
By a similar prove-by-contradiction argument, using the inductive assumption on $r+1$, we obtain
\begin{align*}
    y_{r-1} \oplus k_2 \neq E^{T_{r},K}_{k_0}( x_{r+1}\oplus k_1).
\end{align*}
Repeating this proof by contradiction step by step for $r+2, r+3, \dots, j$, we eventually reach
\begin{align*}
    y_{r-1} \oplus k_2 \neq E^{T_{j-1},K}_{k_0}(x_j \oplus k_1),
\end{align*}
which establishes Property 2.

By these two properties, we can get
\begin{align*}
    E^{T_j,K}_{k_0}(x_{r-1} \oplus k_1)= y_{r-1} \oplus k_2.
\end{align*}
Then by strong induction, the proposition holds for all $i \in \{1, \ldots, j\}$.
\qed
\end{proof}

\subsection{Handling Inverse Queries} \label{sec:inv-queries}

In this section, we extend our analysis to the setting where the distinguisher may issue classical queries in the inverse direction. Conceptually, the overall hybrid structure remains unchanged, i.e., the sequence of hybrids $(\Hyb_0, \dots, \Hyb_{q_C})$ is identical in both the forward and inverse query settings. This is because each hybrid depends only on the transcript $T_j = \{(x_1, y_1), \dots, (x_j, y_j)\}$, which records the pairs of classical queries made so far. The transcript itself is invariant under the query direction: whether the adversary queried $\fx_K(x_i) = y_i$ or $\fx_K^{-1}(y_i) = x_i$, the resulting set of pairs $(x_i, y_i)$ is the same. 

Consequently, the hybrids $\Hyb_j$ remain identical in both settings. The only difference arises in the definition of the intermediate hybrids, where the classical queries are handled in different directions. In particular, the query direction affects how the reprogramming is applied within each step, and thus the intermediate transitions between $\Hyb_j$ and $\Hyb_{j+1}$ must be redefined accordingly.

We now show that these inverse queries can be handled in a fully symmetric manner. Specifically, for each pair of adjacent hybrids $\Hyb_j$ and $\Hyb_{j+1}$, we consider the $(j+1)^{\text{st}}$ classical query as an inverse query, which is answered by the inverse oracle. To formalize this, we first introduce the notation for inverse queries. Using the fact  $$E_{k_0}^{-1}\circ \swap{a}{b}=\swap{E^{-1}_{k_0}(a)}{E^{-1}_{k_0}(b)}\circ E_{k_0}^{-1},$$ 
 we obtain  
\begin{align*}
    (E_{k_0}^{T_1,k})^{-1}&= (\swap{E_{k_0}(x_1\oplus k_1)}{y_1\oplus k_2}\circ E_{k_0})^{-1} \\
    &=E_{k_0}^{-1}\circ \swap{E_{k_0}(x_1\oplus k_1)}{y_1\oplus k_2} \\
    &=\swap{x_1\oplus k_1}{E_{k_0}^{-1}(y_1\oplus k_2)}\circ E^{-1}_{k_0}.
\end{align*}

Extending to $j$~swaps, we get
\begin{align*}
    (E_{k_0}^{T_j,k})^{-1}(y)
    &=E_{k_0}^{-1}\circ \swap{E_{k_0}(x_1\oplus k_1)}{y_1\oplus k_2}\circ \cdots \circ \swap{E^{T_{j-1},k}_{k_0}(x_j\oplus k_1)}{y_j\oplus k_2} (y)\\
    &=(E_{k_0}^{T_{j-1},k})^{-1} \circ \swap{E^{T_{j-1},k}_{k_0}(x_j\oplus k_1)}{y_j\oplus k_2} (y) \\
    &= \swap{x_j\oplus k_1}{(E^{T_{j-1},k}_{k_0})^{-1}(y_j\oplus k_2)}\circ (E_{k_0}^{T_{j-1},k})^{-1}(y)\\
    &=\prod_{i=j}^1 \swap{x_i\oplus k_1}{(E^{T_{i-1},k}_{k_0})^{-1}(y_i\oplus k_2)}\circ E^{-1}_{k_0}(y).
\end{align*}

Moreover, define

\begin{align*}
            (E_{k^*}^{(1)})^{-1}(y)=
    	\begin{cases}
    		E^{-1}_{k^*}(x) &\text{if } k^*\neq k_0\\
    		\left(E^{-1}_{k_0} \circ \swap{s_0}{s_1} \right)(y) &\text{if } k^*=k_0 \text{.}
    	\end{cases}
    \end{align*}

Analogous to the forward case, we first define the intermediate hybrids $\Hyb^{1,\text{inv}}_j$ and $\Hyb^{2,\text{inv}}_j$ ($\Hyb^j_3$ is the same in both cases). We then show that the distinguishing advantage between $\Hyb_j$ and $\Hyb_{j+1}$ remains the same in the inverse case. Specifically, for all $0 \le j < q_C$, we establish the following bounds:
    \begin{align*}
       \text{\expref{Lemma}{lma:fx-Hj-Hj1-inv}:}  &\quad | \Pr[\A(\Hyb_j)=1] - \Pr[\A(\Hyb_j^{1,\text{inv}})=1]| \leq 4\sqrt{\frac{q_Q}{2^m(2^n-j)}} \\
        \text{\expref{Lemma}{lma:fx-Hj1-Hj2-inv}:} & \quad \Hyb_j^{1,\text{inv}} = \Hyb_j^{2,\text{inv}}\\
        \text{\expref{Lemma}{lma:fx-Hj2-Hj3-inv}:} & \quad \Hyb_j^{2,\text{inv}} = \Hyb_j^{3} \\
        \text{\expref{Lemma}{lma:fx-Hj3-Hj+1}:} & \quad | \Pr[\A(\Hyb_j^{3})=1] - \Pr[\A(\Hyb_{j+1})=1]| \leq 2 \cdot q_{Q,j+1} \sqrt{\frac{2(j+1)}{2^{m+n}}},
    \end{align*}
     where $q_{Q,j+1}$ is the expected number of queries $\A$ makes to $P$ in the $(j+1)^{\text{st}}$ stage, i.e., the stage between the $(j+1)^{\text{st}}$ and $(j+2)^{\text{nd}}$ classical queries. By \expref{Equation}{eqn:bound}, we obtain the same bound as in the forward case. Hence, inverse queries can be handled in a fully symmetric manner, yielding the same distinguishing advantage.

\medskip \noindent \textbf{Experiment $\Hyb_j^{1,\text{inv}}$.} \vspace{-5pt}
    \begin{enumerate}
        \item Run $\A$, answering its classical queries using $R$ and its quantum queries using $E$, until $\A$ makes its $(j+1)^{\text{st}}$ classical query $y_{j+1}$, which we assume for concreteness to be in the inverse direction. Let $T_j$ be the list of the classical queries so far.
        \item Define set $S= \bool^n \backslash \{y_1\oplus k_2,\cdots,y_j\oplus k_2\}$. Choose uniform $s \in S$, and define $(E^{(1)})^{-1}$ as
        \begin{align*}
            (E_{k^*}^{(1)})^{-1}(y)=
    	\begin{cases}
    		E^{-1}_{k^*}(y) &\text{if } k^*\neq k_0\\
    		\left(E^{-1}_{k_0} \circ \swap{y_{j+1}\oplus k_2}{s} \right)(y) &\text{if } k^*=k_0 \text{.}
    	\end{cases}
        \end{align*}
    Then 
     \begin{align*}
            (E_{k^*}^{(1)})(x)=
    	\begin{cases}
    		E_{k^*}(x) &\text{if } k^*\neq k_0\\
    		\left(E_{k_0} \circ \swap{E^{-1}_{k_0}(y_{j+1}\oplus k_2)}{E^{-1}_{k_0}(s)} \right)(x) &\text{if } k^*=k_0 \text{.}
    	\end{cases}
        \end{align*}

        Continue running $\A$, answering its remaining classical queries (including the $(j+1)^{\text{st}}$) using $\fx_K[(E^{(1)})^{T_j,K}]$, and its quantum queries using $(E^{(1)})^{T_j,K}$.
    \end{enumerate}

    \medskip\noindent \textbf{Experiment $\Hyb_j^{2,\text{inv}}$.} \vspace{-5pt}
    \begin{enumerate}
        \item Run $\A$, answering its classical queries using $R$ and its quantum queries using $E$, until $\A$ makes its $(j+1)^{\text{st}}$ classical query $y_{j+1}$. Let $T_j$ be the list of classical queries so far.
        \item Define $E^{(1)}$ as in $\Hyb_j^{1,\text{inv}}$ and answer the  $(j+1)^{\text{st}}$ classical query using $\fx_K[(E^{(1)})^{T_j,K}]$. Denote the response as $x_{j+1}$, i.e., 
        \begin{align*}
            x_{j+1} = (\fx_K[(E^{(1)})^j])^{-1} (y_{j+1}).
        \end{align*}
        Then the challenger constructs $E^{j+1}$ adaptively as in \expref{Equation}{eqn:adaptswap}.
        Continue running $\A$, answering its remaining classical queries using $\fx_K[E^{T_{j+1},K}]$, and its quantum queries using $E^{T_{j+1},K}$.
    \end{enumerate}

Next, we prove \expref{Lemma}{lma:fx-Hj-Hj1-inv}, \expref{Lemma}{lma:fx-Hj1-Hj2-inv} and \expref{Lemma}{lma:fx-Hj2-Hj3-inv}. The proofs are proceeded in an entirely analogous manner to those of \expref{Lemma}{lma:fx-Hj-Hj1}, \expref{Lemma}{lma:fx-Hj1-Hj2} and \expref{Lemma}{lma:fx-Hj2-Hj3}.

\begin{lemma}\label{lma:fx-Hj-Hj1-inv}
For $j=0, \ldots, \formerqC$, 
\begin{align*}
    \left|\Pr[\A(\Hyb_j) = 1] - \Pr[\A(\Hyb^{1,\text{inv}}_j)=1]\right| \leq 4\sqrt{\frac{q_Q}{2^m(2^n-j)}}\,.
\end{align*}
\end{lemma}
\begin{proof}
In this lemma, we bound the distinguishability of $\Hyb_j$ and $\Hyb_j^{1,\text{inv}}$.
\begin{alignat*}{5}
\Hyb_j: \;\; &\ket{E}, R, \ket{E}, \cdots, R, \ket{E}, 
~&& \fx_K[~~~E^j~~], \ket{E^j}~~~~~~\;, 
\fx_K \left[E^j\right], \ket{E^j}, \cdots\\
\Hyb_j^{1,\text{inv}}: \;\; &\ket{E}, R, \ket{E}, \cdots, R, \ket{E},  ~&&\fx_K[(E^{(1)})^j], \ket{(E^{(1)})^j} \;, 
\fx_K[(E^{(1)})^j], (E^{(1)})^j, \cdots.
\end{alignat*}
Let $\A$ be a distinguisher between $\Hyb_j$ and $\Hyb_j^{1,\text{inv}}$. 
We construct from $\A$ a distinguisher $\D$ for the resampling experiment of \expref{Lemma}{lem:resampling-ic}.
$\D$ does:	
	\begin{description}
		\item[Phase 1:]
		$\D$ is given quantum access to an ideal cipher~$E$.
		It samples a uniform $R\from \algo P_n$ and then runs $\A$, answering its quantum queries with $E$ and its classical queries with $R$ (in the appropriate directions), until $\A$ submits its \mbox{$(j+1)^{\text{st}}$} classical query~$y_{j+1}$.
		At that point, $\D$ has a list $T_j=\{(x_1, y_1), \cdots, (x_j, y_j)\}$ of the queries/answers $\A$ has made to its classical oracle thus far. Next, $\D$ constructs a distribution $D$ on $\bool^{m+2n}$ and it's sampling algorithm $\Pi$. To sample a tuple $(a_0,z_0,z_1)\leftarrow D$, $\Pi$ does the following:
  \begin{description}
      \item[1]:  Sample uniform $a_0\in \bool^m$ and $z_0 \in \bool^n$.
      \item[2]: Construct $S=\bool^n \backslash \{z_0\oplus y_1\oplus y_{j+1},\cdots, z_0 \oplus y_j \oplus y_{j+1}\}=\bool^n \backslash S^{\bot}.$
      \item[3]: Sample uniform $z_1\in S$, and output $(a_0,z_0,z_1)$.
  \end{description}
    \item[Phase 2:]  
    $\D$ is given $(k_0,s_0,s_1)\leftarrow D$ 
    and quantum oracle access to a cipher~$E^{(b)}$.
    Then $\D$  
    sets $k_2=y_{j+1}\oplus s_0$, uniform $k_1 \in \bool^n$ and $k=(k_0,k_1,k_2)$. 
    It then continues running $\A$, answering its remaining classical queries (including the $(j+1)^{\text{st}}$) using $\fx_K[(E^{(b)})^{T_j,k}]$, and its remaining quantum queries using $(E^{(b)})^{T_j,k}$.	
    $\D$ outputs whatever $\A$ does. 
    \end{description}

 Note that $s_1 \notin S^{\bot}$ where $ S^{\bot}=\{y_1\oplus k_2, \cdots y_j \oplus k_2\}$, by algorithm $\Pi$. In phase~1, distinguisher $\D$ perfectly simulates experiments $\Hyb_j$ and $\Hyb_j^{1,\text{inv}}$ for $\A$ until the point where $\A$ makes its $(j+1)^{\text{st}}$ classical query.
If $b=0$, $\D$ gets access to $E^{(0)} = E$ in phase~2.
Since $\D$ answers all quantum queries using $(E^{(0)})^{T_j,k}$ and all classical queries using $\fx_K[(E^{(0)})^{T_j,k}]$, we see that $\D$ perfectly simulates $\Hyb_j$ for~$\A$ in that case.
If, on the other hand, $b=1$ in phase~2, then $\D$ gets access to $E^{(1)}$, where

        \begin{align*}
            (E_{k^*}^{(1)})^{-1}(y)=
    	\begin{cases}
    		E^{-1}_{k^*}(x) &\text{if } k^*\neq k_0\\
    		\left(E^{-1}_{k_0} \circ \swap{s_0}{s_1} \right)(y) &\text{if } k^*=k_0 \text{.}
    	\end{cases}
        \end{align*}

Since $k_2 \coloneqq s_0 \oplus y_{j+1}$, it holds that 

\begin{align*}
            (E_{k^*}^{(1)})^{-1}(y)=
    	\begin{cases}
    		E^{-1}_{k^*}(x) &\text{if } k^*\neq k_0\\
    		\left(E^{-1}_{k_0} \circ \swap{k_2\oplus y_{j+1}}{s_1} \right)(y) &\text{if } k^*=k_0 \text{.}
    	\end{cases}
    \end{align*}
Moreover, the fact that $s_0$ (and hence $s_0 \oplus y_{j+1}$), $k_0$ and $k_1$ are uniform implies that $\D$ perfectly simulates $\Hyb_j^{1,\text{inv}}$ for~$\A$. Applying \expref{Lemma}{lem:resampling-ic} thus gives
\begin{align*}
    \left|\Pr[\A(\Hyb_j) = 1] - \Pr[\A(\Hyb_j^{1,\text{inv}})=1]\right|
    \leq 4\sqrt{q_Q\cdot\varepsilon\cdot 2^n},
\end{align*}
where $\epsilon$ is the same as in \expref{Lemma}{lma:fx-Hj-Hj1}. Therefore, we have
\begin{align*}
    \left|\Pr[\A(\Hyb_j) = 1] - \Pr[\A(\Hyb_j^{1,\text{inv}})=1]\right|
    \leq 4\sqrt{\frac{q_Q}{2^m(2^n-j)}}\,.
\end{align*}
\qed
\end{proof}
\begin{lemma}\label{lma:fx-Hj1-Hj2-inv}
For $j=0, \ldots, \formerqC$, 
\begin{align*}
    \Hyb_j^{1,\text{inv}}=\Hyb_j^{2,\text{inv}}\,.
\end{align*}
\end{lemma}
\begin{proof}

We bound the distinguishability of $\Hyb_j^{1,\text{inv}}$ and $\Hyb_j^{2,\text{inv}}$.
\begin{alignat*}{5}
    \Hyb_j^{1,\text{inv}}: \;\; &\ket{E}, R, \ket{E}, \cdots, R, \ket{E},  ~&&\fx_K[(E^{(1)})^j], \ket{(E^{(1)})^j} \;, 
    \fx_K[(E^{(1)})^j], (E^{(1)})^j, \cdots \\
    \Hyb_j^{2,\text{inv}}: \;\; &\ket{E}, R, \ket{E}, \cdots, R, \ket{E},
    ~&&\fx_K[(E^{(1)})^j], \ket{E^{j+1}}~~~, 
    \fx_K[E^{j+1}], \ket{E^{j+1}}, \cdots.
\end{alignat*}

To prove $\Hyb_j^{1,\text{inv}}$ and $\Hyb_j^{2,\text{inv}}$ are indistinguishable, it suffices to show that the quantum oracles $(E^{(1)})^j$ and $E^{j+1}$ are identical. 

In both $\Hyb_j^{1,\text{inv}}$ and $\Hyb_j^{2,\text{inv}}$, the response to the $(j+1)^{\text{st}}$ classical query is:
\begin{align*}
    x_{j+1} &\stackrel{\rm def}{=} (\fx_K[(E^{(1)})^{T_j,k}])^{-1}(y_{j+1}) \\ 
    &= ((E^{(1)}_{k_0})^{T_j,k})^{-1}(y_{j+1} \oplus k_2) \oplus k_1\\
    &= \prod_{i=j}^1 \swap{x_i\oplus k_1}{((E_{k_0}^{(1)})^{T_{i-1},k})^{-1}(y_i\oplus k_2)}\circ (E^{(1)}_{k_0})^{-1}(y_{j+1}\oplus k_2)\oplus k_1\\
    &= \prod_{i=j}^1 \swap{x_i\oplus k_1}{(E^{T_{i-1},k}_{k_0})^{-1}(y_i\oplus k_2)}\circ (E^{-1}_{k_0}\circ \swap{y_{j+1}\oplus k_2}{s}(y_{j+1} \oplus k_2)) \oplus k_1\\
    &= \prod_{i=j}^1 \swap{x_i\oplus k_1}{(E^{T_{i-1},k}_{k_0})^{-1}(y_i\oplus k_2)}\circ E^{-1}_{k_0}(s) \oplus k_1\\
    &= (E_{k_0}^{T_{j},k})^{-1}(s) \oplus k_1,
\end{align*}
where the fourth equality comes from \expref{Proposition}{prop:E(1)=E-fx}.
By rearranging $s =  E_{k_0}^{T_{j},k}(x_{j+1} \oplus k_1),$ it follows that for any $y \in \{0,1\}^{n}$,
\begin{align*}
    ((E^{(1)})^j)^{-1}(y)&=\prod_{i=j}^1 \swap{x_i\oplus k_1}{((E^{(1)}_{k_0})^{j-1})^{-1}(y_i\oplus k_2)}\circ E^{-1}_{k_0}\circ \swap{y_{j+1}\oplus k_2}{s}(y)\\
    &=\prod_{i=j}^1 \swap{x_i\oplus k_1}{(E^{j-1}_{k_0})^{-1}(y_i\oplus k_2)}\circ E^{-1}_{k_0}\circ \swap{y_{j+1}\oplus k_2}{s}(y)\\
    &= (E^j)^{-1}\circ \swap{y_{j+1}\oplus k_2}{E_{k_0}^{j}(x_{j+1} \oplus k_1)}(y) \\
    &= \swap{x_{j+1}\oplus k_1}{(E^j)^{-1}(y_{j+1}\oplus k_2)}\circ (E^{j})^{-1}(y) \\
    &= (E^{j+1})^{-1}(y)
\end{align*}
where the second equality comes from \expref{Proposition}{prop:E(1)=E-fx}. This concludes the proof that $(E^{(1)})^j \equiv E^{j+1}$. Therefore, hybrids $\Hyb_j^{1,\text{inv}}$ and $\Hyb_j^{2,\text{inv}}$ are perfectly indistinguishable:
\begin{align*}
    \Hyb_j^{1,\text{inv}} = \Hyb_j^{2,\text{inv}}\,.
\end{align*}
\qed
\end{proof}
\begin{lemma}\label{lma:fx-Hj2-Hj3-inv}
For $j=0, \ldots, \formerqC$, 
$\Hyb_j^{2,\text{inv}} = \Hyb_j^{3}.$
\end{lemma}
\begin{proof}
\begin{alignat*}{5}
    \Hyb_j^{2,\text{inv}}: \;\; &\ket{E}, R, \ket{E}, \cdots, R, \ket{E},
    ~&&\fx_K[(E^{(1)})^j], \ket{E^{j+1}}~~~, 
    \fx_K[E^{j+1}], \ket{E^{j+1}}, \cdots\\
    \Hyb_j^{3}: \;\; &\ket{E}, R, \ket{E}, \cdots, R, \ket{E}, ~&&~~~~~~~R~~~~~~~,\ket{E^{j+1}}~~~,
    \fx_K[E^{j+1}], \ket{E^{j+1}}, \cdots.
\end{alignat*}
We observe that $\Hyb_j^{2,\text{inv}}$ and $\Hyb_j^{3}$ differ only in the response to the $(j+1)^{\text{st}}$ classical query. In $\Hyb_j^{2,\text{inv}}$, this query is answered using
\begin{align*}
    x_{j+1} \;=\; (\fx_K\!\big[(E^{(1)})^{T_j,K}\big])^{-1}(y_{j+1}) 
\;=\; (E_{k_0}^{T_j,K})^{-1}(s) \oplus k_1,
\end{align*}
as shown in \expref{Lemma}{lma:fx-Hj1-Hj2}. In contrast, in $\Hyb_j^{3}$ the same query is answered with 
\begin{align*}
    x_{j+1} \;=\; R^{-1}(y_{j+1}).
\end{align*}
Using \expref{Proposition}{prop:E-T_jk-y} and following the same strategy as in \expref{Lemma}{lma:fx-Hj2-Hj3}, we conclude that the two hybrids induce identical distributions on $x_{j+1}$. Hence, $\Hyb_j^{2,\text{inv}} = \Hyb_j^{3}.$

\qed
\end{proof}

\section{Detailed Attacks on \lrw and \xext} \label{app:attacks}
We provide a detailed description of the attacks on standardized \lrw and \xext, as outlined in \expref{Section}{sec:attacks}.

\subsection{Classical Distinguishing Attacks}

\subsubsection{Birthday Collision Attack}
As mentioned in \expref{Section}{sec:related_work}, $q = O(2^{n/2})$ queries are required to distinguish \lrw~(\expref{Definition}{def:lrw}) from an ideal tweakable block cipher. The attack is outlined below, given access to a classical oracle:
\begin{enumerate}
    \item Randomly choose $\tau \in \{0,1\}^n$.
    \item Query $\widetilde{O}$ with $(m, \tau)$ for $2^{n/2}$ distinct values of $m$, obtaining pairs $(m,c)$, where $c$ is the response, and store them. 
    \item Randomly choose $\tau' \in \{0,1\}^n$.
    \item Query $\widetilde{O}$ with $(m',\tau')$ for $2^{n/2}$ distinct values of $m'$, obtaining pairs $(m',c')$, where $c'$ is the response, and store them. 
    \item Find a pair $(c,c')$ such that $m \oplus c = m' \oplus c'$.
    \item Based on the construction of $h$, use the identified pair $(c,c')$ (or multiple such pairs) to distinguish between an \textsf{ideal tweakable cipher} and \lrw.
\end{enumerate}

\medskip\noindent \textbf{Example.} Consider the \lrw construction as established in~\cite{P1619D4}, where $h_{k'}(\tau) = k' \times_* \tau$. For this construction, the attack proceeds as follows: 
\begin{enumerate}
    \item Perform steps 1–5 twice to obtain two pairs $(c_1,c_1')$ and $(c_2,c_2')$ corresponding to two pairs of tweaks $(\tau_1,\tau_1')$ and $(\tau_2,\tau_2')$.
    \item Compute $k_1' = (c_1\oplus c_1')\times_*(\tau_1 \oplus \tau_1')^{-1}$ and similarly $k_2'= (c_2\oplus c_2')\times_*(\tau_2 \oplus \tau_2')^{-1}$.  
    \item If $k_1'=k_2'$, output \lrw. Otherwise output \textsf{ideal tweakable cipher}.
\end{enumerate}

\medskip\noindent \textbf{Analysis of the Algorithm.} If $\widetilde{O}$ is \lrw, i.e. $\widetilde{O} = \lrw_{k,k'}$, we define $H_{k'}(m, \tau) := m \oplus h_{k'}(\tau)$. Since $h$ is a universal hash function, for all $m,\tau$ with a uniformly selected $k'$, $H_{k'}(m, \tau)$ is uniformly distributed. Now, recall that a birthday attack can be applied to $H_{k'}$. Specifically, with $2^{n/2}$ randomly selected inputs, there exists (with overwhelming probability) a pair $(m, \tau)$ and $(m',\tau')$ such that 
\begin{align*}
    H_{k'}(m, \tau) = H_{k'}(m', \tau') \quad \Rightarrow \quad m \oplus h_{k'}(\tau) = m' \oplus h_{k'}(\tau').
\end{align*}
This implies:
\begin{align*}
    c \oplus c' =  m \oplus m' = h_{k'}(\tau) \oplus h_{k'}(\tau').
\end{align*}
Substituting $h_{k'}(\tau) = k' \times_* \tau$, we get:
\begin{align*}
    c \oplus c' = k' \times_* \tau \oplus k' \times_* \tau'.
\end{align*}
Using this, we compute the hash key:
\begin{align*}
    k' = (c \oplus c') \times_* \tau \oplus \tau')^{-1}.
\end{align*} The key $k'$ will always be the same for all such collisions.

If $\widetilde{O}$ is an \textsf{ideal tweakable cipher}, then similarly, by the birthday attack, among $2^{n/2}$ randomly selected inputs, there will exist a collision where:
\begin{align*}
    c=\widetilde{P}(m, \tau) = \widetilde{P}(m',\tau')=c'
\end{align*}
i.e. a collision on the independent random permutations $f_0(\tau) = f_1(\tau')$. However, these collisions do not reveal any structure of $k'$; instead, the two pairs of ciphertexts and tweaks provide a random guess of $k'$. Thus, with overwhelming probability, $k_1' \neq k_2'$.

In conclusion, this classical attack requires $O(2^{n/2})$ queries. This attack bound has been mentioned in several related works~\cite{liskov2002tweakable}.  

\medskip\noindent \textbf{Remark.} If $\widetilde{O}$ is \xext, i.e. $h_{k'}(\tau) = h_{k'}(i,j) = \alpha^j \otimes E_{k'}(2^i)$, we also follow step 1-5 twice for two pairs of tweaks $(i,j_1,i,j'_1)$ and $(i,j_2,i,j'_2)$, obtaining the corresponding ciphertext pairs $(c_1,c'_1)$ and $(c_2,c'_2)$. Using these, we compute $(c_1\oplus c'_1)\times_* (\alpha^{j_1} \oplus \alpha^{j'_1})^{-1}$ and $(c_2\oplus c'_2)\times_*(\alpha^{j_2} \oplus \alpha^{j'_2})^{-1}$ and check whether they match. It they do, this value is the candidate for $E_{k'}(2^i)$. Similar arguements work for \xext.


\subsubsection{Even-Mansour Attack}

Here, we use the fact that \lrw with one tweak is an \emr to construct the distinguishing attack:  
\begin{enumerate}
    \item Randomly choose two tweaks $t,\tau'\in \{0,1\}^n$, and define two oracles 
    \begin{align*}
        P(\cdot) := \widetilde{O}(\cdot, \tau), \quad R(\cdot) := \widetilde{O}(\cdot, \tau').
    \end{align*}
    \item Define $f(x) = P(x) \oplus R(x)$. Query $f(\cdot)$ $2^{n/2}$ times and find a collision, i.e., $f(x)=f(x')$.
    \item Compute $\widetilde{k} = x \oplus x'$. 
    \item Randomly select $x^* \in \{0,1\}^n$. Compute $R(x^* \oplus \widetilde{k}) \oplus \widetilde{k}$ and $P(x^*)$.
    \item If the two results are equal, output \lrw; otherwise, output \textsf{ideal tweakable cipher}.
\end{enumerate}

\medskip\noindent \textbf{Analysis of the Algorithm.} If $\widetilde{O}$ is \lrw, i.e. $\widetilde{O} = \lrw_{k,k'}$, then 
    \begin{align*}
        P(x) &= \widetilde{O}(x,\tau) = E_k(x \oplus h_{k'}(\tau)) \oplus h_{k'}(\tau)\\
        R(x) &= \widetilde{O}(x,\widetilde{\tau})= E_k(x \oplus h_{k'}(\widetilde{\tau})) \oplus h_{k'}(\widetilde{\tau}).
    \end{align*}
We can rewrite $P(x)$ as follows:
\begin{align*}
    P(x) &= E_k(x \oplus \underline{h_{k'}(\tau) \oplus h_{k'}(\widetilde{\tau})} \oplus h_{k'}(\widetilde{\tau}))\oplus \underline{h_{k'}(\tau) \oplus h_{k'}(\widetilde{\tau})} \oplus h_{k'}(\widetilde{\tau})\\
    &= E_k(x \oplus \widetilde{k} \oplus h_{k'}(\widetilde{\tau}))\oplus \widetilde{k} \oplus h_{k'}(\widetilde{\tau})\\
    &= R(x \oplus \widetilde{k}) \oplus \widetilde{k},
\end{align*}
This shows that $P$ is an \emr  of $R$. Next, consider 
\begin{align*}
    f(x) &= P(x) \oplus R(x)\\
    &= R(x \oplus \widetilde{k}) \oplus \widetilde{k} \oplus R(x)\\
    &=  R(x \oplus \widetilde{k}) \oplus \widetilde{k} \oplus R(x\oplus \widetilde{k} 
 \oplus \widetilde{k})\\
    &=R(x \oplus \widetilde{k}) \oplus P(x \oplus \widetilde{k})\\
    &= f(x \oplus \widetilde{k}).
\end{align*}
This periodicity allows us to use $2^{n/2}$ queries to find a collision on $f(x) = f(x')$. The \emr key is then computed as $\widetilde{k} = x \oplus x'$. For any $x^* \in \{0,1\}^n$, the \emr relation between $P$ and $R$ would hold.

If $\widetilde{O}$ is an \textsf{ideal tweakable cipher}, then $P(x)$ and $R(x)$ are independent random permutations of input $x$. In this case, $f(x) = P(x) \oplus R(x)$ is the XOR of random permutations. While collisions can still be found using $2^{n/2}$ queries, the computed $\widetilde{k}$ is meaningless. For a random $x^*$, $P(x^*)$ and $R(x^* \oplus \widetilde{k}) \oplus \widetilde{k}$ are independently and uniformly distributed and so will not be equal for overwhelming probability.

We note that this attack works similarly for \xext.

\subsection{Classical Complete Key-Search Attack}
\begin{enumerate}
    \item Randomly select $\tau \in \{0,1\}^n$.
    \item  Use the classical distinguishing attack described above to recover $k'$ (assuming an appropriate hash function construction). 
    \item Construct an oracle $E_k(\cdot)$ such that, for an input $x$, it replies $$E_k(x) = \lrw_{k,k'}(x \oplus h_{k'}(\tau),\tau) \oplus h_{k'}(\tau).$$ 
    \item Perform an exhaustive search to get key $k$.
\end{enumerate}
 
\medskip\noindent \textbf{Analysis of the Algorithm.} If the tweakable block cipher under investigation is standardized \lrw, then in step 2, $k'$ can be computed directly. We then perform an exhaustive search to recover $k$. Specifically, we first query \lrw using a constant number of inputs, such as $(x \oplus h_{k'}(\tau),\tau)$. Then, for each candidate $k$, we query $E_k(\cdot)$ with the same $x$ values and verify whether the equation from step 3 holds. This process requires a total of $O(2^{n/2})$ online queries and $O(2^m)$ offline queries for complete key recovery. 

If the tweakable block cipher under investigation is $\xext$, then in step 2, using the ciphertext pair $(c,c')$ with corresponding tweaks $(i,j),(i,j')$, we can compute $E_{k'}(2^i) = (c\oplus c')\times_*(\alpha^{j} \oplus \alpha^{j'})^{-1}$. Next, we query $E(\cdot,2^i)$ $2^m$ times with different candidate values for $k'$. With high probability $(1-\frac{1}{2^n})$, we identify the correct tweak key $k'^*$ by verifying that $E(k'^*,2^i)= (c\oplus c')\times_*(\alpha^{j} \oplus \alpha^{j'})^{-1}$. Once $k'$ is obtained, we follow a similar process as for \lrw to recover $k$.

Thus, the total number of queries required for complete key recovery remains $O(2^{n/2})$ online queries (plus  $O(2^m)$ offline queries).

\subsection{Quantum Complete Key Recovery Attack (shorter cipher key)} \label{sec:lrw-Q1-attack}

The best-known attack for quantum complete key search is the Offline-Simon attack~\cite{bonnetain2019quantum}. Let $u$ be an integer such that $0 \leq u \leq n$. Define the function $F: \{0,1\}^{m+(n-u)} \times \{0,1\}^u \rightarrow \{0,1\}^n$ as follows:
\begin{align*}
    F(i\|j,x)=E_i(x\|j) \ (i \in \{0,1\}^{m}, j \in \{0,1\}^{n-u}).
\end{align*}
Additionally, define the function $g: \{0,1\}^u \rightarrow \{0,1\}^n$ as 
\begin{align*}
    g(x) = \lrw_{k,k'}(x\|0^{n-u},\tau).
\end{align*} 

Note that $F(k\|h^{(2)}_{k'}(\tau),x)\oplus g(x)$ has the period $h_{k'}^{(1)}(\tau)$ since $$F(k\|h^{(2)}_{k'}(\tau),x)\oplus g(x) = E_k(x \|h^{(2)}_{k'}(\tau))\oplus E_k\left((x\oplus h^{(1)}_{k'}(\tau))\|h_{k'}^{(2)}(\tau)\right)\oplus h_{k'(\tau)}.$$ The procedure for recovering $k$ and $k'$ is as follows:  
\begin{enumerate}
    \item Randomly choose $\tau \in \{0,1\}^n.$
    \item Run \textsf{Alg-ExpQ1} on the defined $F$ and $g$ to recover $k$ and $h_{k'}^{(2)}(\tau)$.
    \item Use \textsf{SimQ1} on $f_{k\|h_{k'}^{(2)}(\tau)}$ to recover $h_{k'}^{(1)}(\tau)$.
    \item (Verification)  Check if $E(k,k')(0^n,\tau) \oplus E_{k}(h_{k'}(\tau))=h_{k'}(\tau)$. If the check fails, repeat steps 2–4.
    \item Based on construction of hash function $h$, compute $k'$. 
\end{enumerate}

\medskip\noindent \textbf{Analysis of the Algorithm.} 
Similar to the classical case, if the tweakable block cipher under investigation is standardized \lrw, $k'$ can be easily computed from $h_{k'}(\tau)$ using $k' =  h_{k'}(\tau) \times_{*} t^{-1}$. So $\widetilde{O}(2^{(n+m)/3})$ queries (with $m \leq 2n$) are required for a complete key search, same as the query complexity of the attack on the \fx construction~\cite{bonnetain2019quantum}.

If the tweakable block cipher under investigation is $\xext$, then $h_{k'}(\tau):= h_{k'}(i,j) = E_{k'}(2^i)\times \alpha^j$ is recovered from step $2$ and $3$. As in the classical case, we can query $E(\cdot,2^i)$ $2^m$ times with different $k'$ to identify the true key that satisfies the matching condition. Thus, $D=O(2^u)$ classical queries are made to $\lrw_{k,k'}(\cdot,\cdot)$ and $T=O(n^32^{(m+n-u)/2}+2^{m/2})=O(n^32^{(m+n-u)/2})$ offline quantum queries are performed. When $D \leq 2n$, the balance is achieved at at $T=D=\widetilde{O}(2^{(n+m)/3})$.

\subsection{Quantum Complete Key Recovery Attack (longer cipher key)} \label{sec:Grover+Kuwakado-Morii}

Another possible quantum attack involves combining Grover's algorithm with the Kuwakado-Morii attack~\cite{kuwakado2012security} by treating the Kuwakado-Morii method as a predicate and using Grover's search to find the cipher key $k$. Detailedly,
\begin{enumerate}
    \item Fix a key $k \in \{0,1\}^m$.
    \item Select $\tau \in \{0,1\}^n$ at random.
    \item Choose $x_i$ ($i=1,\ldots,s$) from $\{0,1\}^n$ at random, ensuring $x_i \neq x_j$ for $i \neq j$ and $\overline{x_i} \notin \{x_i \mid i=1,2,\ldots,s\}$ for any $i$.
    \item For $i = 1,2,\ldots,s$, compute $$d_i = \lrw_{k,k'}(\tau,x_i) \oplus \lrw_{k,k'}(\tau,\overline{x_i})$$ by querying $\lrw_{k,k'}(\cdot, \cdot)$ classically. Then define a set $$D = \{d_i| i = 1,2, \ldots, s\}.$$ Assume all elements in $D$ are distinct for simplicity.
    \item Create a table $S$ containing pairs $(d_i,x_i)$, sorted by $d_i$.
    \item Find $z$ such that $L_D(z)=1$ using the Grover algorithm with the unitary operator $V_{L_D}$, where $L_{D}: \{0,1\}^{n} \rightarrow\{0,1\}$ is defined as
    $$
    L_{\mathcal{D}}(x)= \begin{cases}0 & \text { if } E(k,x) \oplus E(k,\bar{x}) \notin D \\ 1 & \text { if } E(k,x) \oplus E(k,\bar{x})  \in D,\end{cases},
    $$
    and 
    $$V_{L_{D}}|x\rangle=(-1)^{L_{D}(x)}|x\rangle.$$
    \item Compute $$d = E(k,z) \oplus E(k,\bar{z}),$$ then find an $x$ from table $S$ such that $$d = \lrw_{k,k'}(\tau,x) \oplus \lrw_{k,k'}(\tau,\bar{x}).$$
    \item Compute $h_{k'}(\tau) = x \oplus z$.
    \item Select $x' \in \{0,1\}^n$ at random and obtain $\lrw_{k,k'}(\tau,x')$ by querying $\lrw_{k,k'}(\cdot, \cdot)$ classically.
    \item If $E(k,x') \oplus h_{k'}(\tau) \neq \lrw_{k,k'}(\tau,x')$, then set $h_{k'} = \bar{x} \oplus z$.
    \item Repeat step 2--10 three times to obtain three pairs: $(\tau_1, h_{k'}(\tau_1))$, $(\tau_2, h_{k'}(\tau_2))$, and $(\tau_3, h_{k'}(\tau_3))$. If $h_{k'}(\tau_1), h_{k'}(\tau_2), h_{k'}(\tau_3)$ are not uniformly distributed, output $1$; otherwise, output $0$.
    \item Treat steps 2--11 as a predicate $f$, where $f(k) \in \{0,1\}$, and find $k$ such that $f(k) = 1$ using the Grover algorithm with the unitary operator $U_f$.
    \item (Verification) If a valid $k$ is found, compute $k'$ based on the construction of the hash function $h$ and output $1$ or $0$ accordingly.
\end{enumerate}

\medskip\noindent \textbf{Analysis of the Algorithm.} This combined attack employs BHT to recover the tweak term $h_{k'}(\tau)$ and Grover's algorithm to find the cipher key $k$. Therefore, it requires $q_Q = O(2^{m/2} \cdot 2^{n/3}) = O(2^{m/2 + n/3})$ quantum queries and $q_C = O(2^{n/3})$ classical queries.

\section{Post-Quantum Security Proof of \lrw in \expref{Theorem}{thm:Q1-secure-LRW-1}} \label{app:lrw-proof}

\begin{theorem}[Restatement]
    Let \lrw be as in \cref{def:lrw} and let $\A$ be an adversary making $q_C$ classical queries to its first oracle and $q_Q \geq 1$ quantum queries to its second oracle. Assuming $h$ is XOR-universal and uniform, it holds that in the ideal cipher model,
    \begin{align*}
        \left| \Pr_{\substack{k \leftarrow \{0,1\}^m; k' \leftarrow\{0,1\}^{\kappa} \\  E \leftarrow  \mathcal{E}(m,n)}} \left[ \A^{\pm\lrw_{k,k'}[E_{\_}],\ket{\pm E}} = 1\right]  - \Pr_{\substack{\widetilde{\Pi} \leftarrow  \mathcal{E}(\tweak,n); \\ E \leftarrow  \mathcal{E}(m,n)}} \left[ \A^{\pm\widetilde{\Pi},\ket{\pm E}} = 1\right] \right|  \\ 
       \leq \frac{6 q_C^2}{2^n} + \frac{4}{2^{(m+n)/2}}\left( q_C \sqrt{q_Q}+q_Q \sqrt{q_C}\right).
    \end{align*}
\end{theorem}

\begin{proof}
    We adapt a similar hybrid technique to that used in proving the post-quantum security of \fx (see \expref{Theorem}{thm:Q1-secure-FX}) for the case of \lrw. Similarly, we assume that all classical queries are made in the forward direction. The inverse queries can be handled in a similar manner as \fx. We begin by setting down a way of modifying a given cipher based on a choice of keys and a list of classical queries to an \lrw oracle. Let $E \in \mathcal{E}(m,n)$, $K=(k, k') \in \bool^{m+\kappa}$. Fix a list $\{(\tau_1,x_1,y_1),(\tau_2,x_2,y_2),\ldots, (\tau_{q_C}, x_{q_C},y_{q_C})\}$. Each pair $(\tau_i,x_i,y_i)$ represents an input--output pair of a classical query. Repeated queries are not allowed. For any $j \leq q_C$, let $T_j$ be the list containing the first $j$ queries.

    The ``modified cipher'' after $j$-many queries is denoted $E^{T_j,K}$, and is given by
    \begin{align}
    E_{k^*}^{T_j,K}(x)= \begin{cases}
        E_{k^*}(x) &\text{ if } k^* \neq k\\
        E^{T_j,K}_{k}(x) &\text{ if } k^* = k.
    \end{cases}
    \end{align}
    where $E^{T_j,K}_{k}$ is defined as follows. First, define $E^{T_0,k}=E$, and for all $j \in [1, q_C]$,
    \begin{align} \label{eqn:lrw-Ej}
        E^{T_j,K}_{k}(x) &= \swap{E^{T_{j-1},K}_{k}( x_{j}\oplus h_{k'}(\tau_j))}{y_{j}\oplus h_{k'}(\tau_j)} \circ E^{T_{j-1},K}_{k}(x) \,.
    \end{align}

    Roughly speaking, $E^{T_j,K}$ is the minimal modification of $E$ that is consistent with the forward $(\rightarrow)$ and inverse $(\leftarrow)$ queries from the transcript $T_j$. For compactness, we occasionally write $E^j$ in place of $E^{T_j,K}$ when $T_j$ and $K$ are clear to the context. We also set $E^0=E$. 

    We now define the bad event $\textsf{Bad}_j = \textsf{Bad}_{j,1} \vee \textsf{Bad}_{j,2}$, where:  
    \begin{itemize}
        \item $\textsf{Bad}_{j,1}$: A collision occurs in the set $\{x_1 \oplus h_{k'}(\tau_{1}), \ldots, x_j \oplus h_{k'}(\tau_{j})\}$.
        \item $\textsf{Bad}_{j,2}$: A collision occurs in the set $\{y_1 \oplus h_{k'}(\tau_{1}), \ldots, y_j \oplus h_{k'}(\tau_{j})\}$.
    \end{itemize}
    The probability of the bad event occurring depends on the distributions of $E_k$, $k$ and $T_j$. Given the XOR-universality of $h$ and that $k'$ is sampled uniformly at random (after $T_j$ is fixed), the bad event occurs with probability $\Pr[\textsf{Bad}_j] \leq \frac{j^2}{2^n}$. Conditioned on this event, we derive statements analogous to those in \expref{Proposition}{prop:E-T_jk-y} in the proof of \fx.
    \begin{proposition}\label{prop:lrw-E-T_jk-y}
    For any $E \in \mathcal{E}(m,n)$, $K=(k,k')\in \bool^{m+\kappa}$, $j \in \{1,\ldots,q_C\}$, transcript $T_j = \{(x_1, y_1), \ldots, (x_j, y_j)\}$ without repetition, and any $i \in \{1, \ldots, j\}$, if $\neg \textsf{Bad}_{j}$, \begin{align*}
        E^{T_j,K}_{k}(x_i \oplus h_{k'}(\tau_i) )\oplus h_{k'}(\tau_i) = y_i.
    \end{align*}
\end{proposition} 
\begin{proof}
We prove by strong induction on $i$. We start by the base case $i=j$. 
\begin{align*}
    E^{T_j,K}_{k}(x_j \oplus h_{k'}(\tau_j)) = \swap{E^{T_{j-1},K}_{k}(x_{j}\oplus h_{k'}(\tau_j))}{y_{j}\oplus h_{k'}(\tau_j)} \circ E^{T_{j-1},K}_{k}(x_j \oplus h_{k'}(\tau_j)) = y_j \oplus h_{k'}(\tau_j).
\end{align*}
Assume for all $i \in \{r, \ldots, j\}$ that $E^{T_j,K}_{k}(x_i \oplus h_{k'}(\tau_i)) \oplus h_{k'}(\tau_i) = y_i$. Expanding it, we obtain
\begin{align*}
    &E^{T_j,K}_{k}(x_{i} \oplus h_{k'}(\tau_i))\\
    &= \swap{E^{T_{j},K}_{k}(x_{j}\oplus h_{k'}(\tau_{j}))}{y_{j}\oplus h_{k'}(\tau_{j})} \circ \ldots \circ \swap{E^{T_{i-1},K}_{k}( x_{i}\oplus h_{k'}(\tau_{i}))}{y_{i}\oplus h_{k'}(\tau_{i})} \circ E^{T_{i-1},K}_{k}(x_{i} \oplus h_{k'}(\tau_{i}))\\
    &= F_i \circ E^{T_{i-1},K}_{k}(x_{i} \oplus h_{k'}(\tau_{i})) = y_i \oplus h_{k'}(\tau_{i}),
\end{align*}
where we define
\begin{align*}
    F_i =  \swap{E^{T_{j},K}_{k}(x_{j}\oplus h_{k'}(\tau_{j}))}{y_{j}\oplus h_{k'}(\tau_{j})} \circ \ldots \circ \swap{E^{T_{i-1},K}_{k}( x_{i}\oplus h_{k'}(\tau_{i}))}{y_{i}\oplus h_{k'}(\tau_{i})}.
\end{align*}
By rearranging, we get for all $i \in \{r, \ldots, j\}$,
\begin{align*}
    E^{T_{i-1},K}_{k}(x_{i} \oplus h_{k'}(\tau_{i})) = F_i^{-1} (y_i \oplus h_{k'}(\tau_{i})).
\end{align*}
We now work on case $r-1$, 
\begin{align*}
    &E^{T_j,K}_{k}(x_{r-1} \oplus h_{k'}(\tau_{r-1}))\\
    &= \swap{E^{T_{j-1},K}_{k}( x_{j}\oplus h_{k'}(\tau_{j}))}{y_{j}\oplus h_{k'}(\tau_{j})} \circ \ldots \circ \swap{E^{T_{r-2},K}_{k}( x_{r-1}\oplus h_{k'}(\tau_{r-1}))}{y_{r-1}\oplus h_{k'}(\tau_{r-1})} \\
    & \hspace{1cm} \circ E^{T_{r-2},K}_{k}(x_{r-1} \oplus h_{k'}(\tau_{r-1}))\\
    &= \swap{E^{T_{j-1},K}_{k}(x_{j}\oplus h_{k'}(\tau_{j}))}{y_{j}\oplus h_{k'}(\tau_{j})} \circ \ldots \circ \swap{E^{T_{r-1},K}_{k}( x_{r}\oplus h_{k'}(\tau_{r}))}{y_{r}\oplus h_{k'}(\tau_{r})} (y_{r-1} \oplus h_{k'}(\tau_{r-1}))\\
    &= F_r(y_{r-1} \oplus h_{k'}(\tau_{r-1})).
 \end{align*}
 We now argue the following two properties:

\medskip \noindent \textbf{Property $1$: $y_{r-1} \oplus h_{k'}(\tau_{r-1}) \notin \{y_{r} \oplus h_{k'}(\tau_{r}), \ldots, y_j \oplus h_{k'}(\tau_{j})\}$.} This follows from the assumption $\neg \textsf{Bad}_{j,2}$.

\medskip \noindent \textbf{Property $2$: $y_{r-1} \oplus h_{k'}(\tau_{r-1}) \notin \{E^{T_{r-1},K}_{k}( x_{r}\oplus h_{k'}(\tau_{r})), \ldots, E^{T_{j-1},K}_{k}(x_{j}\oplus h_{k'}(\tau_{j}))\}.$} We prove step by step. First by contradiction, suppose $y_{r-1} \oplus h_{k'}(\tau_{r-1}) = E^{T_{r-1},K}_{k}(x_{r}\oplus h_{k'}(\tau_{r}))$. Then by the inductive assumption, we get
\begin{align*}
    y_{r-1} \oplus h_{k'}(\tau_{r-1}) &= F_r^{-1}(y_{r}\oplus h_{k'}(\tau_{r}))\\
    F_r(y_{r-1} \oplus h_{k'}(\tau_{r-1})) &=  y_{r}\oplus h_{k'}(\tau_{r})\\
    E^{T_j,K}_{k}(x_{r-1} \oplus h_{k'}(\tau_{r-1})) & = E^{T_j,K}_{k}(x_{r} \oplus h_{k'}(\tau_{r})).
\end{align*}
Since $E^{T_j,K}_{k_0}$ is a permutation, this contradicts the assumption ($\neg \textsf{Bad}_{j,1}$) that $x_{r-1}\oplus h_{k'}(\tau_{r-1}) \neq x_{r}\oplus h_{k'}(\tau_{r})$. Therefore we have $y_{r-1} \oplus h_{k'}(\tau_{r-1}) \neq E^{T_{r-1},K}_{k}(x_{r}\oplus h_{k'}(\tau_{r}))$. 
Now, we can write
\begin{align*}
    &E^{T_j,K}_{k}(x_{r-1} \oplus h_{k'}(\tau_{r-1}))\\
    &= \swap{E^{T_{j-1},K}_{k}( x_{j}\oplus h_{k'}(\tau_{j}))}{y_{j}\oplus h_{k'}(\tau_{j})} \circ \ldots \circ \swap{E^{T_{r},K}_{k}( x_{r+1}\oplus h_{k'}(\tau_{r+1}))}{y_{r+1}\oplus h_{k'}(\tau_{r+1})} (y_{r-1} \oplus h_{k'}(\tau_{r-1}))\\
    &= F_{r+1}(y_{r-1} \oplus h_{k'}(\tau_{r-1})).
\end{align*}
By a similar prove-by-contradiction argument, using the inductive assumption on $r+1$, we obtain
\begin{align*}
    y_{r-1} \oplus h_{k'}(\tau_{r-1}) \neq E^{T_{r},K}_{k}( x_{r+1}\oplus h_{k'}(\tau_{r+1})).
\end{align*}
Repeating this proof by contradiction step by step for $r+2, r+3, \dots, j$, we eventually reach
\begin{align*}
    y_{r-1} \oplus h_{k'}(\tau_{r-1}) \neq E^{T_{j-1},K}_{k}(x_j \oplus h_{k'}(\tau_{j})),
\end{align*}
which establishes Property 2.

By these two properties, we can get
\begin{align*}
    E^{T_j,K}_{k}(x_{r-1} \oplus h_{k'}(\tau_{r-1}))= y_{r-1} \oplus h_{k'}(\tau_{r-1}).
\end{align*}
Then by strong induction, the proposition holds for all $i \in \{1, \ldots, j\}$.
\qed
\end{proof}
    
    We now define a sequence 
    \begin{equation}
        \Hyb_0, \Hyb_0^{1}, \Hyb_0^{2}, \Hyb_0^{3}, \Hyb_1, \Hyb_1^1,  \dots, \Hyb_{q_C}^{3}
    \end{equation}
    of hybrid experiments. Each experiment begins with sampling uniform $\widetilde{\Pi} \in \mathcal{E}(\mathcal{T},n)$ and $E \in \mathcal{E}(m,n)$, and a uniform key pair $K=(k,k') \in \{0,1\}^{m+\kappa}$. The remaining steps of each hybrid are as follows.
    
    \medskip \noindent \textbf{Experiment $\Hyb_j$.}
    \begin{enumerate}
        \item Run $\A$, answering its classical queries using $\widetilde{\Pi}$ and its quantum queries using $E$, stopping immediately \textit{before} its $(j+1)^{\text{st}}$ classical query. Let $T_j = \{(\tau_1, x_1,y_1),\ldots,(\tau_j,x_j,y_j)\}$ be the list of classical queries so far.
        \item For the remainder of the execution of $\A$, answer its classical queries by $\lrw_K[E^{T_j,K}]$ and its quantum queries by $E^{T_j, k}$.
    \end{enumerate}

    \medskip \noindent \textbf{Experiment $\Hyb_j^{1}$.} 
    \begin{enumerate}
        \item Run $\A$, answering its classical queries using $\widetilde{\Pi}$ and its quantum queries using $E$, until $\A$ makes its $(j+1)^{\text{st}}$ classical query $(\tau_{j+1},x_{j+1})$. 
        \item Choose uniform $s \in \bool^{n}$, and define $E^{(1)}$ as
        \begin{align*}
            E_{k^*}^{(1)}(x)=
    	\begin{cases}
    		E_{k^*}(x) &\text{if } k^*\neq k\\
    		\left(E_{k} \circ \swap{x_{j+1} \oplus h_{k'}(\tau_{j+1})}{s} \right)(x) &\text{if } k^*=k \text{.}
    	\end{cases}
        \end{align*}
        Continue running $\A$, answering its remaining classical queries (including the $(j+1)^{\text{st}}$) using $\lrw_K[(E^{(1)})^{T_j,K}]$, and its quantum queries using $(E^{(1)})^{T_j,K}$.
    \end{enumerate}

    \medskip\noindent \textbf{Experiment $\Hyb_j^{2}$.}
    \begin{enumerate}
        \item Run $\A$, answering its classical queries using $\widetilde{\Pi}$ and its quantum queries using $E$, until $\A$ makes its $(j+1)^{\text{st}}$ classical query $(\tau_{j+1}, x_{j+1})$. 
        \item Define $E^{(1)}$ as in $\Hyb_j^{1}$ and answer the  $(j+1)^{\text{st}}$ using $\lrw_K[(E^{(1)})^{T_j,K}]$. Denote the response as $y_{j+1}$, i.e., 
        \begin{align*}
            y_{j+1} = \lrw_K[(E^{(1)})^j] (\tau_{j+1}, x_{j+1}).
        \end{align*}
        Continue running $\A$, answering its remaining classical queries using $\lrw_K[E^{T_{j+1},K}]$, and its quantum queries using $E^{T_{j+1},K}$.
    \end{enumerate}

    \medskip \noindent \textbf{Experiment $\Hyb^3_j$.}
    \begin{enumerate}
        \item Run $\A$, answering its classical queries using $\widetilde{\Pi}$ and its quantum queries using $E$, stopping immediately \textit{after} its $(j+1)^{\text{st}}$ classical query. Let $T_{j+1} = \{(\tau_1, x_1,y_1),\ldots,(\tau_{j+1}, x_{j+1},y_{j+1})\}$ be the classical queries so far.
        \item For the remainder of the execution of $\A$, answer its classical queries using $\lrw_K\left[E^{T_{j+1},K}\right]$ and its quantum queries using $E^{T_{j+1},K}$, i.e. $\ket{E^{j+1}}$. 
    \end{enumerate}

    We can compactly represent $\left\{\Hyb_j, \Hyb_j^{1}, \Hyb_j^{2}, \Hyb_j^{3}, \Hyb_{j+1}\right\}$ as the experiment in which $\A$'s queries are answered using the following oracle sequences. We also define $(E^{(1)})^j$ to denote $(E^{(1)}_{k})^{T_j,K}$.
    \begin{alignat*}{5}
        \Hyb_j: \;\; &\ket{E}, \widetilde{\Pi}, \ket{E}, \cdots, \widetilde{\Pi}, \ket{E}, 
        ~&& \lrw_K[~~~E^j~~], \ket{E^j}~~~~~~\;, 
        \lrw_K \left[E^j\right], \ket{E^j}, \cdots\\
        \Hyb_j^{1}: \;\; &\ket{E}, \widetilde{\Pi}, \ket{E}, \cdots, \widetilde{\Pi}, \ket{E},  ~&&\lrw_K[(E^{(1)})^j], \ket{(E^{(1)})^j} \;, 
        \lrw_K[(E^{(1)})^j], (E^{(1)})^j, \cdots \\
        \Hyb_j^{2}: \;\; &\ket{E}, \widetilde{\Pi}, \ket{E}, \cdots, \widetilde{\Pi}, \ket{E},
        ~&&\lrw_K[(E^{(1)})^j], \ket{E^{j+1}}~~~, 
        \lrw_K[E^{j+1}], \ket{E^{j+1}}, \cdots\\
        \Hyb_j^{3}: \;\; &\ket{E}, \widetilde{\Pi}, \ket{E}, \cdots, \widetilde{\Pi}, \ket{E}, ~&&~~~~~~~~\widetilde{\Pi}~~~~~~~~,\ket{E^{j+1}}~~~,
        \lrw_K[E^{j+1}], \ket{E^{j+1}}, \cdots\\
        \Hyb_{j+1}: \;\; &\underbrace{\ket{E}, \widetilde{\Pi}, \ket{E}, \cdots, \widetilde{\Pi}, \ket{E}}_{j \text { classical queries }}, 
        ~&&\underbrace{~~~~~~~~\widetilde{\Pi}~~~~~~~~,\ket{E}~~~~~~~}_{(j+1)^{\text{st}} \text{ classical query}}, 
        \underbrace{\lrw_K[E^{j+1}], \ket{E^{j+1}}, \cdots}_{q_C-j-1 \text { classical queries }}.
    \end{alignat*}
    
    Then we establish the following bounds on the distinguishability of $\Hyb_j$ and $\Hyb_{j+1}$, step by step, for $0 \leq j < q_C$:
    \begin{align*}
        \text{\expref{Lemma}{lma:lrw-Hj-Hj1}:}  &\quad | \Pr[\A(\Hyb_j)=1 \land \neg \textsf{Bad}_j] - \Pr[\A(\Hyb_j^{1})=1 \land \neg \textsf{Bad}_j]| \leq 4 \sqrt{\frac{ q_Q}{2^{m+n}}}\\
        \text{\expref{Lemma}{lma:lrw-Hj1-Hj2}:}  &\quad | \Pr[\A(\Hyb_j^{1})=1\land \neg \textsf{Bad}_j] - \Pr[\A(\Hyb_j^{2})=1\land \neg \textsf{Bad}_j]| \leq \frac{2j}{2^n}\\
        \text{\expref{Lemma}{lma:lrw-Hj2-Hj3}:}  &\quad | \Pr[\A(\Hyb_j^{2})=1\land \neg \textsf{Bad}_j] - \Pr[\A(\Hyb_j^{3})=1\land \neg \textsf{Bad}_j]| \leq \frac{2j}{2^n}\\
        \text{\expref{Lemma}{lma:lrw-Hj3-Hj+1}:}  &\quad | \Pr[\A(\Hyb_j^{3})=1\land \neg \textsf{Bad}_j] - \Pr[\A(\Hyb_{j+1})=1\land \neg \textsf{Bad}_{j+1}]| \leq 2 \cdot q_{Q,j+1} \sqrt{\frac{2(j+1)}{2^{m+n}}},
    \end{align*}
    where the last equation, $q_{Q,j+1}$ is the expected number of queries $\A$ makes to $P$ in the $(j+1)^{\text{st}}$ stage. 

    Using the above, we have 
    \begin{align*}
        &| \Pr[\A(\Hyb_0)=1] - \Pr[\A(\Hyb_{q_C})=1]| \\
        &= | \Pr[\A(\Hyb_0)=1] - \Pr[\A(\Hyb_{q_C})=1 \land \neg \textsf{Bad}_{q_C}] - \Pr[\A(\Hyb_{q_C})=1 \land \textsf{Bad}_{q_C}]| \\
        &\leq | \Pr[\A(\Hyb_0)=1] - \Pr[\A(\Hyb_{q_C})=1 \land \neg \textsf{Bad}_{q_C}]| + |\Pr[\textsf{Bad}_{q_C}]| \\
        &\leq \sum_{j=0}^{q_C-1} | \Pr[\A(\Hyb_j)=1 \land \neg \textsf{Bad}_{j}] - \Pr[\A(\Hyb_{j+1})=1 \land \neg \textsf{Bad}_{j+1}]| + |\Pr[\textsf{Bad}_{q_C}]| \\
        &\leq \sum_{j=0}^{q_C-1} \left(4\sqrt{\frac{q_Q}{2^{m+n}}} + \frac{4j}{2^n} + 2 \cdot q_{Q,j+1} \sqrt{\frac{2(j+1)}{2^{m+n}}}\right) + \frac{2q_C^2}{2^n}\\
        &\leq \frac{6q^2_C}{2^n} + \sum_{j=0}^{q_C-1} \left(4 \sqrt{\frac{q_Q}{2^{m+n}}} + 2\sqrt{2} \cdot q_{Q,j+1} \sqrt{\frac{q_C}{2^{m+n}}} \right)\\
        &\leq \frac{6q^2_C}{2^n} + \frac{ 4}{\sqrt{2^{m+n}}}\cdot \left( q_C \sqrt{q_Q} + q_{Q} \sqrt{q_C} \right),
    \end{align*}
    which concludes the proof.
    \qed
\end{proof}
We now prove \expref{Lemma}{lma:lrw-Hj-Hj1}, \expref{Lemma}{lma:lrw-Hj1-Hj2}, \expref{Lemma}{lma:lrw-Hj2-Hj3} and \expref{Lemma}{lma:lrw-Hj3-Hj+1}.

\begin{lemma}\label{lma:lrw-Hj-Hj1}
For $j=0, \ldots, q_C$, 
\begin{align*}
    \left|\Pr[\A(\Hyb_j ) = 1 \land \textsf{Bad}_{j}] - \Pr[\A(\Hyb^{1}_j)=1 \land \textsf{Bad}_{j}]\right| \leq 4 \sqrt{\frac{q_Q}{2^{n+m}}}\,.
\end{align*}
\end{lemma}
\begin{proof}
In this lemma, we bound the distinguishability of $\Hyb_j$ and $\Hyb_j^{1}$.
\begin{alignat*}{5}
\Hyb_j: \;\; &\ket{E}, \widetilde{\Pi}, \ket{E}, \cdots, \widetilde{\Pi}, \ket{E}, 
~&& \lrw_K[~~~E^j~~], \ket{E^j}~~~~~~\;, 
\lrw_K \left[E^j\right], \ket{E^j}, \cdots\\
\Hyb_j^{1}: \;\; &\ket{E}, \widetilde{\Pi}, \ket{E}, \cdots, \widetilde{\Pi}, \ket{E},  ~&&\lrw_K[(E^{(1)})^j], \ket{(E^{(1)})^j} \;, 
\lrw_K[(E^{(1)})^j], (E^{(1)})^j, \cdots.
\end{alignat*}

Let $\A$ be a distinguisher between $\Hyb_j$ and $\Hyb_j^{1}$. We construct from $\A$ a distinguisher $\D$ for the resampling experiment of \expref{Lemma}{lem:resampling-ic}. Fixing $D$ to be uniform distribution over $\{0,1\}^{m+2n}$ (so $\varepsilon = 2^{-(m+2n)}$ in \Cref{lem:resampling-ic}), $\D$ does:

\begin{description}
    \item[Phase 1:]
        $\D$ is given quantum access to an ideal cipher $E$. It samples a uniform $\widetilde{\Pi} \leftarrow \mathcal{E}(\mathcal{T},n)$ and then runs $\A$, answering its quantum queries with $E$ and its classical queries with $\widetilde{\Pi}$ (in the appropriate directions), until $\A$ submits its $(j+1)^{\text{st}}$ classical query $(\tau_{j+1},x_{j+1})$. At that point, $\D$ has a list $T_j = 
        \{(\tau_1,x_1,y_1),\ldots,(\tau_j,x_j,y_j)\}$ of the queries/answers $\A$ has made to its classical oracle thus far. If $\textsf{Bad}_j$, then set $\textsf{flag}$ to 1; otherwise, set $\textsf{flag}$ to 0.
    \item[Phase 2:]  
        $\D$ is given $(k,s_0,s_1)\leftarrow D$ 
        and quantum oracle access to a cipher~$E^{(b)}$.
        Then $\D$ selects the $k'$ conditioned on $h_{k'}(\tau_{j+1})=x_{j+1}\oplus s_0$ (at least one such $k'$ must exist since $\kappa \geq n$), and sets $K=(k, k')$. It then continues running $\A$, answering its remaining classical queries (including the $(j+1)^{\text{st}}$) using $\lrw_K[(E^{(b)})^{T_j,K}]$, and its remaining quantum queries using $(E^{(b)})^{T_j,K}$. If $\textsf{flag}=1$, $\D$ outputs $0$; otherwise, $\D$ outputs whatever $\A$ does. 
\end{description}

    Note that in phase $1$, distinguisher $\D$ perfectly simulates experiments $\Hyb_j$ and $\Hyb^{1}_j$ for $\A$ until the point where $\A$ makes its $(j+1)^{\text{st}}$ classical query. 
    In phase~2, we now argue that $k'$ is uniformly distributed. For any $k^* \in \{0,1\}^{\kappa}$,
    \begin{align*}
        \Pr_{k':h_{k'}(\tau_{j+1}) = x_{j+1} \oplus s_0}[k' = k^*] &= \Pr_{k'\in \{0,1\}^{\kappa}}[k' = k^*|h_{k'}(\tau_{j+1}) = x_{j+1} \oplus s_0] \\
        &= \frac{\Pr_{k'\in \{0,1\}^{\kappa}}[h_{k'}(\tau_{j+1}) = x_{j+1} \oplus s_0 |k' = k^*] \cdot \Pr_{k'\in \{0,1\}^{\kappa}}[k'=k^*]}{\Pr_{k'\in \{0,1\}^{\kappa}}[h_{k'}(\tau_{j+1}) = x_{j+1} \oplus s_0]}\\
        &= \frac{\Pr_{k'\in \{0,1\}^{\kappa}}[h_{k'}(\tau_{j+1}) = x_{j+1} \oplus s_0 |k' = k^*] \cdot \Pr_{k'\in \{0,1\}^{\kappa}}[k'=k^*]}{\sum_{\widetilde{k}\in\{0,1\}^{\kappa}} \Pr_{k'\in \{0,1\}^{\kappa}}[h_{k'}(\tau_{j+1}) = x_{j+1} \oplus s_0|k' = \widetilde{k}]\Pr_{\widetilde{k}}[k' = \widetilde{k}]}\\
        &= \frac{\Pr_{k'\in \{0,1\}^{\kappa}}[h_{k'}(\tau_{j+1}) = x_{j+1} \oplus s_0 |k' = k^*]}{\sum_{\widetilde{k}\in\{0,1\}^{\kappa}} \Pr_{k'\in \{0,1\}^{\kappa}}[h_{k'}(\tau_{j+1}) = x_{j+1} \oplus s_0|k' = \widetilde{k}]}\\
        &= \frac{\Pr[h_{k^*}(\tau_{j+1}) = x_{j+1} \oplus s_0 ]}{\sum_{\widetilde{k}\in\{0,1\}^{\kappa}} \Pr[h_{\widetilde{k}}(\tau_{j+1}) = x_{j+1} \oplus s_0]}\\
        &=\frac{1/2^{\kappa}}{\sum_{\widetilde{k}} 1/2^{\kappa}}\\
        &= \frac{1}{2^{\kappa}}
    \end{align*}
    where the third equality follows from $\Pr_{k'}[k'=k^*] = \Pr_{k'}[k'=\widetilde{k}] = \tfrac{1}{2^{\kappa}}$. The second last equality holds because $s_0$ is drawn uniformly at random and $h$ is uniform.
    
    If $b=0$ and $\textsf{flag}=0$, $\D$ gets access to $E^{(0)}= E$ in phase $2$. Since $\D$ answers all quantum queries using $\left( E^{(0)}\right)^{T_j,K}$ and all classical queries using $\lrw_K\left[\left( E^{(0)}\right)^{T_j,K}\right]$, we see that $\D$ perfectly simulates $\Hyb_j$ for $\A$ in that case. If, on the other hand, $b=1$ and $\textsf{flag}=0$ in phase $2$, then $\D$ gets access to $E^{(1)}$, where 
    \begin{align*}
        E_{k^*}^{(1)}(x)= \begin{cases}E_{k^*}(x) & \text { if } k^* \neq k \\ E_k \circ \swap{s_0}{s_1}(x) & \text { if } k^*=k .\end{cases}
    \end{align*}
    Since $h_{k'}(\tau_{j+1}):=s_0 \oplus x_{j+1}$, it holds that
    \begin{align*}
        E_{k^*}^{(1)}(x)= \begin{cases}E_{k^*}(x) & \text { if } k^* \neq k \\ E_k \circ \swap{x_{j+1} \oplus h_{k'}(\tau_{j+1})}{s_1}(x) & \text { if } k^*=k .\end{cases}
    \end{align*}
    Moreover, the uniformity property of $h$ and the fact that $s_0$ (and then $s_0 \oplus x_{j+1}$) is uniform imply that the joint distribution of $k'$ and $s_0 \oplus x_{j+1}$ is equal to the joint distribution of $\widetilde{k}$ and $h_{\widetilde{k}}(\tau_{j+1})$ for a uniform $\widetilde{k}$. Thus, in this case $\D$ perfectly simulates $\Hyb^{1}_j$ for $\A$. Applying \expref{Lemma}{lem:resampling-ic} thus gives 
    \begin{align*}
        &|\Pr[\A(\Hyb_j)=1 \land \neg \textsf{Bad}_j] - \Pr[\A(\Hyb^{1}_j)=1\land \neg\textsf{Bad}_j]| \\
        &=|\Pr[\D=1 \land \textsf{flag}=0|b=0] - \Pr[\D=1 \land \textsf{flag}=0|b=1]| \\
        &\leq|\Pr[\D=1 |b=0] - \Pr[\D=1 |b=1]| \\
        &\leq 4 \sqrt{2^n \cdot q_Q \cdot \varepsilon}\\
        &\leq 4 \sqrt{\frac{q_Q}{2^{m+n}}},
    \end{align*}
    where the first inequality holds due to the independence of $\D$'s output from $\A$'s output when $\textsf{flag}=1$, and $\Pr[\textsf{flag}=1]$ being the same for $b=0$ and $b=1$.
    \qed
\end{proof}

\begin{lemma}\label{lma:lrw-Hj1-Hj2}
For $j=0, \ldots, q_C$, 
\begin{align*}
    \left|\Pr[\A(\Hyb^{1}_j) = 1 \land \neg \textsf{Bad}_j] - \Pr[\A(\Hyb^{2}_j)=1 \land \neg \textsf{Bad}_j]\right| \leq \frac{2j}{2^n} \,.
\end{align*}
\end{lemma}
\begin{proof}
Next, we bound the distinguishability of $\Hyb_j^{1}$ and $\Hyb_j^{2}$.
\begin{alignat*}{5}
    \Hyb_j^{1}: \;\; &\ket{E}, \widetilde{\Pi}, \ket{E}, \cdots, \widetilde{\Pi}, \ket{E},  ~&&\lrw_K[(E^{(1)})^j], \ket{(E^{(1)})^j} \;, 
    \lrw_K[(E^{(1)})^j], (E^{(1)})^j, \cdots \\
    \Hyb_j^{2}: \;\; &\ket{E}, \widetilde{\Pi}, \ket{E}, \cdots, \widetilde{\Pi}, \ket{E},
    ~&&\lrw_K[(E^{(1)})^j], \ket{E^{j+1}}~~~, 
    \lrw_K[E^{j+1}], \ket{E^{j+1}}, \cdots.
\end{alignat*}
We first prove the following proposition. 
\begin{proposition}\label{prop:lrw-E(1)=E}
    For all $i \in \{1,\dots, j\}$ and all $r \in \{0,\dots, j\}$
    \begin{align*}
        (E_{k}^{(1)})^{T_r,K}(x_i \oplus h_{k'}(\tau_i)) \;=\; E_{k}^{T_r,K}(x_i \oplus h_{k'}(\tau_i)),
    \end{align*}
    when $s \notin \{x_1\oplus h_{k'}(\tau_{1}), \cdots x_j \oplus h_{k'}(\tau_{j})\}$ and $x_{j+1}\oplus h_{k'}(\tau_{j+1}) \notin \{x_1\oplus h_{k'}(\tau_{1}), \cdots x_j \oplus h_{k'}(\tau_{j})\}$. 
\end{proposition}
\begin{proof}
    We will do a proof by induction on $r$. We start with the base case $r=0$. By the constraint on $s$ and $x_{j+1}\oplus h_{k'}(\tau_{j+1})$, we have that for all $i \in \{1,\cdots, j\}$
     \begin{align*}
        (E_{k}^{(1)})^{T_0,K}(x_i \oplus h_{k'}(\tau_{i})) 
        &:= E_{k}^{(1)}(x_i \oplus h_{k'}(\tau_{i}))\\ &= E_{k} \circ \swap{x_{j+1} \oplus h_{k'}(\tau_{j+1})}{s}(x_i \oplus h_{k'}(\tau_{i})) \\
        &=E_{k}(x_i \oplus h_{k'}(\tau_{j+1})).
    \end{align*}
    Assume for some $r-1 \geq 0$ that
    $$
    (E_{k}^{(1)})^{T_{r-1},K}(x_i \oplus h_{k'}(\tau_{i})) = E_{k}^{T_{r-1},K}(x_i \oplus h_{k'}(\tau_{i})).
    $$
    Then by the definition,
    \begin{align*}
        (E^{(1)}_{k})^{T_{r},K}(x_i\oplus h_{k'}(\tau_{i}))&= \swap{(E^{(1)}_{k})^{T_{r-1},K}(x_{r}\oplus h_{k'}(\tau_{r}))}{y_{r}\oplus h_{k'}(\tau_{r})} \circ (E^{(1)}_{k})^{T_{r-1},K}(x_{i}\oplus h_{k'}(\tau_{i}))\\
        &= \swap{E_{k}^{T_{r-1},K}(x_{r} \oplus h_{k'}(\tau_{r}))}{y_{r}\oplus h_{k'}(\tau_{r})} \circ E_{k}^{T_{r-1},K}(x_{i} \oplus h_{k'}(\tau_{i}))\\
        &= E_{k}^{T_{r},K}(x_i\oplus h_{k'}(\tau_{i})).
    \end{align*}
    By induction, the proposition holds for all $r \in \{0,\dots, j\}$.
    \qed
\end{proof}
In particular, for all $i \in \{1, \ldots, j\}$,
\begin{align*}
    (E_{k}^{(1)})^{T_{i-1},K}(x_i \oplus h_{k'}(\tau_{i})) = E_{k_0}^{T_{i-1},K}(x_i \oplus h_{k'}(\tau_{i})).
\end{align*}

To analyze the distinguishability between $\Hyb_j^{1}$ and $\Hyb_j^{2}$, it suffices to examine the distinguishability between the quantum oracles $(E^{(1)})^j$ and $E^{j+1}$. We begin under the constraint that
\begin{itemize}
    \item $s \notin \{x_1 \oplus h_{k'}(\tau_{1}), \ldots, x_j \oplus h_{k'}(\tau_{j})\}$
    \item $x_{j+1}\oplus h_{k'}(\tau_{j+1}) \notin \{x_1 \oplus h_{k'}(\tau_{1}), \ldots, x_j \oplus h_{k'}(\tau_{j})\}.$
\end{itemize}
In this case, the response to the $(j+1)^{\text{st}}$ classical query is identical in both $\Hyb_j^{1}$ and $\Hyb_j^{2}$:
\begin{align*}
    y_{j+1} &\stackrel{\rm def}{=} \lrw_K[(E^{(1)})^{T_j,K}](\tau_{j+1}, x_{j+1}) \\ 
    &= (E^{(1)}_{k})^{T_j,K}(x_{j+1} \oplus h_{k'}(\tau_{j+1})) \oplus h_{k'}(\tau_{j+1})\\
    &=\swap{(E_{k}^{(1)})^{T_{j-1},K}( x_{j}\oplus h_{k'}(\tau_{j}))}{y_{j}\oplus h_{k'}(\tau_{j})} \circ (E_{k}^{(1)})^{T_{j-1},K}(x_{j+1} \oplus h_{k'}(\tau_{j+1})) \oplus h_{k'}(\tau_{j+1})\\
    &= \left(\prod_{i=j}^1\swap{(E_{k}^{(1)})^{T_{i-1},K}( x_{i}\oplus h_{k'}(\tau_{i}))}{y_{i}\oplus h_{k'}(\tau_{i})} \right) \circ E_{k}^{(1)}(x_{j+1} \oplus h_{k'}(\tau_{j+1})) \oplus h_{k'}(\tau_{j+1})\\
    &= \left(\prod_{i=j}^1\swap{(E_{k})^{T_{i-1},K}( x_{i}\oplus h_{k'}(\tau_{i}))}{y_{i}\oplus h_{k'}(\tau_{i})} \right) \circ E_{k}(s) \oplus h_{k'}(\tau_{j+1})\\
    &= E_{k}^{T_{j},K}(s) \oplus h_{k'}(\tau_{j+1}).
\end{align*}

By rearranging $s = \left( E_{k_0}^{T_{j},K}\right)^{-1}(y_{j+1} \oplus h_{k'}(\tau_{j+1})).$ It follows that for any $x \in \{0,1\}^{n}$,
\begin{align*}
    &(E^{(1)})^j(x) \\
    &= \swap{(E^{(1)}_{k})^{{j-1}}(x_{j}\oplus h_{k'}(\tau_{j}))}{y_{j}\oplus h_{k'}(\tau_{j})} \circ \cdots \circ \swap{E_{k}^{(1)}(x_1\oplus h_{k'}(\tau_{1}))}{y_1\oplus h_{k'}(\tau_{1})} \\
    & \hspace{1cm}\circ E_{k} \circ \swap{x_{j+1}\oplus h_{k'}(\tau_{j+1})}{s} (x) \\
    &= \swap{E^{j-1}_{k}( x_{j}\oplus h_{k'}(\tau_{j}))}{y_{j}\oplus h_{k'}(\tau_{j})} \circ \cdots \circ \swap{E_{k}(x_1\oplus h_{k'}(\tau_{1}))}{y_1\oplus h_{k'}(\tau_{1})}\circ E_{k} \circ \swap{x_{j+1}\oplus h_{k'}(\tau_{j+1})}{s} (x) \\
    &= E^{j} \circ \swap{x_{j+1}\oplus h_{k'}(\tau_{j+1})}{(E^j)^{-1}(y_{j+1}\oplus h_{k'}(\tau_{j+1}))}(x) \\
    &= \swap{E^j(x_{j+1}\oplus h_{k'}(\tau_{j+1}))}{y_{j+1}\oplus h_{k'}(\tau_{j+1})} \circ E^j(x)\\
    &=E^{j+1}(x),
\end{align*}
where the second equality comes from \expref{Proposition}{prop:lrw-E(1)=E}. Under these constraints, we have $(E^{(1)})^j \equiv E^{j+1}$, making the hybrids $\Hyb_j^{1}$ and $\Hyb_j^{2}$ perfectly indistinguishable. We then introduce the constraints $\textsf{Bad}_j$, which is independent of the new constraints on $s$ and $x_{j+1} \oplus h_{k'}(\tau_{j+1})$ and has the same probability in two hybrids. Thus, these constraints do not affect the perfect indistinguishability. To bound the distinguishability of the hybrids with $\neg \textsf{Bad}_j$, we analyze the probability of these new constraints occurring:
\begin{align*}
&\left| \Pr[\mathcal{A}(\Hyb^1_j) = 1 \land \neg \textsf{Bad}_j] - \Pr[\mathcal{A}(\Hyb^2_j) = 1 \land \neg \textsf{Bad}_j] \right|\\
&\leq \left| \Pr[\mathcal{A}(\Hyb^1_j) = 1] - \Pr[\mathcal{A}(\Hyb^2_j) = 1] \right| \\
&\leq \Pr[s \in S] + \Pr[x_{j+1} \oplus h_{k'}(\tau_{j+1}) \in S] \\
&= \frac{2j}{2^n},
\end{align*}
where $S = \{x_1 \oplus h_{k'}(\tau_1), \dots, x_j \oplus h_{k'}(\tau_j)\}.$
Since $s$ is uniformly distributed over $\{0,1\}^n$, the probability of the first term is $\tfrac{j}{2^n}$.  
For the second term,  
\begin{align*}
    x_{j+1} \oplus h_{k'}(\tau_{j+1}) \in S 
    \;\;\Rightarrow\;\; \bigcup_{i=1}^{j} \Big\{\, h_{k'}(\tau_{j+1}) \oplus h_{k'}(\tau_i) = x_{j+1} \oplus x_i \,\Big\},
\end{align*}
and by the XOR-universality of $h$, this occurs with probability $\tfrac{j}{2^n}$.
\qed
\end{proof}

\begin{lemma}\label{lma:lrw-Hj2-Hj3}
For $j=0, \ldots, q_C$, 
\begin{align*}
    \left|\Pr[\A(\Hyb^{2}_j) = 1 \land \textsf{Bad}_{j}] - \Pr[\A(\Hyb^{3}_j)=1 \land \textsf{Bad}_{j}]\right| \leq \frac{2j}{2^n}\,.
\end{align*}
\end{lemma}
\begin{proof}
We bound the distinguishability of $\Hyb_j^{2}$ and $\Hyb_j^{3}$.
\begin{alignat*}{5}
    \Hyb_j^{2}: \;\; &\ket{E}, \widetilde{\Pi}, \ket{E}, \cdots, \widetilde{\Pi}, \ket{E},
    ~&&\lrw_K[(E^{(1)})^j], \ket{E^{j+1}}, 
    \lrw_K[E^{j+1}], \ket{E^{j+1}}, \cdots\\
    \Hyb_j^{3}: \;\; &\ket{E}, \widetilde{\Pi}, \ket{E}, \cdots, \widetilde{\Pi}, \ket{E}, ~&&~~~~~~~~\widetilde{\Pi}~~~~~~~~,\ket{E^{j+1}},
    \lrw_K[E^{j+1}], \ket{E^{j+1}}, \cdots.
\end{alignat*}
We observe that $\Hyb_j^{2}$ and $\Hyb_j^{3}$ differ only in the response to the $(j+1)^{\text{st}}$ classical query. In $\Hyb_j^{2}$, as specified in \expref{Lemma}{lma:lrw-Hj1-Hj2}, under the constraint
\begin{itemize}
\item $s \notin \{x_1 \oplus h_{k'}(\tau_{1}), \ldots, x_j \oplus h_{k'}(\tau_{j})\}$,
\item $x_{j+1}\oplus h_{k'}(\tau_{j+1}) \notin \{x_1 \oplus h_{k'}(\tau_{1}), \ldots, x_j \oplus h_{k'}(\tau_{j})\}$,
\end{itemize}
the $(j+1)$-st query is answered as
\begin{align*}
    y_{j+1} = \lrw_K \big[(E^{(1)})^{T_j,K}\big](\tau_{j+1}, x_{j+1}) \;=\; E_{k}^{T_j,K}(s) \oplus h_{k'}(\tau_{j+1}).
\end{align*}
as shown in \expref{Lemma}{lma:lrw-Hj1-Hj2}. We then introduce the event $\neg \textsf{Bad}_j$. Recall \expref{Proposition}{prop:lrw-E-T_jk-y}, under $\neg \textsf{Bad}_j$, we have  
\begin{align*}
    y_i \;=\; E_{k}^{T_j,K}(x_i \oplus h_{k'}(\tau_i)) \oplus h_{k'}(\tau_i), \quad \forall i \in [1,j],
\end{align*}
and since $s \in \{0,1\}^n \setminus \{x_1 \oplus h_{k'}(\tau_{1}), \ldots, x_j \oplus h_{k'}(\tau_{j})\}$, it follows that in $\Hyb_j^{2}$,
\begin{align*}
    y_{j+1} = E_{k}^{T_j,K}(s) \oplus h_{k'}(\tau_{j+1}) \in \{0,1\}^n \setminus \{y_1\oplus h_{k'}(\tau_{1}) \oplus h_{k'}(\tau_{j+1}), \ldots, y_j \oplus h_{k'}(\tau_{j}) \oplus h_{k'}(\tau_{j+1})\}.
\end{align*}
Moreover, because $E_{k}^{T_j,K}$ is a permutation, the mapping
\begin{align*}
    x_i \oplus h_{k'}(\tau_i) \;\mapsto\; y_i = E_{k}^{T_j,K}(x_i \oplus h_{k'}(\tau_i)) \oplus h_{k'}(\tau_i)
\end{align*}
is injective. Since $s$ is chosen uniformly from
$\{0,1\}^n \setminus \{x_1 \oplus h_{k'}(\tau_{1}), \ldots, x_j \oplus h_{k'}(\tau_{j})\}$, $E_{k}^{T_j,K}(s)$ is uniformly distributed over $$\{0,1\}^n \setminus \{E_{k}^{T_j,K}(x_1 \oplus h_{k'}(\tau_{1})), \ldots, E_{k}^{T_j,K}(x_j \oplus h_{k'}(\tau_{j}))\} = \{0,1\}^n \setminus \{y_1 \oplus h_{k'}(\tau_1), \ldots, y_j \oplus h_{k'}(\tau_j)\}.$$ It follows that the output $y_{j+1}$ is uniformly distributed over $\{0,1\}^n \setminus \{y_1\oplus h_{k'}(\tau_{1}) \oplus h_{k'}(\tau_{j+1}), \ldots, y_j \oplus h_{k'}(\tau_{j}) \oplus h_{k'}(\tau_{j+1})\}$ in $\Hyb_j^2$. 

In contrast, in $\Hyb_j^{3}$ the same query is answered with 
\begin{align*}
    y_{j+1} \;=\; \widetilde{\Pi}(\tau_{j+1}, x_{j+1}).
\end{align*}

Next, we conduct a case-by-case analysis to demonstrate that $y_{j+1}$ has the same distribution conditioned on all previous queries is identical in both $\Hyb_j^{2}$ and $\Hyb_j^{3}$.
\begin{itemize}
\item \textbf{Case} $\tau_{j+1} \notin \{\tau_1, \ldots, \tau_j\}$. In $\Hyb_j^{2}$, since $h$ is XOR-universal, for all $i \in [1,j]$, the value $h_{k'}(\tau_{i}) \oplus h_{k'}(\tau_{j+1})$ is uniformly distributed over $\{0,1\}^n$. Consequently, $y_{j+1}$ is uniformly distributed over $\{0,1\}^n$. In $\Hyb_j^{3}$, as $\tau_{j+1}$ is new, $\widetilde{\Pi}_{\tau_{j+1}}$ represents a fresh permutation, and thus $y_{j+1} = \widetilde{\Pi}_{\tau_{j+1}}(x_{j+1})$ is uniformly distributed over $\{0,1\}^n$.
\item \textbf{Case} $\tau_{j+1} \in \{\tau_1, \ldots, \tau_j\}$. Assume $\tau_{j+1} = \tau_a$ for some $a \in [1,j]$. In $\Hyb_j^{2}$, we have $y_{j+1} \neq y_a \oplus h_{k'}(\tau_a) \oplus h_{k'}(\tau_a) = y_a$. For the remaining elements, since $h$ is XOR-universal, they are uniformly distributed over $\{0,1\}^n$. Thus, $y_{j+1}$ is uniformly distributed over $\{0,1\}^n \setminus \{y_a\}$. In $\Hyb_j^{3}$, note that repeated queries are not permitted, so $x_{j+1} \neq x_a$. Therefore, $y_{j+1} = \widetilde{\Pi}_{a}(x_{j+1}) \neq \widetilde{\Pi}_{a}(x_a) = y_a$. Consequently, $y_{j+1}$ is also uniformly distributed over $\{0,1\}^n \setminus \{y_a\}$. We note that $\tau_{j+1}$ may equal multiple elements in $\{\tau_1, \ldots, \tau_j\}$, but the argument still holds.
\end{itemize}
Moreover, in both $\Hyb_j^{2}$ and $\Hyb_j^{3}$, the construction of $E^{j+1}$ follows exactly the same procedure, i.e., from the first $j+1$ classical input–output pairs and $E$ as specified in \expref{Equation}{eqn:lrw-Ej}. Consequently, the two hybrids yield identical distributions under constraints described before.

Therefore, the probability with which $\A$ can distinguish $\Hyb^{2}_j$ and $\Hyb^{3}_j$ is 
\begin{align*}
    &|\Pr[\A(\Hyb^{2}_j)=1 \land \neg \textsf{Bad}_j] - \Pr[\A(\Hyb^{3}_j)=1\land \neg \textsf{Bad}_j]| \\
    &\leq|\Pr[\A(\Hyb^{2}_j)=1 | \neg \textsf{Bad}_j] - \Pr[\A(\Hyb^{3}_j)=1| \neg \textsf{Bad}_j]| \\
    & \leq \Pr[s \in S] + \Pr[x_{j+1} \oplus h_{k'}(\tau_{j+1}) \in S] \\
    &= \frac{2j}{2^n},
    \end{align*}
where $S = \{x_1 \oplus h_{k'}(\tau_1), \dots, x_j \oplus h_{k'}(\tau_j)\}.$
\qed
\end{proof}

\begin{lemma}\label{lma:lrw-Hj3-Hj+1}
    For $j=0, \ldots, q_C-1$, 
    \begin{align*}
	\Pr[\A(\Hyb_j^{3}) = 1 \land \textsf{Bad}_{j}] - \Pr[\A(\Hyb_{j+1})=1 \land \textsf{Bad}_{j+1}]| \leq 2 \cdot q_{Q,j+1} \sqrt{\frac{2 (j+1)}{2^{m+n}}},
    \end{align*}
    where $q_{Q,j+1}$ is the expected number of queries $\A$ makes to $P$ in the $(j+1)^{\text{st}}$ stage  in the ideal world (i.e., in $\Hyb_{q_C}$). 
\end{lemma}
\begin{proof}
    Let $\A$ be a distinguisher between $\Hyb^{3}_j$ and $\Hyb_{j+1}$. We construct from $\A$ a distinguisher $\D$ for the experiment from \expref{Lemma}{lma:reprogramming}: 
    
    \medskip \noindent \textbf{Phase 1:} $\D$ samples uniform $\widetilde{\Pi} \in \mathcal{E}(\mathcal{T},n)$ and $E \in \mathcal{E}(m,n)$. It then runs $\A$, answering its quantum queries using $E$ and its classical queries using $\widetilde{\Pi}$, until after it responds to $\A$'s $(j+1)^{\text{st}}$ classical query. If $\textsf{Bad}_{j+1}$, set $\textsf{flag}=1$; otherwise, set $\textsf{flag}=0$. Let $T_{j+1} = \{(\tau_1,x_1,y_1),\ldots,(\tau_{j+1},x_{j+1},y_{j+1})\}$ be the list of input/output pairs $\A$ received from its classical oracle thus far. $\D$ defines $F(a,K^*,x):= E^a_{k^*}(x)$ for $a \in \{1,-1\}$. It also define the following randomized algorithm $\B$: sample $K \leftarrow \{0,1\}^m \times \{0,1\}^{\kappa}$ and then compute the blinded set $B$ to be the input/output pairs that are reprogrammed so that $F^{(B)}(a,K^*,x)=\left(E^{T_{j+1},K}_{k^*} \right)^a(x)$ for all $a,K^*,x,$ i.e.,  $$B = \left\{((a,K^*,x),\left(E^{T_{j+1},K}_{k^*} \right)^a(x))\right\}.$$
     \medskip \noindent \textbf{Phase 2:} $\B$ is run to generate $B$ and $\D$ is given quantum access to an oracle $F_b$. $\D$ resumes running $\A$, answering its quantum queries using $F_b$. Phase $2$ ends when $\A$ makes its next (i.e., ($(j+2)^{\text{nd}}$) classical query.
     
     \medskip \noindent \textbf{Phase 3:} $\D$ is given the randomness used by $\B$ to generate $K$. It resumes running $\A$, answering its classical queries using $\lrw_K\left[E^{T_{j+1},K}\right]$ and its quantum queries using $E^{T_{j+1},K}$. Finally, if $\textsf{flag}=1$, $\D$ outputs $0$; otherwise, it outputs whatever $\A$ outputs.

     It is immediate that if $b=0$ and $\textsf{flag}=0$ (i.e., $\D$'s oracle in phase $2$ is $F_0=F$), then $\A$'s output is identically distributed to its output in $\Hyb_{j+1}$, whereas if $b=1$ (i.e., $\D$'s oracle in phase $2$ is $F_1=F^{(B)}$) and  and $\textsf{flag}=0$, then $\A$'s output is identically distributed to its output in $\Hyb^{3}_j$. It follows that $|\Pr[\A(\Hyb^{3}_j)=1] - \Pr[\A(\Hyb_{j+1})=1]|$ is equal to the distinguishing advantage of $\D$ in the reprogramming experiment of \expref{Lemma}{lma:reprogramming}. To bound this quantity, we bound the parameter $\varepsilon$ and the expected number of queries made by $\D$ in phase $2$ (when $F=F_0$.)

     The value of $\varepsilon$ can be bounded using the definition of $E^{T_{j+1},K}$ and the fact that $F^{(B)}(a,K^*,x) = \left(E^{T_{j+1},K}_{k^*} \right)^a(x)$. Fixing $E$ and $T_{j+1}$, the probability that any particular input $(a,K^*,x)$ is reprogrammed is at most the probability (over $K$) that it is in the set 
     \begin{align*}
        \left\{\begin{array}{c}
        \left(1, k, x_i \oplus h_{k'}(\tau_i)\right),\left(1, k, E_k^{-1}\left(y_i \oplus h_{k'}(\tau_i)\right)\right), \\
        \left(-1, k, E_k\left(x_i \oplus h_{k'}(\tau_i )\right)\right),\left(-1, k, y_i \oplus h_{k'}(\tau_i )\right)
        \end{array}\right\}_{i=1}^{j+1}.
     \end{align*}
     Since $h_{k'}(\tau_i)$ is uniform, taking a union bound gives $\varepsilon \leq 2(j+1)/2^{m+n}$.

     The expected number of queries made by $\D$ in phase $2$ when $F=F_0$ is equal to the expected number of queries made by $\A$ in its $(j+1)^{\text{st}}$ stage in $\Hyb_{j+1}$. Since $\Hyb_{j+1}$ and $\Hyb_{q_C}$ are identical until after the $(j+1)^{\text{st}}$ is complete, this is precisely $q_{Q,j+1}$. 

     Applying \expref{Lemma}{lma:reprogramming} thus gives
     \begin{align*}
        &|\Pr[\A(\Hyb^{3}_j)=1 \land \neg \textsf{Bad}_{j}] - \Pr[\A(\Hyb_{j+1})=1\land \neg\textsf{Bad}_{j+1}]| \\
        &=|\Pr[\A(\Hyb^{3}_j)=1 \land \neg \textsf{Bad}_{j+1}] - \Pr[\A(\Hyb_{j+1})=1\land \neg\textsf{Bad}_{j+1}]| \\
        &=|\Pr[\D=1 \land \textsf{flag}=0|b=0] - \Pr[\D=1 \land \textsf{flag}=0|b=1]| \\
        &\leq|\Pr[\D=1 |b=0] - \Pr[\D=1 |b=1]| \\
        &\leq 2 \cdot q_{Q,j+1} \sqrt{\frac{2 (j+1)}{2^{m+n}}},
    \end{align*}
    where the first equality arises because the $(j+1)^{\text{st}}$ query is answered identically by $\widetilde{\Pi}$ in both worlds. The first inequality holds due to the independence of $\D$'s output from $\A$'s output when $\textsf{flag}=1$, and $\Pr[\textsf{flag}=1]$ being the same for $b=0$ and $b=1$.
    \qed
\end{proof}

 \section{Quantum Oracle Constructions in the Proof of \expref{Theorem}{thm:security-ideal-cipher}}\label{sec:oracle-for-E-prime}

We describe how to implement the functions $E'$ and $E'^{-1}$ using a quantum oracle $f$, given the full descriptions of $P(\cdot)$ and $E(\cdot, \cdot)$. The definitions are:
\begin{align*}
E'(k, \tau) = 
\begin{cases} 
P(\tau) & \text{if } f(k) = 1, \\ 
E(k, \tau) & \text{if } f(k) \neq 1,
\end{cases}
\quad
E'^{-1}(k, \tau) = 
\begin{cases} 
P^{-1}(\tau) & \text{if } f(k) = 1, \\ 
E^{-1}(k, \tau) & \text{if } f(k) \neq 1.
\end{cases}
\end{align*}

Using quantum oracle access to $\Or_f$, we can construct $\Or_{E'}$ and $\Or_{E'^{-1}}$. Below is the construction of $\Or_{E'}$:
\begin{align*}
    & \ket{k}\ket{\tau}\ket{y}\\
    \xrightarrow[]{\text{add \textsf{aux} registers}}&\ket{k}_1\ket{\tau}_2\ket{y}_3\ket{0}_4\ket{0}_5\ket{0}_6\ket{0}_7\\
    \xrightarrow[]{\Or_{E,1,2,7} \Or_{P,1,6} X_5   \Or_{f,1,5}   \Or_{f,1,4}}&\ket{k}\ket{\tau}\ket{y}\ket{f(k)}\ket{\overline{f(k)}}\ket{P(x)}\ket{E(k,\tau)}\\
    \xrightarrow[]{\text{CCNOT}_{6,4,3}} &\ket{k}\ket{\tau}\ket{y \oplus (P(\tau)\cdot f(k)) }\ket{f(k)}\ket{\overline{f(k)}}\ket{P(x)}\ket{E(k,\tau)}\\
    \xrightarrow[]{\text{CCNOT}_{7,5,3}} &\ket{k}\ket{\tau}\ket{y \oplus (P(\tau)\cdot f(k)) \oplus (E(k,\tau)\cdot \overline{f(k)})}\ket{f(k)}\ket{\overline{f(k)}}\ket{P(x)}\ket{E(k,\tau)}\\
    \xrightarrow[]{\Or_{E,1,2,7} \Or_{P,1,6} X_5   \Or_{f,1,5}   \Or_{f,1,4}} &\ket{k}\ket{\tau}\ket{y \oplus (P(\tau)\cdot f(k)) \oplus (E(k,\tau)\cdot \overline{f(k)})}\ket{0}\ket{0}\ket{0}\ket{0}\\
    \xrightarrow[]{\text{drop \textsf{aux}}}&\ket{k}\ket{\tau}\ket{y\oplus (P(\tau)\cdot f(k))\oplus (E(k,\tau)\cdot \overline{f(k)})}
\end{align*}

The construction of $\Or_{E'^{-1}}$ follows the same process, substituting $P^{-1}$ and $E^{-1}$ in place of $P$ and $E$, respectively:
\begin{align*}  
\ket{k}\ket{\tau}\ket{y} \longrightarrow \ket{y \oplus (P^{-1}(\tau) \cdot f(k)) \oplus (E^{-1}(k,\tau) \cdot \overline{f(k)})}.
\end{align*}

\section{Post-Quantum Lower Bounds for Modes}\label{app:modes-table}

A summary of applications of the general theorem for block cipher mode post-quantum security is given in \expref{Table}{tab:general-thm-app}.

\begin{table}
    \centering
    \bgroup
\def\arraystretch{1.8}
\begin{tabular}{ |p{4cm}||p{8cm}|  }
 \hline
 Block Cipher Modes& Post-Quantum Security Bound\\
 \hline
 \cbc   & $\frac{q_Q^2\ell^2}{2^m}+\frac{q_C^2\ell^2}{2^n}$\\
 \ecbc  & $\frac{2q_Q^2\ell^2}{2^m}+\frac{4q_C^2\ell^2}{2^n}$\\
 \cmac  & $\frac{(q_Q\ell+1)^2}{2^m}+\frac{5(\ell^2+1)q^2}{2^n}$ \\
 \gcm   & $\frac{q_Q^2\ell^2}{2^m}  + \frac{(\sigma+q_C+q_C'+1)^2}{2^{n+1}}+ \frac{\sigma+q_C+q_C'}{2^{n-1}} + \frac{q_C'(\ell+1)}{2^s}$ \\
 \textsf{GCM-SST}&   $\frac{q_Q^2\ell^2}{2^m} +\frac{(\sigma+3(q_C+q_C'))^2}{2^{n+1}}+ \frac{q_C'\ell}{2^{n}} + \frac{q_C'}{2^s}$  \\
 \hline
 \textsf{\lrw} &   $\frac{q_Q^2}{2^m}+\frac{q_C^2}{2^n}$  \\
 $\xext$ &  $\frac{q_Q^2}{2^m}+\frac{3q_C^2}{2^n}$  \\
 \hline
\end{tabular}
\egroup
\vspace{2mm}
    \caption{Post-quantum security bound of block cipher modes in the QICM. $\ell$ is the block length for every query. Detailed explanations of the corresponding parameters can be found in \expref{Section}{sec:general-thm-app}.} \label{tab:general-thm-app}
\end{table}

\end{document}